\documentclass[12pt]{amsart}
\usepackage{amssymb}
\usepackage{amsfonts}
\usepackage{amssymb,latexsym}
\usepackage{enumerate}
\usepackage{mathrsfs}

\usepackage{graphicx}
\usepackage[all]{xy}

\newtheorem{thm}{Theorem}[section]

\newtheorem{cor}[thm]{Corollary}
\newtheorem{lem}[thm]{Lemma}
\newtheorem{prop}[thm]{Proposition}

\theoremstyle{definition}
\newtheorem{definition}[thm]{Definition}

\theoremstyle{remark}

\begin{document}

\title[{\normalsize {\large {\normalsize {\Large {\LARGE }}}}} P-class is a proper subclass of  NP-class; and more
]{ P-class is a proper subclass of  NP-class; and more  }

\author{ Kim JongJin, Kim GwangJin, Lee JongPyo, Wang ShuanHong, Nam Ki-Bong, Seo GyungSig, Kim InSu,    Kim YangGon* }

\address{emeritus professor,
	 Department of Mathematics,  
	Jeonbuk National University, 567 Baekje-daero, Deokjin-gu, Jeonju-si,
	Jeollabuk-do, 54896, Republic of Korea.}

\

\email{jjkim@jbnu.ac.kr }

\address{professor emeritus, Department of Mathematics, Woosuk university,Jeonju city},
\email{kimkj@woosuk.ac.kr}

\address {Professor emeritus,Department of Mathematics,Jeonbuk National university,Republic of Korea.}

\email {jplee@jbnu.ac.kr}

\address{Department of mathematics,Southeast university,China.}

\email{shuanhwang@seu.edu.cn}

\address{Department of mathematics and computer science,university of Wisconsin at Whitewater,USA.}

\email{namk@uww.edu}

\address{emeritus professor,Department of Mathematics,Jeonbuk National University,Republic of Korea.}

\email{sgs0303@jeonbuk.ac.kr}
\address{emeritus professor,Department of Mathematics,Jeonbuk National University,Republic of Korea.}

\email{Insu@jeonbuk.ac.kr}

\address{Professor emeritus,Department of Mathematics,Jeonbuk National university,Republic of Korea.}
\email{kyk1.chonbuk@daum.net}
\email{kyk2jeonbuk@naver.com}

\

\subjclass[2020]{Primary17B10,17B50,Secondary68Q15,68Q17.
                     *corresponding author      }

\begin{abstract}
We  may  give rise to  some questions related to the  mathematical structures of  $P$-class and $NP$-class. We have seen that one is a  proper subclass of the other. Here we disclose more  that  $P$- class  turns out to be the proper distributive sublattice of the  $NP$- class.

\end{abstract}

\maketitle

\section{introduction and  motivation }

\

\

\large{

We announce here that  representations  of modular  Lie algebras  are  important  as well as  the  well known $P$ versus  $NP$.\newline

Because choosing remarkable bases  in our certain algebras  seems to be  important  in arguing about  $P$-problems, $NP$-problems,\newline
$NP$- hard problems, $NP$- complete problems,etc., \newline
we  prefer to  be absorbed  in modular representation theory of Lie algebras among other things.\newline

Even though solving the  Kim's conjecture directly  is itself an  important problem, we think however that  such a way could be more important as using Kim's conjecture to solve a problem relating to $NP$- completeness.\newline

We already gave an $NP$- problem coming from a  sort of counting problem  in  the references [NWK-1] and [NWK-2] by making use of  a choice function and a modular representation of  Lie algebras.\newline

We proved in  the  reference [KLWNKK]  that  $P$-class is  a proper subclass of  $NP$-class and conjectured  that  the counting problem under consideration could be an $NP$-complete problem, i.e.,

 such a problem called a certain counting problem could become an $NP$- hard problem, and consequently 

an $NP$- complete problem, which shall be  shown  again  in the section5 of this paper after recapitulating the proof in  the reference [NWK-1,2].\newline

Finally in section 6, we ask two questions about the  mathematical structures of these classes and  answer to these questions in self-contained manner.\newline

 Before we construct such a conjecture and questions, preliminaries and  rigorous definitions under consideration shall be given in advance.\newline

Section2  explains some languages regarding computer science. \newline

Afterwards sections 3  and 4 remark  on the modular Lie algebras and their representations because  the bases of  certain algebras associated with irreducible  modules over classical modular Lie algebras are very  crucial  and important  to  construct  our counting problem mentioned above.\newline

We shall use conventions and notations in [NWK-1,2] unless otherwise stated.\newline

\section{prerequisite language of computer science }

We should be ready in this section for the argument relating to computer science  in the later sections .

So we exhibit some definitions needed for our progression ahead.

Small $o$ and big $\mathcal {O}$ are often used in computer science.

\begin{definition}

In complexity theory we are usually concerned with functions from $\mathbf {Z}^{\geq 0}$ to itself.

 So the letter $x$ shall denote usually   the standard variable of  these functions such as  $e^x, ax^2+bx+ xc, a^x +b^x, sin x, ln x,$ etc.\newline

In fact $\sqrt {x}$ or $ln x$ or $(ln x)^2- 7 \sqrt {x}$ may not  produce nonnegative integers; nevertheless we shall usually think of them as such integers, and so $f(x)$ for a given function $f$ really means the greatest integer $\leq f(x)$. Hence we use the notation $[f(x)]$ to denote such an integer.\newline

Suppose that $f_1$ and $f_2$ are functions  from  $\mathbf {Z}^{\geq 0}$ to itself. If we have an integer $R$ such that $f_1(x)\leq R\cdot f_2 (x)$ $\forall x \geq x_0$, then  in this case we say that $f_1(x)$ is $\textit {big oh}$ of $f_2(x)$; whereas  if the ratio $f_1(x)/ f_2(x) \rightarrow 0$ as $x$ tends to  $\infty$, we  say that  $f_1(x)$ is  $\textit{small oh}$ of $f_2(x)$.\newline

Informally speaking the notation $f_1(x)= \mathcal {O}( f_2(x))$ gives the meaning that $f_1$ grows slower than $f_2$. Hence $f_2(x)= \mathcal O(f_1(x))$ means the opposite situation, in which case the notation $f_1(x)= \Omega(f_2(x))$ is also used.\newline

If both situations happen, then we simply write $f_1(x)= \Theta (f_2(x))$ instead of  $f_1(x)= \Omega (f_2(x))$ and $f_1(x)= \mathcal O(f_2(x))$. Clearly the same rate of growth arises in this situation. \newline

$\mathbf {[example]}$\newline
Suppose that  $f(x)$ is a  polynomial  of degree $n$; then $ f(x)= \Theta (x^n)$ may be used in this case meaning that  the polynomial's first nonzero term captures the rate of growth of $f(x)$.\newline

Next if we suppose that $R$ is an integer $\geq 2$ and $f(x)$ is any polynomial, then it is known that $f(x)\leq  \mathcal {O}(R^x)$, which probably becomes the important and most useful fact about  growths of functions  in complexity theory.\newline

 So we understand that the expression $f(x)= \mathcal {O} ( R^x)$ means informally that  some  exponential  function   with the base number $R \geq 2$   grows precisely faster after all than any exponential function , whose fact can be easily shown by dint of L'hospital formula.\newline

 It follows in the same context as this that $ln (x)= \mathcal {O} (x)$ and moreover $(ln x)^p= \mathcal {O} (x)$ for any integer power $p\geq 1.$\newline

What is more important is what we are dealing with  as the source  of our consideration in connection with the required time. So we may say for example that it suffices to
say that we are contented with  the rate of time growth $\mathcal O(x^2)$ in this paper.\newline

 Hence as a token of the problem having been solved satisfiably, such  polynomial rates of growth shall be regarded as acceptable time requirements. Evidently a cause of concern in contrast with polynomial rates arises because of exponential rates or factorial rates such as $x!,a^x,$etc.\newline

If we encounter with a problem having such persistent rates of growth and if algorithm after algorithm we have invented does not satisfy the problem within a polynomial time, then we are liable to say that  there does not seem to be a practical useful solution which is amenable or tractable to the problem in hand.\newline

So it is undoubtedly uncontroversial that we consider the dichotomy between two kinds of  time bounds such as polynomial time bounds and nonpolynomial time bounds and that  we identify the  polynomial time algorithms with the  intuitive concept of practical plausible computation matter.\newline

$\mathbf{[remark]}$\newline
Some efficient computations in practice  don't  happen as  polynomial time ones and some computations don't  seem to  be useful. \newline

Hence  linear programming is just an important subject that deals with the problems which give us examples  for either kind of  exception and polynomial time algorithms. For more references to linear programming and its simplex method, see 9.5.34  in [Pa]. \newline

We see that the simplex method is just considered as a widely used classical algorithm for this basic problem even though in the worst case the simplex method turns out to be exponential.\newline

 In practice the performance of the simplex method is to be done consistently superbly within a polynomial time, but the ellipsoid algorithm has impractically slow performance in contrast with the simplex method.\newline

 Generally speaking it seems that the story of linear programming stands  in favor of  the methodology of complexity theory in spite of  the fact that a subtle discussion  arises  as a little obstruction to the methodology;\newline

 in other words linear programming looks like an indication  showing that  polynomial time solvable problems with practical algorithms usually appear  although the empirically good algorithm may not be coincident with the polynomial time algorithm.\newline

For example $\mathcal O(x^{100})$ or $\mathcal O(x^{121})$ algorithm would produce limited practical values. But in contrast  $\mathcal O (7^{\frac{x}{89}})$ or $\mathcal O(x^{2ln x})$ algorithm could be efficient and useful empirically.\newline

 Well  at any rate the polynomial  time paradigm has some sustainable strong supports. Note that we usually use exponential functions  rather than polynomial functions from the view point of the fact that  the formers prevail the latters after all except for a finite set  of instances  of the problem and the former expressions may be simple.\newline

 However in practice such exceptional instances are those which are liable to arise within the universe confining us and moreover $\mathcal O(x^{1000})$ algorithms or $ \mathcal O(x^{\frac{x}{89}})$ algorithms seldom show up.\newline

 In fact   it seems that  exponential growth algorithms turn out to be impractical, although  polynomial time growth algorithms usually possess reasonable constants $R$  and small powers as  in the above example.\newline

Against another criticism of our polynomial time paradigm that the paradigm checks only how the polynomial time algorithm in the least favorable instances behaves roughly, we insist that  on the average the polynomial time  algorithm may be satisfiable and the most inconvenient  algorithms may occur because of a  statistically insignificant fraction of the  input sources.\newline

If we analyze the expected performance  of an algorithm as opposed to the worst case, then informative results usually happen. Nevertheless  the input distribution  of a problem is hard to perceive,i.e., we hardly know  each possible instance's probability with which  the instance arises  as an input  to our feasible  algorithm. Hence unfortunately we cannot implement  a truly informative average case analysis.\newline

So it becomes of little consolation or help to recognize that we have stumbled upon a statiscally insignificant exception if we are eager to know just one special instance abysmally. \newline

We may also study algorithms from a  probabilistic point of view, say as in chapter 11[Pa]. So in such a situation we must  clarify the sourse of  randomness of inputs  which exists within the algorithm itself  rather than input's lottery. The readers may also  refer to the reference mentioned above  for an interesting complexity theoretic treatment of average  case performance. \newline

Criticism from everywhere may bother our choice of polynomial time algorithms as the  mathematical concept which is expressed  to capture  the informal  notion  of  practically  efficient  computation.\newline

 We note that mathematicians in any field of mathematics intend to capture  an intuitive  real life concept such as smooth function by a mathematical one  such as $C^{\infty}$ excluding undesirable elements. \newline

Notwithstanding considering the structure  arising from this construction may still allow  the notion to include certain undesirable specimens. After all a useful and elegant theory may come out  from adopting polynomial  worst case  performance as our criterion of efficiency. Such an attempt usually could not happen without such a beautiful theory  saying about practical meaningful computation.\newline

 As a  matter of fact, polynomials usually have much mathematical handiness. For example they are much stable under various useful operations. Moreover  a  constant $\Theta (ln x)$  relates  the  logarithms of  such polynomial functions. Besides such  a constant  could result in  a convenient one in several opportunities such as  space bounds  argument.\newline

For a while  let's think of the  notion of space bounds or  space requirements paying attention to the reachability problem and  its search algorithm explained below.\newline

Assme that $E$ is a set of edges and $V$ is a finite set of vertexes. By a  $\textit{graph}$  we mean  $G= (V, E)$ consisting  of pairs of vertexes which are connected and directed .  A famous so-called reachability problem is the following: If we are given two vertexes $n_1,n_k\in V$ of a graph $G$, then we ask 'may we go  from $n_1$ to $n_k$?'\newline

We contend that  the reachability problem may be solved by the so-called search algorithm  described as follows:\newline

At first we use a certain set of vertexes. We denote it by $V'$ which may vary depending on $n_1$ and $n_k$. In the first place we start with  $V'= \{n_1\}$.
We can mark or unmark  each vertex in the sense that  we may mark the vertex if and only if  $n_i$ is in $V'$ at present or $n_i$ has belonged to  $V'$ at some point in the past. We may mark $n_1$ in the first step. \newline

According to every step iterated by this search algorithm , a vertex $n_i$ is selected and  is got rid of  from $V'$. Such a process should be repeated  for every edge $n_i\rightarrow n_j$ out of $n_i$. If we continue marking, then $V'$ becomes empty at last. \newline

Now we may  answer the  reachability problem saying no if $n_k$ is not marked and saying yes if  $n_k$ is marked. Evidently the markings of the vertexes  and the maintaining of the designated set $V'$ are  basically required in this algorithm. \newline

Here we may say  that  the search algorithm  requires  $\mathcal O(x)$ space in the sense  that  we can have at most $x$ markings and $V'$ cannot  also become bigger than $x$.

\end{definition}

Next we remind ourselves of a  formal and systematic model, the so called Turing machine for representing arbitrary algorithms. In the next paragraph we define it precisely.\newline

But roughly speaking, any algorithm  can be efficiently simulated by the Turing machine although it has  a weak appearance and looks clumsy.
For this reason we look at  the flow chart  relating any algorithm to a relevant Turing machine and vice versa. Hereafter we shall use some programming languages including Turing machine.\newline

The Turing machine has a string of symbols rather primitively as a single data structure. We permit the program  to move a cursor right and left  on the string of symbols, to write on the present position, and to branch depending on the value of  the present symbol if we use the available operations.\newline

We present here the exact  definition of this rather primitive language together with others.\newline

We say in principle  that  a mathematical problem belongs to  the $P$- class if it is solvable in polynomial time by an ordinary (deterministic) Turing machine, while  a  mathematical problem belongs to the $NP$- class  if  it is solvable by a nondeterministic Turing machine whose time complexity is  bounded by  a polynomial function of input size.\newline

In the following paragraph we shall  give  definitions of  the Turing machine, $P$-class, and $NP$- class more or less comprehensibly and precisely.

\begin{definition}
The Turing machine should be in one of  a finite number of internal states $q_0,q_1,\cdots ,q_m$ which are possible contents of  a memory device  and are fixed at first  for this machine, being supplied with a tape divided into squares.\newline

Such a tape must be infinite in both directions. Each square must be blank or any of the previously specified symbols \newline

$s_0,s_1,\cdots ,s_n$ must be printed on each square with $s_0$ being blank. The number of nonblank squares is of course finite in any situation.\newline

Suppose that we got a mapping $\mu$: $Q\times M \rightarrow Q\times M\times \{R,L,0\}$ , where $Q= \{q_0,q_1,\cdots ,q_m \}$ and $M=  \{s_0.s_1,\cdots, s_n\}$. A Turing machine $\tau$ is  specified  by the five tuple $\{Q,M,\mu,q_0,F\}$, where $F$ is a subset of  $Q$ called a set of final states  and $q_0$ indicates  the specified initial state. $M$ is usually called the $\textit {alphabet}$  of $\tau$. \newline

 If  $\mu (q,s)= (q',s',x) $, then we call  the five tuple $(q,s,q',s',x)$ an $\textit {instruction}$  for $\mu$. The following actions  must be satisfied  by $\tau $.

(i)In the beginning, $\tau$ stands  in  the state $q_0$. As an  intput  some sequence of symbols  in $M$ are written on the tape, and  the machine stands on the square next to the leftmost nonblank symbol.\newline

(ii)The machine chooses an instruction $(q,s,q',s',x)$ with respect to the present state $q$  and the scanned symbol $s$. Afterwards the next state is  rendered to $q'$ , $s'$ replaces $s$, and $\tau$ moves one square left or right or stands still according as $x$ becomes $L,R$, or $0$.\newline

(iii)Whenever the state  becomes  an element of $F$, $\tau$ stops moving. So the sequence  of symbols  on the tape is the result  of the computation as the output.\newline

Here in case there always exists  at most  one  instruction associated with  a given pair, we shall call $\tau$ $\textit {deterministic}$; otherwise we shall call it  $\textit {nondeterministic}$. \newline

Of course we may have other definitions of  Turing machine. \newline

We are already aware that an algorithm exists for a problem if and only if a deterministic  Turing  machine exists for  a problem.\newline

It is evident  that the containment $P\subseteq NP$ holds. In 1971 Stephen A.Cook(1939-) asked the problem about the reverse containment. We call such  a problem $P$ versus $NP$.\newline

$\textbf{[example]}$
We let the Turing machine $\tau$  defined  as follows: $\tau= \{Q,M,\mu, q_0, F \}, Q=\{q_0, k, k_0, k_1 \}, M= \{\bar {0}, 1,s_0, s_1 \},F= \{q_0 \}$, and 
$\mu $ is designated  as in the  table1  having 3 columns and 17 rows  in  the next page. \newline

How is this program going to begin and operate ? Of course  the state is initially $q_0$. After a finite long string, we may initialize this string to the first leftmost nonblank symbol $s_1$.\newline

Such a string is said to be the input of the Turing machine in general. In our situation of  the above example, we can take $\{ \overline{0},1,$ $s_0 \}$ as an input of this Turing machine.\newline

It is notable that  the output  could be  the same  as  the input  in this example.

\begin{table}
\centering
\caption{example of  $\tau$}
\label{t1}
\begin{tabular}{|c|c|c|c|}
\noalign{\smallskip}\noalign{\smallskip}\hline
& column1 & column2 & column3 \\
\hline
row1& $ q\in Q$ & $s\in M$ & $\mu (q,s)$ \\
\hline
row2 &$q_0 $ & $\overline {0}$ & $(q_0, \overline {0}, R)$\\
\hline
row3  &$q_0$ & $1$ &$(q_0,1,R)$ \\
\hline
row4 &$q_0$ &$s_0 $ &$(k,s_0, L)$ \\
\hline
row5 &$q_0$ &$s_1$ &$(k,s_1, R)$ \\
\hline
row6 &$k$ &$\overline {0}$&$(k_0,s_0, R)$\\
\hline
row7&$k$ &$1$&$(k_1,s_0,R)$\\
\hline
row8 &$k$ &$s_0$ &$(k,s_0, \overline{0})$\\
\hline
row9&$k$&$s_1$&$(k,s_1,\overline {0})$\\
\hline
row10&$k_0$&$\overline{0}$&$(q_0,\overline{0}, L)$\\
\hline
row11&$k_0$&$1$&$(q_0,\overline{0},L)$\\
\hline
row12&$k_0$&$s_0$&$(q_0,\overline{0},L)$\\
\hline
row13&$k_0$&$s_1$&$(k_0,s_1,R)$\\
\hline
row14&$k_1$&$\overline{0}$&$(q_0,1,L)$\\
\hline
row15&$k_1$&$1$&$(q_0,1,L)$\\
\hline
row16&$k_1$&$s_0$&$(q_0,1,L)$\\
\hline
row17&$k_1$&$s_1$&$(k_1,s_1,R)$\\
\hline
\end{tabular}
\end{table}

\end{definition}

\begin{definition}($P$-problem,$NP$-problem)

By $NP$- problems we mean  a class of problems solutions of which are hard to find but easy to verify and which can be solved by Non-Deterministic Turing machine in polynomial time.\newline

By $P$-problems we mean a class of problems  solutions of which are hard to find but which have an algorithm,i.e., which can be solved by Deterministic Turing machine in polynomial time.
\end{definition}

\begin{definition}($NP$-hard, $NP$- complete)
We say a problem $M$ is $NP$- $\textit {complete}$ in case $M$-  is $NP$- hard and  in case we may find an $NP$- problem $M'$ which  is reducible to $M$ in polynomial time. \newline

We say a problem $M$ is  $NP$- $\textit {hard}$  in case all  problems in $NP$- class are reducible  to $M$  in polynomial time.
\end{definition}

 \section{modular Lie algebras}

Next  in order to attain our  purpose we have to connect  the representation theory of modular Lie algebras with NP problems.\newline

We begin with some definitions first prior to some preliminary results related to modular Lie algebra theory.\newline
Let $F$ be a field of nonzero characteristic $p$. 

\begin{definition}
Let $L$ be a  Lie algebra over $F$. We shall call a mapping $[p]: L\rightarrow L $ via $a\rightarrow a^{[p]}$ a $p-mapping$ in case that  the following three conditions are satisfied\newline

$(i)\forall a\in L, ad (a^{[p]}) = (ad a)^p,$\newline

$(ii)  \forall c\in F,\forall a\in L, (ca)^{[p]}= c^pa^{[p]},$\newline

$(iii) \forall a,b\in L, (a+ b)= a^{[p]}+ b^{[p]}+ \sum_ {i=1}^{p-1} s_i(a,b)  $,\newline

where $(ad(a\otimes \mathbb{X}+ b\otimes1))^{p-1}(a\otimes1)= \sum_{i=1}^{p-1}is_i(a,b)\otimes \mathbb {X}^{i-1}= 
$   in $L\otimes_FF[\mathbb{X}]$.                                                       
We call such a pair $(L, [p])$ a $\textit {restricted}$ Lie algebra over $F$.

\end{definition}

\begin{definition}
Let $(L,[p])$ be a restricted Lie algebra over $F$ and $\chi$  be  a linear form of $L$, i.e., let  $\chi$  be an element of $L^{\ast}$.\newline

A representation of $L $, $\rho_{\chi}: L \rightarrow gl(V)$ is called a $\chi$-representation in case that  $ x\in L,$  

$\rho_{\chi} (x)^p- \rho_{\chi} (x^{[p]})= \chi(x)^p id_V$.

We call $\chi$ the $p$-character of the representation or of the corresponding module.\newline

If $\chi \neq 0$ in particular, then $\rho_{\chi}$ is called  a $\textit {nonrestricted representation}$. Otherwise $\rho_0$ is called $\textit {p-representation}$ or $\textit {restricted representation}$.

\end{definition}

 Let $L$ be a finite dimensional indecomposable  restricted Lie algebra over $F$.

Suppose that $\rho$ is the corresponding representation to  an $L$- module $V$.

\begin{prop}
We let  $(L, [p])$ be a restricted Lie algebra  defined as above and $V$  a finite dimensional $L$-module over an algebraically closed field $F$ of nonzero characteristic $p$. \newline

We then have some submodules $V_1,V_2,\cdots, V_s$  and  some characters $\chi_1,\chi_2,\cdots, \chi_s$  satisfying that $V= \sum_{i=1}^{s} V_i$ and $\{\rho{(x)}^p- \rho(x^{[p]})- \chi_i(x)^p id\}\vert_{V_i}$ becomes nilpotent $\forall x\in L$.\newline

 If $V$ is irreducible in particular, then there exists $\chi\in L^{\ast}$ such that $\rho(x)^p- \rho(x^{[p]})= \chi(x)^p id_V.$

\end{prop}
\begin{proof}
 As we are well aware, we have  that $ x^p- x^{[p]}\in \frak Z(u(L)), \newline
\forall x\in L$, which is just  the center of  the universal enveloping algebra of $L$. \newline

With respect to the finite dimensional  abelian Lie algebra $C:= <\{\rho(x)^p- \rho (x^{[p]})\vert x\in L\}>_F$ which is the $F$- algebra generated by 
$\rho(x)^p- \rho(x^{[p]})  \forall x\in L,$we can decompose $V$ into weight spaces. \newline

We thus have $V= \oplus_i  V_{\phi_i}$ with $\phi_i\in Map(C, F[\mathbb{X}])$, where  $V_{\phi_i}= \{v\in V \vert \forall c\in C, \exists n(c,v)\in \mathbb {N}  s.t. (\phi_i (c)\rho(c))^ {n(c,v)}(v)= 0 \}. $\newline

Because $F$ is algebraically closed , it follows  that  $\phi_i\in Map(C,F)$ actually.
\newline
Obviously $V_{\phi_i}$ are submodules since $[C, \rho(L)]= 0$ is trivial. Putting $\chi_i(x):=\phi_i (\rho(x)^p- \rho(x^{[p]})) ^{\frac{1}{p}} $, we perceive 
that  $\{\rho{(x)}^p- \rho(x^{[p]})- \chi_i(x)^p id)\} \vert V_{\phi_i}$ is  nilpotent for each $x\in L$.
Because of abelian $C$, we get $\phi_i \in C^{\ast}$ and because  the map $f: x\longmapsto \rho(x)^p- \rho(x^{[p]})$  is  semilinear,i.e.,\newline

 $f(\alpha x+ y)= \alpha^p f(x)+ f(y)  \forall x,y \in L, \forall \alpha \in F $,  we see that $\chi_i$ becomes linear. Due do Schur's lemma, the latter claim is clear.

\end{proof}

\begin{definition}
Let $F$ be an algebraically closed field of  nonzero characteristic $p$. By the  $\textit {clasical type Lie algebras}$  over $F$  in the narrow sense, we mean  analogues over $F$ of the  $A_l,B_l,C_l,D_l$  types of simple Lie algebras  over the complex number field $\mathbb {C}$, whereas  in the wider sense  we mean  the so called  9-types  of simple Lie algebras over the complex number  field $\mathbb{C}$ as follows:\newline

 $A_l(l\geq 1), B_l(l\geq 2),C_l(l\geq 3), D_l(l\geq 4),G_2, F_4, E_6,E_7,$ and $E_8$. \newline

In other words any Lie algebra of classical type over $F$ is isomorphic to the Lie algebra $\sum_{j=1}^{n} \mathbb {Z} e_j\otimes_{\mathbb {Z}} F$ for some Chevalley basis $\{e_1,\cdots, e_n\}$ of  the type above.\newline

Next we would like to supplement Zassenhauss' thory related to  the background work on modular representation theory.\newline

 We assume for a while that $F$ is an algebraically closed field of nonzero characteristic $p$  unless  it is stated otherwise.

\end{definition}

\begin{definition}   
We let $\mathcal {O}(L)$ denote  the $p$- $\textit {center}$ for a fixed restricted Lie algebra $L$ of dimension $n$ over $F$, i.e., $\mathcal {O}(L):= <\{x^p-x^{[p]} \vert x\in L\}\cup \frak Z$$(L)>_ F$ is  the subalgebra of $u(L)$ generated by $\frak Z (L)   $and  $1$ and $\{x^p- x^{[p]}\vert x\in L\}$, where $\frak Z (L)$ indicates the center of $L$  and $u(L)$  denotes the  universal enveloping algebra of $L$. \newline

We denote the center of $u(L)$ by $\frak Z (u(L))$ and simply by $\frak Z$ if there is no confusion. We are well  aware that $\frak Z$ becomes an  affine algebra as an integral domain. \newline

We thus may apply the commutative algebra theory  or the elementary algebraic geometry theory to this  structure $\frak Z$. There is an affine variety  associated with  this affine algebra. We shall call  this variety the $\textit {Zassenhaus variety}$. 
\end{definition}

\begin{prop}
We have $\frak Z:= \frak Z (u(L))= \mathcal O (L)[s_1,\cdots, s_k]$, where $s_i$'s are integral over $\mathcal O (L)$  as elements of $u(L)$.

\end{prop}
\begin{proof}
Easy to know, but  the readers may refer to  theorem 1 in [ZH] or chapter 6 in [SF].

\end{proof}

\begin{prop}
If the mappimg $h: \mathcal O (L) [\mathbb {X}_1,\cdots, \mathbb {X}_k]\rightarrow \mathcal O (L)[s_1,\cdots, s_k]$  is an evaluation algebra homomorphism,
then we have the following:

(i)Ker h just becomes a prime ideal of $\mathcal O (L)[\mathbb {X}_1,\cdots, \mathbb {X}_k]$.\newline

(ii)Any point of  $\nu (Ker h)$ in $F^{n+k}$ beomes a normal point  of $\nu (Ker h)= \nu ( \sqrt {Ker h}) \neq \phi$, where $\sqrt {Ker h } $ indicates the nilradical of $Ker h$ and $\nu$ indicates the affine variety cofunctor.

\end{prop}
\begin{proof}
  (i)Since $\frak Z$ becomes an integal domain, the identical relations $\frak Z:= \frak Z (u(L))=\mathcal O (L)[s_1,\cdots, s_k]$\newline

$\cong_{alg} \mathcal O (L)[\mathbb {X}_1,\cdots, \mathbb {X}_k]/Ker h$ prove  our assertion.\newline

(ii)We get $\mathcal I (\nu (Ker h))= \sqrt {ker h}= Ker h$ by virtue of (i). \newline

We thus have $\frak Z (u(L))= \mathcal O (L)[\mathbb {X}_1,\cdots, \mathbb {X}_k]/\mathcal I (\nu (Ker h))  $. Because $\frak Z$ is integrally closed, each local ring $R_x$ also becomes integrally closed.\newline

 Hence every point of $\nu (Ker h)$ becomes a  normal point,i.e., $R_x$ becomes a normal ring.

\end{proof}
\begin{prop}
Suppose that $V$ is an irreducible  algebraic variety and  $x$ is a  simple point of $V$; then  we have:\newline

(i)$R_x$ becomes  a regular (Noetherian) local ring and so $x$ is a normal point,i.e., $R_x$ is normal.\newline
(ii)But the converse is not true.

\end{prop}
\begin{proof}
(i) is clear from chapter 1 , Linear algebraic groups  written by James E. Humphreys.
(ii) holds by the above proposition considering the figure of  node of the polynomial  $(x-c)(x-2)^2- y^2= 0$ for $c\leq 2$  in $\mathbb R^2 \subset \mathbb C^2$.

\end{proof}

\begin{prop}
We obtain the following under the same notation as above:

(i)$\nu (Ker h)$ becomes a normal irreducible algebraic variety.\newline

(ii)Suppose that $Ker h$ is principal with $Tan X_x\subsetneqq F^{n+1}, \forall x\in X:= \nu( Ker h)$; then  we claim that $X$ becomes a smooth irreducible algebraic variety, in other words all points of $X$ are simple and so normal, where the geometric tangent space of $X$ at $x$ is denoted by $Tan X_x$.
\end{prop}

\begin{proof}
(i) is also evident from  propoitions(3.7),(3.8) above.\newline

(ii)Since we know that $X$ is irreducible, we have only to prove the latter asserion. We have that  each point $x\in X$ becomes a simple point  because $ dim_F Tan X_x= n$ from the fact  $Tan X_x= \nu ( \sum _{i=1}^{n+1} \frac {\partial f}{\partial {\mathbb X_i}} (x)(\mathbb X- x_i))$ with $Ker h= (f)$ for some irreducible polynomial  $f\in Ker h$  and  from the fact $ dim \mathbb X= n$ according to proposition 3.6. Normal point is already  proven in the preceding propositions.

\end{proof}

\begin{definition}
Let $\chi \in L^{\ast}$ and let $I_{\chi}$ denote the ideal of $u(L)$ generated by  $(l^p- l^{[p]}- \chi (l)^p)\in \frak Z (u(L)), \forall l\in L$. We shall call $u_{\chi}(L):= u(L)/I_{\chi}$ the $\textit { reduced  enveloping algebra}$ of $L$ associated with the $p$-character $\chi$. If $\chi= 0$ in particular, then $u_{\chi} (L)$ is called  the $\textit{restricted enveloping algebra}$ of $L$.
\end{definition}

\begin{definition}
C.Chevalley(1909-1984) found that  the classification of simple algebraic groups over $F$ reduces essentially to that of  simple Lie algebras  over $\mathbb C$ and that  the nine types of  simple Lie algebras  over $\mathbb C$ are induced from these nine types of simple Lie groups over $\mathbb C$, which are denoted by the same symbols as in definition 3.4. \newline

We usually classify the simple algebraic groups over $F$ into infinite family of four classical groups of types $A_l,B_l,C_l,D_l$ along with five exceptional types $E_6,E_7,E_8, F_4,$  and $ G_2$.\newline

We  shall call  an algeraic group over $F$  $\textit {simply connected}$ when this algebraic group has  its Lie  algebra with the basis consisting essentially of a  Chevalley basis. Suppose that $G$ denotes a simply connected algebraic group over $F$ with its Lie algebra  $\frak g:= L(G)$.\newline

 It is known that $\frak g$ has a triangular decomposition like $\frak g= \frak n^- \oplus h \oplus \frak n^+$ as  in characteristic zero. Here $\frak h$ indicates a  $\textit {Cartan subalgebra}$  of $\frak g$ which is defined to be the Lie algebra of  a maximal torus  of  $G$, whereas  $\frak n^-$ and  $\frak n^+$ indicates the respective sums of  negative and positive root spaces. \newline

We put $\frak B= \frak h \oplus \frak n^+$ and call it a  $\textit {Borel subagebra}$ of $\frak g$.
It is well-known that $\frak B= L(B)$  for some $\textit {Borel subgroup}$  $B$ of $G$, which is  defined to be  a connected maximal solvable subgroup of $G$.

Suppose that $G$ is an algebraic group over $F$ with its Lie algebra  $\frak g = L(G)$. $G$ has  a unique largest connected normal solvable subgroup automatically closed, which  is  just the  identity component  denoted by $R(G)$ and called  the $\textit {radical}$ of $G$. \newline

We shall denote  by $R_u(G)$  the normal subgroup of $R(G)$ consisting of all its unipotent elements and call it  the $\textit { unipotent radical}$ of $G$.\newline

If  a connected algebraic group $G$ of  positive dimension has no closed connected normal subgroup except the identity, i.e.,$R(G)$= identity, then we call 
the group $G$ $\textit {semisimple}$. \newline

A  semisimple algebraic group $G$ has a Borel subgroup of the form $B= T \cdot R_u (B)$ with $T$ a fixed maximal torus. It is known that a  Borel subgroup $B$ containing $T$ has  another Borel subgroup $B^-$ called opposite $B$ relative to $T$ satisfying $B\cup B^-= T$. So we get $B^-= T\cdot R_u {B^-}$, which is known to be unique.

\end{definition}

Suppose that we are given a simple and simply connected algebraic group $G$. Putting $\frak B:= L((T\cdot R_u)B), \frak U:= L(R_u(B)), \frak h:= L(T)$, and  $\frak U^-:= L(R_u(B^-))  $, we let $\overline {\frak B^{\ast}}, \overline {\frak h^{\ast}} $ and $\overline {\frak U^{\ast}}$ denote $\{\bar {\alpha}\in \frak g^{\ast} \vert \bar {\alpha} (\frak U)= 0 \}, \{\bar {\alpha} \in \frak g^{\ast} \vert \alpha (\frak U\oplus \frak U^-)= 0\}$ and $\{\bar {\alpha}\in \frak g^{\ast} \vert \bar {\alpha}( \frak h\oplus \frak U)= 0\}$ respectively. \newline

It is known that any $\bar {\alpha}\in \frak g^{\ast}$  gets a unique Jordan decomposition $\overline {\alpha}= \overline {\alpha_s}+ \overline {\alpha_n}$,
where  the linear functionals $\overline {\alpha_s}$ and $\overline {\alpha_n}$ must satisfy the following conditions: \newline

$\exists g \in G$ such that 

(i)$g\cdot \overline {\alpha_s}\in \overline {h^{\ast}}$,
(ii)$g\cdot \overline {\alpha^n}\in \overline {U^{\ast}}$, and
(iii)$(g\cdot \overline {\alpha_s})(h_{\alpha}) \neq 0 \Longrightarrow (g\cdot \overline {\alpha_n})(x_{\alpha})= (g\cdot \overline {\alpha_n}) (x_{-\alpha})= 0$, where $\{x_{\alpha},x_{-\alpha},h_{\alpha} \}$ denotes the canonical basis of  $sl_2 (F)$ for any root $\alpha$ of the root system $\Phi$ of $T$ in $G$ and the action of $g$ is coadjoint.   \newline

So due to  the above properties, any $p$- character $\chi$ of an irreducible representation $\rho_{\chi}$  has a unique Jordan decomposition like $\chi= \chi_s+ \chi_n$ such that $\chi$ is contained in the $G$- orbit  of $\overline {\frak B^{\ast}}$. \newline

We thus may assume that  the associated $p$-character $\chi$  of $\rho_{\chi}$ always vanishes on the nilradical $\frak n^+= \frak U \subset B$. For any given character $\chi\in g^{\ast}$, we shall consider  linear functionals $\lambda: \frak h\rightarrow F$ called weights such that $\forall h \in \frak h,\lambda (h)^p- \lambda (h^{[p]})= \chi (h)^p $ holds. \newline

Since $\chi (\frak U)= 0$, we can naturally extend  such a $\lambda $ to a Lie algebra homomorphism $\lambda: \frak B \rightarrow F$. An algebra homomorphism  $\lambda: u(\frak B ) \rightarrow F$ may still arise by extending  this $\lambda$.\newline

 Letting $W(\chi)$ be the set of all such weights, we obviously perceive that the cardinality $\#(W(\chi))$ of $W(\chi)$ is equal to $p^l$ with dim $\frak h= l$.

\begin{prop}
As in definition3.11, we let $L(G)= \frak g$ be of classical type. Suppose that $\chi$ is a character in $\overline {\frak B^{\ast}}$ as above. \newline

We then have a $1-1$ correspondence between weights $\lambda \in W(\chi)$ and $1$-dimensional $\frak B_{\chi}$- modules, where  $\frak B_{\chi}$ denotes  the subalgebra of  $\frak U_{\chi}(\frak g)$ generated by   $\frak B$ and $ 1$.\newline

For a given weight  $\lambda$, let $F_{\lambda}$ be the corresponding $1$- dimensional $\frak B_{\chi}$- module and put $Z_{\chi}(\lambda):= \frak U_{\chi}(\frak g) \otimes_{\frak B_{\chi}} F_{\lambda}.$ We then have  dim$Z_{\chi} (\lambda)= p^m$ and  in addition any irreducible $\frak U_{\chi}(\frak g)$- module becomes the quotient  of some $Z_{\chi}(\lambda)$ .
\end{prop}

\begin{proof}
Refer to  proposition1.5 in\newline 
$\textit{modular representation theory 
of  Lie algebras}$ by \newline
E.M.Friedlander and B.J. Parshall,american journal of math,\newline
110(1988)1055-1094,or
proposition3.1.20 in [KY-4] .

\end{proof}

\begin{definition}
Let $ F$ be any field of nonzero characteristic and let $R$ be a commutative algebra over $F$ with a unit element. Let $R'$ be another ring. \newline

We say that  $R'$ is an $R$- ring if the following conditions are satisfied:
\newline

 $Au$ must be defined in $R'$; namely $Au\in R', \forall A\in R, \forall u\in R'$ with the following properties from (i) to (vi).\newline

(i)If $A\equiv A'$ and $u\equiv u'$ in $R$ and $R'$ respectively, then $Au= A'u'$,\newline

(ii)$A(u+ v)= Au+ Av$,\newline

(iii)$(A+ A')u= Au+ A'u$,\newline

(iv)$(AA')u= A(A'u)$,\newline

(v)$ A(uv)= (Au)v= u(Av)$, \newline

(vi)$1_R\cdot u= u, \forall A, A'\in R, \newline
\forall u,u',v\in R'$ and for the unity $1_R$ in $R$.

\end{definition}
Suppose that $M$ is  an $R$-module and let $\theta$ be an $F$- algebra homomorphism of $R$ onto another $F$- algebra $\theta R$.\newline

Obviously we can define  a homomorphism $\theta^M$ of  $R$- module $M$ onto a  $\theta R$- module $\theta^M M$ in as  general a sense as possible.

We  consider the module $\theta^M M$ which is generated by $\{\theta^M u \vert u\in M\}$ and subject to  such defining relations as\newline

 $\theta^M (u+ v)= \theta^M u+ \theta^M v, \forall u,v\in M$ and $\theta^M(Au)= \theta (A)\cdot \theta^M (u), \forall u\in M, \forall A\in R$.\newline

Next if $M$ is an $R$- ring, then $M$ is automatically an $R$- module. If we define a natural multiplication  in $\theta^M M$  by the recipe $\theta^M(u) \cdot \theta^M (v)= \theta^M (uv),\forall u,v \in M$, then we easily see that $\theta^M M$ becomes a $\theta R$- ring.\newline

If the homomorphism $\theta^M: M\rightarrow \theta^M M$ is equipped with a ring homomorphism at the same time as well as a module homomorphism in a general sense above, then we call  this homomorphism a $\textit {specialization}$ of the $R$-ring $M$ over $\theta$.

\begin{prop}Let  $L$ be a finite dimensional  restricted Lie  algebra of  dimension $n$ over an algebraically closed field $F$ of nonzero characteristic $p$.\newline

We then have that $u(L)$ is a free $\frak Z (u(L))$-module and its rank is $p^{2m}$,where $Q(u(L))= (\frak Z (u(L))- \{0\})^{-1}$ and $ [Q(u(L)): Q(\frak Z(u(L))]= p^{2m}$. Here we call $Q(u(L))$  the $\textit {quotient algebra}$ of $u(L)$, which turns out to be  a division algebra.
\end{prop}

\begin{proof}

For the $Q(\frak Z)$- vector space $Q(u(L))$, we may choose a basis of elements in $u(L)$. So our proposition is  evident considering the Poincare-Birkoff- Witt theorem  and elementary algebraic geometry.\newline

 Next putting $\mathcal O_1(L)=alg<\{x^p-x^{[p]}\vert \forall x\in L\}>$, which is the subalgebra of $u(L)$  generated by all $x^p- x^{[p]}, \forall x\in L$, we obtain $p^{dim L}= p^n=$ $ [Q(u(L)): Q(\frak Z)][Q(\frak Z): Q(\mathcal O_1 (L))]$ with $[Q(u(L)): Q(\frak Z)]= p^{2m}$.

\end{proof}

\begin{prop}

Suppose that  $F$ is  a  splitting  field  of nonzero characteristic  $p$   for  the  $F$- algebra  $u(L)$  ;  then we have the  following:\newline

(i)We may extend  any specialization $\theta$ over $F$ of the center $\frak Z$ of $u(L)$ onto an extension field $\theta \frak Z$ of $F$ to a specialization $\theta^{u(L)}$ of $u(L)$ over $F$  onto an algebra $\theta^{u(L)} u(L)$ of rank not greater than $p^{2m}$ over $\theta \frak Z$.\newline

(ii)There is a general element of $u(L)$ over $\frak Z$ whose minimal polynomial is mapped by $\theta$ onto a multiple of  the minimal polynomial of the corresponding  general element of $\theta^{u(L)}$ over $\theta \frak Z$\newline

(iii)The algebra $\theta^{u(L)}u(L)$ of maximal dimension of an indecomposable representation $\theta^{u(L)}$ is seperable over $\theta \frak Z$ if and only if it becomes centrally simple of dimension $p^{2m}$  over $\theta \frak Z$.\newline

If this is the case, then $\theta$ just maps any minimal polynomial of $u(L)$ over $\frak Z$ onto a minimal polynomial of  $\theta^{u(L)}u(L)$ over $\theta \frak Z$.\newline

(iv)The discriminant ideal  over $\frak Z$ defined by H.Zassenhaus never vanish.

\end{prop}
\begin{proof}
(i)Because $\theta \frak Z$ becomes an extension field of $F$, we see that $\theta^{u(L)}u(L)$ has a dimension over $\theta \frak Z$ obviously.\newline

However the rank of $u(L)$ over $\frak Z$ is nothing but $p^{2m}$ by virtue of  the preceding proposition. Hence we have that the rank of $\theta^{u(L)}u(L)$ over $\theta \frak Z$ must be at most $p^{2m}$.\newline

(ii)By virtue of Brauer- Albert theorem and P-B-W theorem , we have general elements $\alpha$ and $\alpha'$ which are conjugate to each other such that 

$\{\alpha^i (\alpha')^j  \vert 0\leq i,j < p^m \}$  constitutes a $Q(\frak Z)$ basis  for $Q(u(L))$ and such that the minimal polynomial of $\alpha$  over $\frak Z$ becomes of the form $P(\mathbb X)=  \mathbb X^{p^m}+ \sum_{i=1}^{p^m} P_i(x_1,\cdots, x_n)\mathbb X^{p^m- i}$, where $P_i$ denotes a homogeneous polynomial of degree $i$ contained in $\frak Z [x_1,\cdots, x_n]$ for an $F$-basis  $\{x_1,\cdots, x_n\}$ of $L$.\newline

It follows that $P(\theta \alpha)=0$ is obtained from $P(\alpha)=0$, whence the minimal polynomial polynomial of $\theta \alpha$ over $\theta \frak Z$ divides $\theta P$.\newline

(iii) Recall that  the $\theta Z$-algebra  $\theta^{u(L)}u(L)$ is defined to be separable  if  for every field extension $K$ of $\theta \frak Z$ , the algebra   $\theta^{u(L)}u(L) \otimes _{\theta \frak Z} K$ is semisimple. Because an irreducible $\theta^{u(L)}$-module gives rise to an irreducible $u(L)$- module by pull-back, we see that its maximal dimension is $p^m$\newline

So by the Jacobson's density theorem the algebra  $\theta^{u(L)} u(L)$ is seperable over $\theta \frak Z$ if and only if it becomes centrally simple of dimension $p^{2m}$ over $\theta \frak Z$.\newline

 No doubt the latter assertion holds by (i).\newline

(iv)In the proof of (i), we note that $\alpha$ is a primitive element of  a maximal separable subfield of  $Q(u(L))$.\newline

So the minimal  polynomial of $\alpha$ over $Q(\frak Z)$ becomes of degree $p^m$.
We see that the second highest coefficient of this minimal polynomial is nothing but $tr(\alpha)$.\newline

We now generalize such a linear functional $tr$ to any $R$-ring $M$ in case $M$ is  of rank $r$ over $R$ as a free $R$- module. We know that  the $\textit{discriminant ideal}$  is defined to be  the ideal $\frak D_{M,R,tr}$ of $R$ generated by the set of all  the determinants $\vert tr(u_i, v_j)\vert$, where $u_i, v_j\in M$ and  $1\leq i,j \leq r$.\newline

 If we use the Laplace development of the determinant concerned, then we have that for the given $R$- basis $\{b_1,b_2,\cdots,b_r\}$ of $M$, $\frak D_{M,R,tr}=R\vert tr (b_i, b_j) \vert.$\newline

So it is clear that  $\frak D_{M,R,tr}\neq 0$ if and only if $tr(-, -):= tr(-\cdot -)$ is nondegenerate.\newline
It is known that  if  the quotient  algebra $Q(M)$ over $Q(R)$ is defined  and is seperable over $Q(R)$, then the trace form  of $M$ over $R$ is nondegenerate.\newline

 So we have $\frak D_{u(L),M,tr} \neq 0$ because $Q(u(L))$ in particular becomes 
centrally simple and hence separable over $Q(\frak Z)$.

\end{proof}

\begin{prop}
We assume that $\theta$ is given as a specialization over $F$ of $\frak Z$ onto an extension field $\theta \frak Z$ of  $F$; \newline

we then have $\theta (\frak D_{u(L),\frak Z,tr})\neq 0$  if and only if $\theta^{u(L)} u(L)$ becomes centrally simple of  dimension $p^{2m}$ over $\theta \frak Z$ for an indecomposable representation $\theta^{u(L)}$. 

\end{prop}

\begin{proof}
($\Longleftarrow)$
By virtue of the preceding proposition 3.15, any minimal polynomial of $u(L)$ over $\frak Z$ is mapped onto a minimal polynomial of $\theta^{u(L)} u(L)$ over $\theta \frak Z$ under the hypothesis  that $\theta^{u(L)}u(L)$ becomes centrally simple of dimension $p^{2m}$ over $\theta \frak Z$, in which case we have $\theta (\frak D_{u(L),\frak,tr})= \frak D_{\theta ^{u(L)}u(L), \theta \frak Z,\theta tr}$ holds inevitably.\newline

 Because $\theta^{u(L)} u(L)$ is separable over $\theta \frak Z$, we have $\frak D_{\theta^{u(L)},\theta \frak Z,\theta tr} \neq 0$. We should remember that  the second highest coefficient of the minimal polynomial  of an element  over  a field  is  just  the trace of the element.\newline

($\Longrightarrow)$
Suppose  that $\theta (\frak D_{u(L),\frak Z, tr}) \neq 0$.  We have immediately that $tr(ab)= tr (ba), tr((ab)c)= tr(a(bc)),\forall {a,b,c\in u(L)}$ and $tr(a^pb^p)= (tr(ab))^p$. Furthermore we have that  $\forall u,v,w\in \theta^{u(L)} u(L)$, and   $\theta tr(uv)= \theta tr (vu),\theta tr((uv)w)= \theta tr(u(vw))$ and $\theta tr (u^pv^p)= \{\theta tr(uv)\}^p$  since $\theta$ is a specialization . \newline

By the given condition we have $\theta (\frak D_{u(L),\frak Z,tr})= \frak D_{\theta^{u(L)},\theta \frak Z, \theta tr} \neq 0$, so that  $\theta tr$ turns out to be  a  nondegenerate symmetric bilinear form $\overline {\theta tr}$ on $\theta^{u(L)}u(L)/\frak A$ over $\theta \frak Z$ \newline

satisfying $\overline {\theta tr}(\overline {u^p},\overline {v^p})= \{\overline {\theta tr}(\bar u, \bar v)\}^p$ and hence $\overline{\theta tr}(\overline u^{p^j},\overline v^{p^j}),$ \newline

$\forall \overline u, \overline v\in \theta^{u(L)}u(L)/\frak A$ and the dimension of $\theta (u(L))$ over $\theta \frak Z$ equals $p^{2m}$, where $\frak A$ indicates  a two-sided ideal of $\theta^{u(L)}u(L)$ consisting of all element $u\in \theta^{u(L)}u(L)$ which satisfies the condition $\theta tr (u,v)= 0, \forall v\in \theta^{u(L)}u(L).$\newline

Moreover  it is easy to see that $dim \theta^{u(L)}u(L)/\frak A= p^{2m}$ over $\theta \frak Z$ considering the proofs of the  above propositions. \newline

Next for any element $\bar x$ of the Jacobson (or Wedderburn) radical of $\theta^{u(L)}$, we have $\overline x^{p^j}= 0$ for some $j$. Hence  we have $\{\overline {\theta tr}(\bar x, \bar v)\}^{p^j}= \overline {\theta tr} (\bar x^{p^j},\bar v^{p^j})= 0$ , and so  $\forall \bar v\in \theta^{u(L)}u(L)/\frak A$, $\overline {\theta tr}(\bar x, \bar v)= 0$. \newline

We thus get $\bar x= 0$, and hence  it follows that  $\theta^{u(L)}u(L) $ is semisimple  over $ \theta \frak Z$. Such an argument  still holds  for  any extension  of the ground field $\theta \frak Z$ of $\theta^{u(L)}u(L)/\frak A$,and so $ \theta^{u(L)}u(L)/\frak A$ also becomes semisimple over $\theta \frak Z$,i.e., $\theta^{u(L)}u(L)/ \frak A$ is separable over $\theta \frak Z$.\newline

 By the way we know from the preceding proposition that the degree of a minimal polynomial of $\theta^{u(L)} \leq p^m$, which holds a fortiori for $\theta^{u(L)}u(L)/\frak A$.\newline

From proposition 3.15(iii),we have that $\theta^{u(L)}/\frak A$ becomes separable if and only if  it is centrally simple  of the same dimension $p^{2m}$ over $\theta \frak Z$.\newline

 Hence we get $\frak A=0$,i.e., $\theta^{u(L)}u(L)$ over $\theta \frak Z$ must be centrally simple of  dimension $p^{2m}$. \newline

We had better  also remember that  if the determinant $\vert tr(b_i, b_j) \vert$ as in the proof of  proposition 3.16, is nonzero, then $\{b_1, b_2,\cdots, b_r \}$ becomes a basis of  the free $R$-module $M$.
\end{proof}

\begin{cor}
An irreducible representation $\theta^{u(L)}$ of the Lie algebra $L$ has its representation space of dimension $p^{2m}$  if and only if $\theta^{u(L)}(\frak D_{u(L),\frak Z,tr})\neq 0.$
\end{cor}
\begin{proof}
By virtue of the preceding propositions3.15 and 3.16, our claim is straightforward.
\end{proof}

Now suppose that $V$ is any finite dimensional irreducible $L$- module and that  $V$ has its associated irreducible representation $\rho_ \chi : u(L) \rightarrow End_F (V)$ with a character $\chi \in L^{\ast} $ such that for any $x \in L$, $\rho_\chi (x)^p- \rho_\chi(x^{[p]})= \chi(x)^p id_V$. \newline

We  should note that $\rho_\chi (u(L))$ becomes dense in $End_{\Delta} (V)$ by dint of  Jacobson's  density theorem, where $\Delta:= End_{u(L)}(V).$    
By Schur's lemma we have $\Delta= F$.\newline

 By the general theory of modular Lie algebra $[Q(u(L): Q(\frak Z(u(L))]= p^{2m}$ for some integer $m$ with $2m \leq n$.\newline

 Here $\frak Z:= \frak Z (u(L))$ denotes the center of  $u(L)$ as before and  $Q$ denotes the relevant quotient  algebra of  noncommutative algebras $u(L)$ and the commutative algebra $\frak Z (u(L))$ respectively. So it follows that  $\rho_\chi (u(L))\cong u(L)/ker \rho_{\chi} \cong M_m(F)$,the full matrix algebra [SF]. \newline

Hence it is sometimes very convenient  for us to  express  the basis of  this factor algebra  related to classical type Lie algebra  by the representatives of the   form:$\{\otimes _{i=1}^ {2m} (B_i+ A_i)^{j}: 0\leq j \leq p-1$,where $B_i $  is an element in the CSA of $ L$ and $A_i$'s are elements of $u(L) \forall i=1,\cdots, 2m$.
We shall call such a form $\textit {Lee's basis}$ of the Lie algebra concerned.

\section{Four kinds of points}

Now let $L$ be a finite dimensional restricted Lie algebra with a $p$- mapping $[p]$ over an algebraically closed field  $F$  of nonzero characteristic $p$  and  with a basis $\{x_1,x_2, \cdots, x_n\}$ which is centerless and indecomposable.\newline

Suppose further that the center $\frak Z(u(L))$ of $u(L)$ has the Noether normalization form $F[x_1- x_1^{[p]},x_2- x_2^{[p]},\cdots,x_n-  x_n^{[p]},s_1,s_2,\cdots, s_k]$,

where $s_i$'s are algebraic over the field generated by algebraically independent elements $\{x_i- x_i^{[p]}: 1\leq i \leq n \}.$\newline

Any maximal ideal of $u(L)$ must contain a certain ideal of the form $\{\sum_{i=1}^n u(L)(x_i- x_i^{[p]}- \xi_i)\} + \sum_{j=1}^k u(L)(s_j- \mu_j)$, where $\xi$'s and $\mu_j$'s are some constants in $F$.\newline

We defined  four kinds of points  in [KY-4]  and shall  exhibit  some  explanation therein as follows:

\begin{definition}

We consider four possible cases and give names  associated with points $(\xi_1,\cdots \xi_n, \mu_1,\cdots, \mu_k)$  on Zassenhaus variety obtained from the center $\frak Z:= \frak Z(u(L))$ of $u(L)$.\newline

[I] If all $\xi_i=0$, then there may exist finitely many left maximal ideals $\rho_l$ containing $\{\sum_{i=1}^n u(L)(x_i- x_i^{[p]}- 0)\} + \sum_{j=1}^k u(L)(s_j- \mu_j)$, so that  $u(L)/\rho_l$ becomes $p$- representation modules for $L$ with dimension$\leq p^m$. We shall call such a point  $(0,\cdots 0, \mu_1,\cdots, \mu_k)$ a $p$- point.\newline

In particular if it is a $p$-point and its associated irreducible module has dimension $p^m$, then we call the point  a regular $p$-point.\newline

 [II] If not all $\xi_i$'s are zero,i.e.,  $\xi_i \neq 0$ for some $i$, then there are
finitely many left maximal ideals $\rho_l$ containing $\{\sum_{i=1}^n u(L)(x_i- x_i^{[p]}- \xi_i)\} + \sum_{j=1}^k u(L)(s_j- \mu_j)$, so that  $u(L)/\rho_l$ become  irreducible modules for $L$ with dimension$\leq p^m$. \newline

If all these irreducible modules become of dimension  $p^m$, then they are isomorphic $L$- irreducible modules.\newline

 In this case we call such a point $(\xi_1,\cdots, \xi_n, \mu_1,\cdots, \mu_k)$ a  regular point. 
We shall denote the set of all  regular points by $R(L,p,\chi)$.\newline

For all left maximal ideals $\rho_l$ containing $\{\sum_{i=1}^n u(L)(x_i- x_i^{[p]}- \xi_i)\} + \sum_{j=1}^k u(L)(s_j- \mu_j)$, $u(L)/\rho_l$ may have  dimension less than $p^m$  and are possibly nonisomorphic.\newline

 We shall call such a point subregular-point. We shall denote the set of all subregular points by $S(L.p,\chi)$.
We should keep track of  the relation of a point and its associated character $\chi$.

\end{definition}

For  our examples of such points mentioned above, we recapitulate  important facts  relating to $A_l,B_l,C_l$ and $D_l$-type modular Lie algebras.\newline

The $A_l$-type Lie algebra  over $\mathbb C$ has its root system $\Phi $= $\{\epsilon_i- \epsilon_j | 1 \leq i \neq j \leq l+1 \}$,where $\epsilon_i$'s are orthonormal unit vectors in the Euclidean space $\mathbb R^{l+1}$.The base of  $\Phi$ is equal to $\{\epsilon_i - \epsilon_{i+1}| 1\leq i \leq l\}$.\newline

We let $L$ be an $A_l$-type simple Lie algebra over an algebraically closed field  $F$ of  characteristic $p \geq 7$.

For a root $\alpha \in \Phi,$ we put $g_\alpha :=  x_\alpha^{p-1}- x_{- \alpha}$and $w_\alpha:= (h_\alpha+ 1)^2+ 4x_{-\alpha}x_\alpha$.\newline

We have seen from [KC] and [KY-4] that  any $A_l$-type modular Lie algebra  over $F$ becomes a Park's  Lie algebra. However we would like to  specfy the proof  or prove  it in a different way. 
We let $L:= L(G)$ be any $A_l$- type classical Lie algebra over $F$ associated with an algebraic group $G$.

\begin{prop}

Let  $\alpha$  be any  root  in the root system  $\Phi$ of $L .$ If $\chi(x_\alpha)$ $\neq0,$ then $dim_F$$ \rho_{\chi}$$(u(L))$ = $p^{2m},$ where $ [Q(u(L)):Q(\mathfrak{Z})]$=$p^{2m}$=$p^{n-l}$ with $\mathfrak{Z}$ the center of $u(L)$  and  $Q$  denotes  the quotient algebra. \newline

So we claim that  the  simple module corresponding to this representation has $p^m$ as its  dimension.

\end{prop}

\begin{proof}

 We may put $\alpha=\epsilon_1-\epsilon_2$ since  there is  only one root  length  for $L$ and all roots of the same length are conjugate under the Weyl group of $\Phi$.\newline

We put $B_i:= b_{i1}h_{\epsilon_1- \epsilon_2}+\cdots + b_{il}h_{\epsilon_l- \epsilon_{l+1}} $ as in [KC],[KY-4] . If we put  the kernel of $\rho_{\chi}$ as $\frak M_{\chi}$, then we  contend  that we have a  basis in $u(L)/\frak M_{\chi}$, $\frak B:= \{(B_1+ A_{\epsilon_1- \epsilon_2})^{i_1}\otimes (B_2+A_{\epsilon_2 - \epsilon_1})^{i_2}\otimes (\otimes_{j=3}^{l+1}(B_j+ A_{\epsilon_1- \epsilon_j})^{i_j})\otimes(\otimes_{j=3}^{l+1}(B_{l-1+j}+ A_{\epsilon_j- \epsilon_1})^{i_{l-1+j}})\otimes (\otimes_{j=3}^{l+1}(B_{2l-2+j}+ A_{\epsilon_2- \epsilon_j})^ {i_{2l-2+j}})\otimes (\otimes _{j=3}^{l+1}(B_{3l-3+j}+ A_{\epsilon_j- \epsilon_2})^{i_{3l-3+j}}\otimes \cdots \otimes (B_{2m-1}+ A_{\epsilon_l- \epsilon_{l+1}})^{i_{2m-1}}\otimes (B_{2m}+ A_{\epsilon_{l+1}-\epsilon_l})^{i_{2m}}\}$ for $0 \leq i_j \leq p-1$, \newline

where we put  \newline

$A_{\epsilon_1-\epsilon_2}= x_\alpha= x_{\epsilon_1- \epsilon_2}$,\newline

$ A_{\epsilon_2- \epsilon_1}=c_{\epsilon_2- \epsilon_1}+ (h_{\epsilon_1- \epsilon_2}+ 1 )^2+ 4x_{\epsilon_2-\epsilon_1}x_{\epsilon_1- \epsilon_2}= c_{\epsilon_2- \epsilon_1}+ w_{\epsilon_1- \epsilon_2},$\newline

$A_{\epsilon_2-  \epsilon_3}=     x_{\epsilon_1- \epsilon_3}(c_{\epsilon_2- \epsilon_3}+ x_{\epsilon_2- \epsilon_3}x_{-(\epsilon_2- \epsilon_3)}\pm x_{\epsilon_1- \epsilon_3}x_{-(\epsilon_1- \epsilon_3)})$, \newline

$A_{\epsilon_2- \epsilon_k}= x_{\epsilon_3- \epsilon_k}
(c_{\epsilon_2- \epsilon_k}+ x_{\epsilon_2- \epsilon_k}x_{-(\epsilon_2- \epsilon_k)}\pm x_{\epsilon_1- \epsilon_k}x_{-(\epsilon_1- \epsilon_k)})  $ \newline

if $k\neq 1$, \newline

$A_{ \epsilon_3- \epsilon_1}= x_{ \epsilon_3- \epsilon_2}(c_{ \epsilon_3- \epsilon_1}+ x_{\epsilon_2- \epsilon_3}x_{-(\epsilon_2- \epsilon_3)}\pm x_{\epsilon_1 -\epsilon_3}x_{-(\epsilon_1- \epsilon_3)}) $,\newline

 $A_{-(\epsilon_1- \epsilon_k)}= x_{-(\epsilon_3- \epsilon_k)}(c_{-(\epsilon_1- \epsilon_k)}+ x_{\epsilon_2- \epsilon_k }x_{-(\epsilon_2- \epsilon_k)}\pm x_{\epsilon_1- \epsilon_k}x_{-(\epsilon_1- \epsilon_k)}) $, \newline

with the signs chosen so that they  may commute with $x_\alpha$ and with $c_\beta\in F$ chosen so that $A_{\epsilon_2-\epsilon_1}$ and parentheses are invertible.\newline
For any other root $\beta$, we put $A_\beta= x_\beta^2$ or $x_\beta^3 $ if possible.\newline

Otherwise we make use of the parentheses(      ) again used for designating $A_{-\beta}$. So in this case we put $A_\beta= x_\gamma^2 $       
 or $ x_\gamma^3  $ attached to these (      ) so that  $x_\alpha$ may commute with $A_\beta$.\newline

We may see without difficulty that $\frak B$ is a linearly independent set in $U(L)$ by virtue of P-B-W theorem.\newline

We shall prove that a nontrivial linearly dependent  equation leads to absurdity. We assume first that we have a dependence equation which is of least degree with respect to $h_{\alpha_j}\in H$ and the number of whose highest  degree terms is also least.\newline

In case it is conjugated by $x_\alpha$, then there arises a nontrivial dependence equation of lower degree than the given one,which contradicts to our assumption.\newline

Otherwise  it reduces to the following form, which we have only to prove its absurdity:\newline

 $(\ast) x_{\epsilon_{1}-\epsilon_{2}}$$K$ + $K'$ $\in$$ \mathfrak{M}_\chi$ , where $K$ and $K'$  commute with $x_{\alpha}$ modulo $\frak {M}_\chi $.\newline

 We thus deal with and assume  $x_{\epsilon_{2}-\epsilon_{1}}$$x_{\epsilon_{1}-\epsilon_{2}}$$K$ + $x_{\epsilon_{2}-\epsilon_{1}}$$K'$ $\in$ $\mathfrak{M}_\chi$.\newline

From $w_{\epsilon_1- \epsilon_2}:=(h_{\epsilon_1- \epsilon_2}+ 1)^2 +4 x_{\epsilon_2- \epsilon_1}x_{\epsilon_1- \epsilon_2}\in $ the center of  $u(\frak{sl}_2(F)), $ we get $ 4^{-1}\{w_{\epsilon_1- \epsilon_2}- (h+ 1)^2\}K+ x_{\epsilon_2- \epsilon_1}K' \equiv 0 $  modulo $\frak M_\chi$.\newline

If $x_{\epsilon_2- \epsilon_1}^p\equiv c $ which is a constant,then \newline

$(\ast \ast)4^{-1}x_{\epsilon_2-\epsilon_1}^{p-1}\{w_{\epsilon_1- \epsilon_2}- (h_{\epsilon_1- \epsilon_2}+ 1)^2\}K+ cK'\equiv 0 $\newline

 is obtained.
\newline
From $(\ast),(\ast \ast)$, we have\newline

 $4^{-1}x_{\epsilon_2- \epsilon_1}^{p-1}\{w_{\epsilon_1- \epsilon_2}- (h_{\epsilon_1- \epsilon_2}+ 1)^2\}K- cx_{\epsilon_1- \epsilon_2}K\newline
\equiv 0$ modulo $\frak M_\chi$.\newline

Multiplying $x_{\epsilon_1- \epsilon_2}^{p-1}$ to this equation,we obtain \newline

 $(\ast \ast \ast)4^{-1}x_{\epsilon_1- \epsilon_2}^{p-1}x_{\epsilon_2- \epsilon_1}^{p-1}\{w_{\epsilon_1- \epsilon_2}- (h_{\epsilon_1- \epsilon_2}+ 1)^2\}K- cx_{\epsilon_1- \epsilon_2}^pK\equiv 0.  $\newline

By making use of $w_{\epsilon_1- \epsilon_2}$, we may deduce from $(\ast \ast \ast)$ an equation of the form \newline
( a polynomial of degree $\geq 1$ with respect to  $ h_{\epsilon_1- \epsilon_2})K- cx_{\epsilon_1- \epsilon_2}^pK\equiv 0.   $\newline

Finally if we use conjugation and subtraction consecutively,then we are led to a  contradiction $K\in \frak M_\chi.$\newline
It may be necessary  for us to change the role of $K$ and $K'$ to obtain the absurdity alike.\newline

\end{proof}

 \begin{prop}
Suppose that $\chi$ is  a character  of any  simple $L$-module with $\chi(h_\alpha)\neq 0 $ for some $\alpha \in$ the base of $\Phi$,where $h_\alpha $ is an  element  in  the Chevalley basis of $L$ such that $Fx_\alpha+ Fx_{-\alpha}+F h_\alpha= \frak {sl}_2 (F) $ with   $[x_\alpha,x_{-\alpha}]= h_\alpha \in H $(a CSA of  $L)$.\newline

We then have that  the dimension of any simple $L$-module  with character $\chi= p^m= p^{(n-l) \over 2}$,where   $  n$= dim $L= 2m+ l$ for $ H $ with dim $  H= l$.
\end{prop}

\begin{proof}
If $\chi(x_\alpha)\neq 0$ or $\chi(x_{-\alpha})\neq 0$, then our assertion is evident from the preceding proposition 4.2 or [KC]. So we may assume that $\chi(x_\alpha)= \chi(x_{-\alpha})= 0 $ but $\chi(h_\alpha)\neq 0$.\newline

Furthermore we may put $\alpha= \epsilon_1- \epsilon_2$ without loss of generality since all roots are conjugate under the Weyl group of $\Phi$.\newline

Since the case for $l=1$ is trivial, we may assume $l \geq 2$. For $i=1.2, \cdots$, we put $B_i:= b_{i1}h_{\epsilon_1- \epsilon_2}+\cdots + b_{il}h_{\epsilon_l- \epsilon_{l+1}} $ as in [KC],[KY-4] and we put $\frak B:= \{(B_1+ A_{\epsilon_1- \epsilon_2})^{i_1}\otimes (B_2+A_{\epsilon_2 - \epsilon_1})^{i_2}\otimes (\otimes_{j=3}^{l+1}(B_j+ A_{\epsilon_1- \epsilon_j})^{i_j})\otimes(\otimes_{j=3}^{l+1}(B_{l-1+j}+ A_{\epsilon_j- \epsilon_1})^{i_{l-1+j}})\otimes (\otimes_{j=3}^{l+1}(B_{2l-2+j}+ A_{\epsilon_2- \epsilon_j})^ {i_{2l-2+j}})\otimes (\otimes _{j=3}^{l+1}(B_{3l-3+j}+ A_{\epsilon_j- \epsilon_2})^{i_{3l-3+j}}\otimes \cdots \otimes (B_{2m-1}+ A_{\epsilon_l- \epsilon_{l+1}})^{i_{2m-1}}\otimes (B_{2m}+ A_{\epsilon_{l+1}-\epsilon_l})^{i_{2m}}\}$ for $0 \leq i_j \leq p-1$, \newline

where we set \newline

$A_{\epsilon_1- \epsilon_2}= g_\alpha =  g_{\epsilon_1- \epsilon_2}= x_{\epsilon_1- \epsilon_2}^{p-1}- x_{\epsilon_2- \epsilon_1},$ \newline

 $A_{\epsilon_2- \epsilon_1}= c_{\epsilon_2- \epsilon_1}+ (h_\alpha +1)^2 + 4^{-1}x_{-\alpha}x_\alpha,$\newline

$A_{\epsilon_1- \epsilon_3}= g_\alpha^2 (c_{\epsilon_1- \epsilon_3}+ x_{\epsilon_2- \epsilon_3}x_{\epsilon_3- \epsilon_2}\pm x_{\epsilon_1- \epsilon_3}x_{\epsilon_3- \epsilon_1}),$\newline

$ A_{\epsilon_3- \epsilon_1}= g_\alpha^3(c_{\epsilon_3- \epsilon_1}+ x_{\epsilon_3- \epsilon_2}x_{\epsilon_2- \epsilon_3}\pm x_{\epsilon_3- \epsilon_1}x_{\epsilon_1- \epsilon_3})$ or $ x_{\epsilon_3- \epsilon_4}(c_{\epsilon_3- \epsilon_1}+ x_{\epsilon_3- \epsilon_2}x_{\epsilon_2- \epsilon_3}\pm x_{\epsilon_3- \epsilon_1}x_{\epsilon_1- \epsilon_3}), $\newline

$A_{\epsilon_2- \epsilon_j}=  g_\alpha^4( c_{\epsilon_2- \epsilon_3}+ x_{\epsilon_2- \epsilon_3}x_{\epsilon_3- \epsilon_2}\pm x_{\epsilon_1- \epsilon_3}x_{\epsilon_3- \epsilon_1}) $(if $j= 3)$ or  $x_{\epsilon_4- \epsilon_j}(c_{\epsilon_2- \epsilon_j}+ x_{\epsilon_2- \epsilon_j}x_{\epsilon_j- \epsilon_2}\pm x_{\epsilon_1- \epsilon_j}x_{\epsilon_j- \epsilon_1})$,\newline

$A_{\epsilon_j- \epsilon_2}= g_\alpha^5(c_{\epsilon_3- \epsilon_2}+ x_{\epsilon_2-\epsilon_3}x_{\epsilon_3- \epsilon_2} \pm x_{\epsilon_1- \epsilon_3}x_{\epsilon_3- \epsilon_1})$ (if $j= 3)$ or $x_{\epsilon_j- \epsilon_4}(c_{\epsilon_j- \epsilon_2}+ x_{\epsilon_j- \epsilon_2}x_{\epsilon_2- \epsilon_j}\pm x_{\epsilon_j- \epsilon_1}x_{\epsilon_1- \epsilon_j}),$\newline

$A_{\epsilon_2- \epsilon_4}= x_{\epsilon_3- \epsilon_4}^2(c_{\epsilon_2- \epsilon_4}+ x_{\epsilon_2- \epsilon_4}x_{\epsilon_4- \epsilon_2}\pm x_{\epsilon_1- \epsilon_4}x_{\epsilon_4- \epsilon_1}),$\newline

$A_{\epsilon_4- \epsilon_2}= x_{\epsilon_4- \epsilon_3}(c_{\epsilon_4- \epsilon_2}+ x_{\epsilon_4- \epsilon_2}x_{\epsilon_2- \epsilon_4}\pm x_{\epsilon_4- \epsilon_1}x_{\epsilon_1- \epsilon_4}) ,$\newline

$A_{\epsilon_1- \epsilon_j}=x_{\epsilon_3- \epsilon_j}^2(c_{\epsilon_1- \epsilon_j}+ x_{\epsilon_1- \epsilon_j}x_{\epsilon_j- \epsilon_1} \pm x_{\epsilon_2- \epsilon_j}x_{\epsilon_j- \epsilon_2})   , $\newline

$ A_{\epsilon_j- \epsilon_1}    =x_{\epsilon_j- \epsilon_3}^2(c_{\epsilon_J- \epsilon_1}+ x_{\epsilon_1- \epsilon_j}x_{\epsilon_j- \epsilon_1} \pm x_{\epsilon_2- \epsilon_j}x_{\epsilon_j- \epsilon_2}) ,$\newline

$A_{\epsilon_i- \epsilon_j}= x_{\epsilon_i- \epsilon_j}^2$ or $x_{\epsilon_i- \epsilon_j}^3 $ for other roots $\epsilon_i- \epsilon_j$,\newline

where signs are chosen so that they may commute with $x_\alpha$ and $c_\beta$ are chosen so that $A_{\epsilon_2- \epsilon_1}$ and parentheses are invertible in $u(L)/\frak M_\chi$ for the kernel $\frak M_\chi$ in $u(L)$ of any given simple representation of $L$ with the character $\chi$.
\newline

We may see without difficulty that $\frak B$ is a linearly independent set in $u(L)$ by virtue of P-B-W theorem.\newline

We shall prove that a nontrivial linearly dependent  equation leads to absurdity. We assume first that we have a dependence equation which is of least degree with respect to $h_{\alpha_j}\in H$ and the number of whose highest  degree terms is also least.\newline

In case it is conjugated by $x_\alpha$, then there arises a nontrivial dependence equation of lower degree than the given one,which contradicts to our assumption.\newline

Otherwise we have to prove that \newline

(i)$x_{\epsilon_l- \epsilon_k}K+ K'\in \frak M_\chi$ with $l,k \neq 1,2$\newline

(ii) $g_\alpha K+ K'\in \frak M_\chi$ \newline

lead to a contradiction, where both $K$ and $K'$ commute with $x_{\pm \alpha}$  modulo $\frak M_\chi $. In particular $K$ commute with $g_\alpha$.
\newline

For the case (i), we may change it to the form $x_{\alpha}K+ K''\in \frak M_\chi$ for some $K''$ commuting with $x_\alpha= x_{\epsilon_1- \epsilon_2}$ modulo $\frak M_\chi$.\newline

So we have $x_\alpha^p K+ x_\alpha^{p-1}K''\equiv 0$, thus $x_\alpha^{p-1}K''\equiv 0$.\newline

Subtracting from this $x_{-\alpha}x_\alpha K+ x_{-\alpha}K''\equiv 0$, we get \newline

$-x_{-\alpha}x_\alpha K+ g_\alpha K'' \equiv 0$. Recall here that $g_\alpha$ is invertible and $w_\alpha$ belongs to the center of $u(\frak {sl}_2 (F))$ according to [RS].\newline

So we get  $4^{-1}\{(h_\alpha +1)^2- w_\alpha\}K+ g_\alpha K''\equiv 0$, and hence\newline

 $(\ast) g_\alpha^{p-1} 4^{-1}\{(h_\alpha + 1)^2- w_\alpha \}K+ cK'' \equiv 0$\newline

 is obtained and from the start equation we have \newline

$(\ast \ast)cx_\alpha K+ c K''\equiv 0$, where $g_\alpha^p- c \equiv 0$.\newline

Subtracting $(\ast \ast)$ from $(\ast)$, we have $4^{-1}g_\alpha^{p-1}\{(h_\alpha+ 1)^2- w_\alpha\}K- cx_\alpha K \equiv 0$.\newline

Multiplying this equation by $g_\alpha^{1-p}$ to the right, we obtain $4^{-1}g_\alpha^{p-1}\{(h_\alpha+ 1)^2- w_\alpha\}g_\alpha^{1-p}K- cx_\alpha g_\alpha^{1-p}K \equiv 0$ \newline

We thus have $4^{-1}\{(h_\alpha+ 1- 2)^2- w_\alpha\}K- x_\alpha g_\alpha K \equiv 0$.

So it follows that $4^{-1}\{(h_\alpha -1)^2- w_\alpha\}K+ x_\alpha x_{-\alpha}K \equiv 0 $.\newline

Next multiplying $x_{-\alpha}^{p-1}$ to the right of this last equation, we obtain $\{(h_{\alpha}- 1)^2- w_\alpha\}K x_{-\alpha}^{p-1}\equiv 0$.
Now multiply $x_\alpha$ in turn consecutively to the left of this equation until it becomes of the form \newline

( a nonzero polynomial of degree $\geq 1$ with respect to $h_\alpha)K $\newline
$\in \frak M_\chi$,\newline

which comes from the fact that the intersection  of  $\frak M_{\chi}$  and the commutative algebra \newline
 $\frak Z (u(L))$[$h_{\alpha},w_{\alpha}]$  becomes nearly a prime ideal  of   $\frak Z (u(L))$[$h_{\alpha},w_{\alpha}]$ ,which  is    generated by $  h_{\alpha},w_{\alpha}$ over $\frak Z (u(L))$. 
 It is actually so by virtue of  the zero ideal (0).\newline

By making use of  conjugation by $x_\alpha$ and subtraction consecutively, we are led to a contradiction $K \in \frak M_\chi$.
\newline

Finally for the case (ii),we consider $K+ g_\alpha^{-1}K' \in \frak M_\chi$.  So we have $x_\alpha K+ x_\alpha g_\alpha^{-1} K' \equiv 0$ modulo $\frak M_\chi$.

By analogy with the argument  as in the case (i), we obtain a contraiction $K \in \frak M_\chi$.

\end{proof}

Next we note first that the orthogonal Lie algebra of $B_l$- type with rank $l$, i.e., the $B_l$-type Lie algera over $\mathbb{C}$ has its root system $\Phi$=$\{\pm\epsilon_i $(of squared lengh 1);  $\pm(\epsilon_i\pm\epsilon_j )$ (of squared length 2) $ |1\leq i \neq j \leq l $ \},where $\epsilon_i , \epsilon_j $ are linearly independent  orthonormal unit vectors in $\mathbb{R}^l$ with $l \geq 2$. \newline

The base of $\Phi$ equals $\{\epsilon_1-\epsilon_2,\epsilon_2- \epsilon_3,\cdots ,\epsilon_{l-1}-\epsilon_l,\epsilon_l\}$.\newline

For an algebraically closed field $F$ of prime characteristic $p\geq 7,$ the $B_l$- type Lie algebra $L$ over $F$ is just the analogue over $F$ of the $B_l$- type 
simple Lie algebra  over $\mathbb{C.}$\newline

 In other words the $B_l$- type Lie algebra over $F$ is isomorphic to the Chevalley Lie algebra of the form 
$\sum_{i=1}^{n}\mathbb{Z}c_i\otimes_\mathbb{Z}F,$ \newline

where $ n$= $dim_FL$ and $x_\alpha$= some $c_i$ for each $\alpha \in \Phi$ , $h_\alpha$= some $c_j$ with  $\alpha$  some base element  of  $\Phi$  for a Chevalley basis 
\{$c_{i}$\} of  the $B_{l}$ - type  Lie algebra over $\mathbb{C}$ .\newline

 Of course we are well aware that $\frak{sl}_2(F)= Fx_\alpha+ Fx_{-\alpha}+Fh_\alpha$,
where $h_\alpha= [x_\alpha,x_{-\alpha}]$.

\begin{prop}\label {thm4.1}

Let  $\alpha$  be any  root  in the root system  $\Phi$ of $L .$ If $\chi(x_\alpha)$ $\neq0,$ then $dim_F$$ \rho_\chi$$(u(L))$ = $p^{2m},$ where $ [Q(u(L)):Q(\mathfrak{Z})]$=$p^{2m}$=$p^{n-l}$ with $\mathfrak{Z}$ the center of $u(L)$  and  $Q$  denotes  the quotient algebra. \newline

So we claim that  the  simple module corresponding to this representation has $p^m$ as its  dimension and\newline
that  the Lee's basis for this simple module is obtained  as follows.

(I) Suppose that $\alpha$ is  a short root. Since all roots of a given length are conjugate under the Weyl group of $\Phi$,we may put $\alpha$=$\epsilon_1$ without loss of generality.\newline

Let us put $B_i$=$b_{i1}$$h_{\epsilon_1 -\epsilon_2}$ +\newline
$b_{i2}h_{\epsilon_2-\epsilon_3}$+ $\cdots$ +$b_{i,l-1}h_{\epsilon_{l-1}-\epsilon_l}$ + $b_{il}h_{\epsilon_l}$ for $i=1,2,\cdots,2m $ ,where ($b_{i1},\cdots,b_{il})
\in F^l$ are chosen so that  arbitrary ($l+1)-B_i$'s are linearly independent in $\mathbb P^l(F)$,the $\frak B$ below becomes  an $F$-linearly independent set in $u(L)$ if necessary and  $x_\alpha B_i$ $\not \equiv B_i x_\alpha $ with $\alpha =\epsilon_1$.\newline

Let $\frak M_\chi$ be the kernel of the irreducible representation $\rho _{\chi}: L  \rightarrow \frak {gl}(V) $ of the restricted Lie algebra ($L ,[p])$ associated  with any given irreducible $L$-module $V$ with a character $\chi$.\newline

In $u(L)/\frak M_\chi$ we give a basis as $\frak B$:=$\{(B_1 +A_{\epsilon_1})^{i_1}\otimes (B_2 + A_{-\epsilon_1})^{i_2}\otimes(B_3+ A_{\epsilon_1 -\epsilon_ 2})^{i_3}\otimes(B_4+ A_{-(\epsilon_1-\epsilon_2) })^{i_4}\otimes \cdots\otimes (B_{2l}+A_{-(\epsilon_{l-1}-\epsilon_l)})^{i_{2l}}\otimes (B_{2l+1}+ A_{\epsilon _l})^{i_{2l+1}}\otimes(B_{2l+2}+A_{-\epsilon_l})^{i_{2l+2}}\otimes_{j=2l+3}^{2m}(B_j+ A_{\alpha_j})^{i_j}   \}$ for 0 $\leq i_j \leq p-1$,\newline

 where we put \newline

$A_{\epsilon_1}= x_{\epsilon_1}$, \newline

$A_{-\epsilon_1}=c_{-\epsilon_1}+ (h_{\epsilon_1} +1)^2 +
4x_{-\epsilon_1}x_{\epsilon_1},$ \newline

 $A_{-\epsilon_1\pm \epsilon_2}= x_{\epsilon_1\pm \epsilon_2}(c_{-\epsilon_1\pm \epsilon_2}+x_{-\epsilon_1\pm \epsilon_2}x_{-(-\epsilon_1\pm \epsilon_2)}\pm x_{\pm \epsilon_2}x_{-(\pm \epsilon_2)}\pm x_{\epsilon_1 \pm \epsilon_2}x_{-(\epsilon_1 \pm \epsilon_2)})$,\newline

$A_{-\epsilon_1 \pm \epsilon_j}=x_{-\epsilon_2\pm \epsilon_j}(c_{-\epsilon_1 \pm \epsilon_j}+x_{(\pm \epsilon_j-\epsilon_1)}
x_{-(\pm \epsilon_j -\epsilon_1)}\pm x_{\pm\epsilon_j}x_{-(\pm \epsilon_j)}\pm x_{\epsilon_1\pm \epsilon_j}x_{-(\epsilon_1\pm \epsilon_j)})$,\newline

$A_{\pm\epsilon_2}=x_{\epsilon_3\pm \epsilon_2}^2 (c_{\pm\epsilon_2}+x_{\epsilon_2}x_{-\epsilon_2}\pm x_{\epsilon_1 +\epsilon_2}x_{-(\epsilon_1 +\epsilon_2)}\pm x_{\epsilon_2-\epsilon_1}x_{\epsilon_1-\epsilon_2})$,\newline

$A_{\epsilon_j}=x_{\epsilon_2+ \epsilon_j}(c_{\epsilon_j}+ x_{\epsilon_j}x_{-\epsilon_j}\pm x_{\epsilon_1+\epsilon_j}x_{-(\epsilon_1 + \epsilon_j)}\pm 
 x_{\epsilon_j- \epsilon_1}x_{\epsilon_1-\epsilon_j})$, \newline

$A_{-\epsilon_j}=x_{\epsilon_2-\epsilon_j}(c_{-\epsilon_j}+ x_{-\epsilon_j}x_{\epsilon_j}\pm x_{\epsilon_1- \epsilon_j}x_{-(\epsilon_1-\epsilon_j)}\pm x_{-\epsilon_j-\epsilon_1}x_{\epsilon_1+ \epsilon_j} )  $,\newline

with the sign chosen so that  they commute with $x_\alpha$ and with $c_\beta \in F$ chosen so that $A_{-\epsilon_1}$ and parentheses(             ) are invertible.\newline
For any other root $\beta$,  we put $A_\beta={x_\beta}^2$  or $ x_\beta^3 $ if possible.\newline
Otherwise we make use of the parentheses(      )again used for designating $A_{-\beta}$. So in this case we put $A_\beta = { x_\gamma}^2$ or ${x_\gamma}^3$ attached to these (        ) so that  $x_\alpha$ may commute with $A_\beta$.\newline

(II) Suppose that  $\alpha$ is a long root; then we may put $\alpha=\epsilon_1-\epsilon_2$ since all roots of the same length are conjugate under the Weyl group of $\Phi$.\newline

Similarly as in (I), we put $B_i:=$the same as in (I) except that  $\alpha=\epsilon_1 -\epsilon_2$ this time instead of $\epsilon_1$.\newline

In $u(L)/\frak M_\chi$ we have a basis $\frak B$$:= \{(B_1+ A_{\epsilon_1- \epsilon_2})^{i_1}\otimes (B_2+ A_{-(\epsilon_1-\epsilon_2)})^{i_2}\otimes \cdots \otimes(B_{2l-2}+ A_{-(\epsilon_{l-1}-\epsilon_l)})^{i_{2l-2}}\otimes (B_{2l-1}+ A_{\epsilon_l})^{i_{2l-1}}\otimes (B_{2l}+ A_{-\epsilon_l})^{i_{2l}}\otimes (\otimes_{j=2l+1}^{2m}(B_j+ A_{\alpha_j})^{i_j})| 0\leq i_ j\leq p-1 \} $,\newline

where we put  \newline

$A_{\epsilon_1-\epsilon_2}= x_\alpha= x_{\epsilon_1- \epsilon_2}$,\newline

$ A_{\epsilon_2- \epsilon_1}=c_{\epsilon_2- \epsilon_1}+ (h_{\epsilon_1- \epsilon_2}+ 1 )^2+ 4x_{\epsilon_2-\epsilon_1}x_{\epsilon_1- \epsilon_2}$\newline

$A_{\epsilon_2\pm \epsilon_3}=     x_{\pm \epsilon_3}(c_{\epsilon_2\pm \epsilon_3}+ x_{\epsilon_2\pm \epsilon_3}x_{-(\epsilon_2\pm \epsilon_3)}\pm x_{\epsilon_1\pm \epsilon_3}x_{-(\epsilon_1\pm \epsilon_3)})$, \newline

$A_{\epsilon_2\pm \epsilon_k}= x_{\epsilon_3\pm \epsilon_k}
(c_{\epsilon_2\pm \epsilon_k}+ x_{\epsilon_2\pm \epsilon_k}x_{-(\epsilon_2\pm \epsilon_k)}\pm x_{\epsilon_1\pm \epsilon_k}x_{-(\epsilon_1\pm \epsilon_k)})  $ \newline

if $k\neq 1$, \newline

$ A_{\epsilon_2}= x_{\epsilon_1}( c_{\epsilon_2}+ x_{\epsilon_2}x_{-\epsilon_2}\pm x_{\epsilon_1}x_{-\epsilon_1})$,  \newline

                 $A_{-\epsilon_1}= x_{-\epsilon_2}(c_{-\epsilon_1}+ x_{-\epsilon_1}x_{\epsilon_1}\pm x_{-\epsilon_2}x_{\epsilon_2})$,\newline

$A_{-(\epsilon_1\pm \epsilon_3)}= x_{-(\pm \epsilon_3)}(c_{-(\epsilon_1\pm \epsilon_3)}+ x_{\epsilon_2\pm \epsilon_3}x_{-(\epsilon_2\pm \epsilon_3)}\pm x_{\epsilon_1 \pm\epsilon_3}x_{-(\epsilon_1\pm \epsilon_3)}) $,\newline

 $A_{-(\epsilon_1\pm \epsilon_k)}= x_{-(\epsilon_3\pm \epsilon_k)}(c_{-(\epsilon_1\pm \epsilon_k)}+ x_{\epsilon_2\pm \epsilon_k }x_{-(\epsilon_2\pm \epsilon_k)}\pm x_{\epsilon_1\pm \epsilon_k}x_{-(\epsilon_1\pm \epsilon_k)}) $, \newline

$A_{\epsilon_l}= x_{\epsilon_l}^2 ,  $\newline

 $A_{-\epsilon_l}=x_{-\epsilon_l}^2  $, \newline

with the signs chosen so that they  may commute with $x_\alpha$ and with $c_\beta\in F$ chosen so that $A_{\epsilon_2-\epsilon_1}$ and parentheses are invertible.\newline
For any other root $\beta$, we put $A_\beta= x_\beta^3 $ or $x_\beta^4 $ if possible.\newline

Otherwise we make use of the parentheses(      ) again used for designating $A_{-\beta}$. So in this case we put $A_\beta= x_\gamma^2 $       
 or $ x_\gamma^3  $ attached to these (      ) so that  $x_\alpha$ may commute with $A_\beta$.\newline

\end{prop}

\begin{proof}

Refer to proposition2.1 in [KY-6] .

Actually we suggest  and  prove that the Lee's basis for this simple module is obtained  as follows.

We meet with 2 cases of root length.\newline

As suggested  in the  current proposition we  consider two cases separately.\newline

(I) Suppose that $\alpha$ is  a short root. Since all roots of a given length are conjugate under the Weyl group of $\Phi$, we may put $\alpha$= $\epsilon_1$ without loss of generality.\newline

Let us put $B_i$=$b_{i1}$$h_{\epsilon_1 -\epsilon_2}$ +\newline
$b_{i2}h_{\epsilon_2-\epsilon_3}$+ $\cdots$ +$b_{i,l-1}h_{\epsilon_{l-1}-\epsilon_l}$ + $b_{il}h_{\epsilon_l}$ for $i=1,2,\cdots,2m $ ,where ($b_{i1},\cdots,b_{il})
\in F^l$ are chosen so that  arbitrary ($l+1)-B_i$'s are linearly independent in $\mathbb P^l(F)$,the $\frak B$ below becomes  an $F$-linearly independent set in $u(L)$ if necessary and  $x_\alpha B_i$ $\not \equiv B_i x_\alpha $ with $\alpha =\epsilon_1$.\newline

Let $\frak M_\chi$ be the kernel of the irreducible representation $\rho _{\chi}: L  \rightarrow \frak {gl}(V) $ of the restricted Lie algebra ($L ,[p])$ associated  with any given irreducible $L$-module $V$ with a character $\chi$.\newline

In $u(L)/\frak M_\chi$ we give a basis as $\frak B$:=$\{(B_1 +A_{\epsilon_1})^{i_1}\otimes (B_2 + A_{-\epsilon_1})^{i_2}\otimes(B_3+ A_{\epsilon_1 -\epsilon_ 2})^{i_3}\otimes(B_4+ A_{-(\epsilon_1-\epsilon_2) })^{i_4}\otimes \cdots\otimes (B_{2l}+A_{-(\epsilon_{l-1}-\epsilon_l)})^{i_{2l}}\otimes (B_{2l+1}+ A_{\epsilon _l})^{i_{2l+1}}\otimes(B_{2l+2}+A_{-\epsilon_l})^{i_{2l+2}}\otimes_{j=2l+3}^{2m}(B_j+ A_{\alpha_j})^{i_j}   \}$ for 0 $\leq i_j \leq p-1$,\newline

 where we put \newline

$A_{\epsilon_1}= x_{\epsilon_1}$, \newline

$A_{-\epsilon_1}=c_{-\epsilon_1}+ (h_{\epsilon_1} +1)^2 +
4x_{-\epsilon_1}x_{\epsilon_1},$ \newline

$A_{-\epsilon_1\pm \epsilon_2}= x_{\epsilon_1\pm \epsilon_2}(c_{-\epsilon_1\pm \epsilon_2}+x_{-\epsilon_1\pm \epsilon_2}x_{-(-\epsilon_1\pm \epsilon_2)}\pm x_{\pm \epsilon_2}x_{-(\pm \epsilon_2)}\pm x_{\epsilon_1 \pm \epsilon_2}x_{-(\epsilon_1 \pm \epsilon_2)})$,\newline

$A_{-\epsilon_1 \pm \epsilon_j}=x_{-\epsilon_2\pm \epsilon_j}(c_{-\epsilon_1 \pm \epsilon_j}+x_{(\pm \epsilon_j-\epsilon_1)}
x_{-(\pm \epsilon_j -\epsilon_1)}\pm x_{\pm\epsilon_j}x_{-(\pm \epsilon_j)}\pm x_{\epsilon_1\pm \epsilon_j}x_{-(\epsilon_1\pm \epsilon_j)}),$\newline

$A_{\pm\epsilon_2}=x_{\epsilon_3\pm \epsilon_2}^2 (c_{\pm\epsilon_2}+x_{\epsilon_2}x_{-\epsilon_2}\pm x_{\epsilon_1 +\epsilon_2}x_{-(\epsilon_1 +\epsilon_2)}\pm x_{\epsilon_2-\epsilon_1}x_{\epsilon_1-\epsilon_2})$,\newline

$A_{\epsilon_j}=x_{\epsilon_2+ \epsilon_j}(c_{\epsilon_j}+ x_{\epsilon_j}x_{-\epsilon_j}\pm x_{\epsilon_1+\epsilon_j}x_{-(\epsilon_1 + \epsilon_j)}\pm \newline
 x_{\epsilon_j- \epsilon_1}x_{\epsilon_1-\epsilon_j})$, \newline

$A_{-\epsilon_j}=x_{\epsilon_2-\epsilon_j}(c_{-\epsilon_j}+ x_{-\epsilon_j}x_{\epsilon_j}\pm x_{\epsilon_1- \epsilon_j}x_{-(\epsilon_1-\epsilon_j)}\pm x_{-\epsilon_j-\epsilon_1}x_{\epsilon_1+ \epsilon_j} )  $,\newline

with the sign chosen so that  they commute with $x_\alpha$ and with $c_\beta \in F$ chosen so that $A_{-\epsilon_1}$ and parentheses(             ) are invertible.\newline
For any other root $\beta$,  we put $A_\beta={x_\beta}^2$  or $ x_\beta^3 $ if possible.\newline

Otherwise we make use of the parentheses(      )again used for designating $A_{-\beta}$. So in this case we put $A_\beta = { x_\gamma}^2$ or ${x_\gamma}^3$ attached to these (        ) so that  $x_\alpha$ may commute with $A_\beta$.\newline

(II)Suppose next  that  $\alpha$ is a long root; then we may put $\alpha=\epsilon_1-\epsilon_2$ since all roots of the same length are conjugate under the Weyl group of $\Phi$.\newline

Similarly as in (I), we put $B_i:=$the same as in (I) except that  $\alpha=\epsilon_1 -\epsilon_2$ this time instead of $\epsilon_1$.\newline

In $u(L)/\frak M_\chi$ we have a basis $\frak B$$:= \{(B_1+ A_{\epsilon_1- \epsilon_2})^{i_1}\otimes (B_2+ A_{-(\epsilon_1-\epsilon_2)})^{i_2}\otimes \cdots \otimes(B_{2l-2}+ A_{-(\epsilon_{l-1}-\epsilon_l)})^{i_{2l-2}}\otimes (B_{2l-1}+ A_{\epsilon_l})^{i_{2l-1}}\otimes (B_{2l}+ A_{-\epsilon_l})^{i_{2l}}\otimes (\otimes_{j=2l+1}^{2m}(B_j+ A_{\alpha_j})^{i_j})| 0\leq i_ j\leq p-1 \} $,\newline

where we put  \newline

$A_{\epsilon_1-\epsilon_2}= x_\alpha= x_{\epsilon_1- \epsilon_2}$,\newline

$ A_{\epsilon_2- \epsilon_1}=c_{\epsilon_2- \epsilon_1}+ (h_{\epsilon_1- \epsilon_2}+ 1 )^2+ 4x_{\epsilon_2-\epsilon_1}x_{\epsilon_1- \epsilon_2}$\newline

$A_{\epsilon_2\pm \epsilon_3}=     x_{\pm \epsilon_3}(c_{\epsilon_2\pm \epsilon_3}+ x_{\epsilon_2\pm \epsilon_3}x_{-(\epsilon_2\pm \epsilon_3)}\pm x_{\epsilon_1\pm \epsilon_3}x_{-(\epsilon_1\pm \epsilon_3)})$, \newline

$A_{\epsilon_2\pm \epsilon_k}= x_{\epsilon_3\pm \epsilon_k}
(c_{\epsilon_2\pm \epsilon_k}+ x_{\epsilon_2\pm \epsilon_k}x_{-(\epsilon_2\pm \epsilon_k)}\pm x_{\epsilon_1\pm \epsilon_k}x_{-(\epsilon_1\pm \epsilon_k)})  $ \newline

if $k\neq 1$, \newline

$ A_{\epsilon_2}= x_{\epsilon_1}( c_{\epsilon_2}+ x_{\epsilon_2}x_{-\epsilon_2}\pm x_{\epsilon_1}x_{-\epsilon_1})$,  \newline

                 $A_{-\epsilon_1}= x_{-\epsilon_2}(c_{-\epsilon_1}+ x_{-\epsilon_1}x_{\epsilon_1}\pm x_{-\epsilon_2}x_{\epsilon_2})$,\newline

$A_{-(\epsilon_1\pm \epsilon_3)}= x_{-(\pm \epsilon_3)}(c_{-(\epsilon_1\pm \epsilon_3)}+ x_{\epsilon_2\pm \epsilon_3}x_{-(\epsilon_2\pm \epsilon_3)}\pm x_{\epsilon_1 \pm\epsilon_3}x_{-(\epsilon_1\pm \epsilon_3)}) $,\newline

 $A_{-(\epsilon_1\pm \epsilon_k)}= x_{-(\epsilon_3\pm \epsilon_k)}(c_{-(\epsilon_1\pm \epsilon_k)}+ x_{\epsilon_2\pm \epsilon_k }x_{-(\epsilon_2\pm \epsilon_k)}\pm x_{\epsilon_1\pm \epsilon_k}x_{-(\epsilon_1\pm \epsilon_k)}) $, \newline

$A_{\epsilon_l}= x_{\epsilon_l}^2 ,  $\newline

 $A_{-\epsilon_l}=x_{-\epsilon_l}^2  $, \newline

with the signs chosen so that they  may commute with $x_\alpha$ and with $c_\beta\in F$ chosen so that $A_{\epsilon_2-\epsilon_1}$ and parentheses are invertible.\newline
For any other root $\beta$, we put $A_\beta= x_\beta^3 $ or $x_\beta^4 $ if possible.\newline

Otherwise we make use of the parentheses(      ) again used for designating $A_{-\beta}$. So in this case we put $A_\beta= x_\gamma^2 $       
 or $ x_\gamma^3  $ attached to these (      ) so that  $x_\alpha$ may commute with $A_\beta$.\newline
The rest of the complete proof is already shown in [KY-4].
However we may actually exhibit  the complete proof  as follows.\newline

We may see without difficulty that $\frak B$ is a linearly independent set in $u(L)$ by virtue of P-B-W theorem.\newline

We shall prove that a nontrivial linearly dependent  equation leads to absurdity. We assume first that we have a dependence equation which is of least degree with respect to $h_{\alpha_j}\in H$ and the number of whose highest  degree terms is also least.\newline

In case it is conjugated by $x_\alpha$, then there arises a nontrivial dependence equation of lower degree than the given one,which contradicts to our assumption.\newline

Otherwise  it reduces to the following form, which we have only to prove:\newline

 $(\ast) x_{\epsilon_{1}-\epsilon_{2}}$$K$ + $K'$ $\in$$ \mathfrak{M}_\chi$ or $x_{\epsilon_1}K+ K' \in \mathfrak {M}_\chi$,

respectively, where $K$ and $K'$  commute with $x_{\alpha}$ modulo $\frak {M}_\chi $.\newline

We have only to prove  one of these two cases, the other case being completely similar to the one. We thus deal with and assume  $x_{\epsilon_{2}-\epsilon_{1}}$$x_{\epsilon_{1}-\epsilon_{2}}$$K$ + $x_{\epsilon_{2}-\epsilon_{1}}$$K'$ $\in$ $\mathfrak{M}_\chi$.\newline

From $w_{\epsilon_1- \epsilon_2}:=(h_{\epsilon_1- \epsilon_2}+ 1)^2 +4 x_{\epsilon_2- \epsilon_1}x_{\epsilon_1- \epsilon_2}\in $ the center of  $u(\frak{sl}_2(F)), $ we get $ 4^{-1}\{w_{\epsilon_1- \epsilon_2}- (h+ 1)^2\}K+ x_{\epsilon_2- \epsilon_1}K' \equiv 0 $  modulo $\frak M_\chi$.\newline

If $x_{\epsilon_2- \epsilon_1}^p\equiv c $ which is a constant,then \newline

$(\ast \ast)4^{-1}x_{\epsilon_2-\epsilon_1}^{p-1}\{w_{\epsilon_1- \epsilon_2}- (h_{\epsilon_1- \epsilon_2}+ 1)^2\}K+ cK'\equiv 0 $\newline

 is obtained.
\newline
From $(\ast),(\ast \ast)$, we have\newline

 $4^{-1}x_{\epsilon_2- \epsilon_1}^{p-1}\{w_{\epsilon_1- \epsilon_2}- (h_{\epsilon_1- \epsilon_2}+ 1)^2\}K- cx_{\epsilon_1- \epsilon_2}K\newline
\equiv 0$ modulo $\frak M_\chi$.\newline

Multiplying $x_{\epsilon_1- \epsilon_2}^{p-1}$ to this equation,we obtain \newline

 $(\ast \ast \ast)4^{-1}x_{\epsilon_1- \epsilon_2}^{p-1}x_{\epsilon_2- \epsilon_1}^{p-1}\{w_{\epsilon_1- \epsilon_2}- (h_{\epsilon_1- \epsilon_2}+ 1)^2\}K- cx_{\epsilon_1- \epsilon_2}^pK\equiv 0.  $\newline

By making use of $w_{\epsilon_1- \epsilon_2}$, we may deduce from $(\ast \ast \ast)$ an equation of the form \newline
( a polynomial of degree $\geq 1$ with respect to  $ h_{\epsilon_1- \epsilon_2})K- cx_{\epsilon_1- \epsilon_2}^pK\equiv 0.   $\newline

Finally if we use conjugation and subtraction consecutively,then we are led to a  contradiction $K\in \frak M_\chi.$\newline
It may be necessary  for us to change the role of $K$ and $K'$ to obtain the absurdity alike.\newline

\end{proof}

\begin{prop}

Let $\chi $ be a character of any simple $L$-module with $\chi(h_\alpha )\neq $ 0  for  some $\alpha \in \Phi$,where $h_\alpha$ is an element in the Chevalley basis of  $L$ such that $F x_\alpha + Fh_\alpha+ Fx_{-\alpha}= \frak {sl}_2(F)$ with $[x_\alpha,x_{-\alpha}]= h_\alpha \in H$.\newline

 We then  claim  that any simple $L$-module with character $\chi$ is of dimension  $p^m=p^{n-l\over 2}$,where $n= dim L= 2m +l $ for a CSA  H with $dim H =l$ and  that  the Lee's basis for this simple module is given as follows.

Let $ \mathfrak {M}_\chi$ be the kernel of this irreducible representation; so $\frak M_\chi$ is  a certain (2-sided) maximal ideal of $u(L).$ \newline

(I) Assume first that $\alpha$ is a  short root; then we may put $\alpha$=$\epsilon_1$  without loss of generaity since all roots of a given length 
are conjugate under the Weyl group of the root system $\Phi$.\newline

If $\chi (x_{\epsilon_1}) \neq 0$ or $\chi (x_{- \epsilon_1}) \neq 0$, then our assertion is evident from the preceding proposition 4.4.\newline

So we may let $x_{\epsilon_1}^p \equiv x_{- \epsilon_1}^p \equiv 0$ modulo $\frak M_\chi$.\newline

 We let  $B_i$:=$b_{i1}$ $h_{\epsilon_{1}- \epsilon_{ 2}}$+ $b_{i2}$ $h_{\epsilon_{2}-\epsilon_{3}}$+$\cdots$+$b_{i,l-1}$ $h_{\epsilon_{l-1}-\epsilon_{ l }}$ + $b_{il}$ $h_{\epsilon_{l}}$  for $i=  1,2,$ $\cdots,2m,$ where ($b_{i1}$,$b_{i2}$   $\cdots,$$b_{il}$) $\in F^{l}$ are chosen so that any $ (l+1)-$$ B_i$'s  are linearly independent in $\mathbb {P}^l$ and  $\frak B$   below  becomes an $F-$ linearly independent  set in $U(L)$ if necessary and $x_\alpha$$B_i$ $\not\equiv$$B_i$$x_\alpha$ for $\alpha= \epsilon_1$

In  $u(L)$/$\mathfrak{M}_\chi$ we claim that we have a basis \newline

$\frak B$:= 
$\{(B_1 +A_{\epsilon_1})^{i_1}\otimes(B_2 +A_{-\epsilon_{1}})^{i_2}(B_3+ A_{\epsilon_1- \epsilon_2})^{i_3}\otimes(B_4+ A_{-\{\epsilon_1- \epsilon_2)})^{i_4}   \otimes \cdots \otimes(B_{2l} +A_{-(\epsilon_{l-1}-\epsilon_{l})})^{i_{2l}}\otimes(B_{2l+1}+ A_{\epsilon_{l}})^{i_{2l+1}}\otimes(B_{2l+2}+A_{-\epsilon_{l}})^{i_{2l+2}}\otimes(\otimes_{j=2l+3}^{2m}(B_j+A_{\alpha_{j}})^{i_{j}}) | 0 \leq i_{j}\leq p-1 \}$, \newline

where we put\newline

$A_{\epsilon_{1}}$= $g_\alpha$=
$g_{\epsilon_{1}},$ \newline

$A_{-\epsilon_{1}}$=$c_{-\epsilon_{1}}$ +$(h_{\epsilon_{1}}+1)^2$ +4$x_{-\alpha}$ $x_\alpha,$ \newline

$A_{-\epsilon_1\pm \epsilon_2}= x_{\epsilon_3\pm \epsilon_2}(c_{-\epsilon_{1}\pm \epsilon_{2}}+x_{-\epsilon_{1}\pm \epsilon_{2}}x_{\epsilon_1\mp \epsilon_2}\pm  x_{\pm \epsilon_{2}}x_{-(\pm \epsilon_{2})}\pm x_{\epsilon_1\pm \epsilon_2}x_{-(\epsilon_1\pm \epsilon_2)}),$\newline

$A_{\epsilon_1+ \epsilon_2}$=$x_{\epsilon
_3-\epsilon_2}^2$ $(c_{\epsilon_1+\epsilon_2}+ x_{-\epsilon_1-\epsilon_2}x_{\epsilon_1+\epsilon_2}\pm x_{-\epsilon_2}x_{\epsilon_2}\pm  x_{\epsilon_1- \epsilon_2}x_{-(\epsilon_1-\epsilon_2)})$,\newline

$A_{-\epsilon_1\pm \epsilon_j}= x_{-\epsilon_2\pm \epsilon_j}(c_{-\epsilon_1 \pm \epsilon_j}+x_{-\epsilon_1 \pm \epsilon_j}x_{\epsilon_1\mp \epsilon_j}\pm  x_{\pm \epsilon_j}x_{-(\pm \epsilon_j)}\pm x_{\epsilon_1\pm \epsilon_j}x_{-(\epsilon_1\pm \epsilon_j)}),$ \newline

$A_{\epsilon_1+ \epsilon_j}$=$x_{-\epsilon
_2-\epsilon_j}^2 (c_{\epsilon_1+\epsilon_j}+ x_{-\epsilon_i-\epsilon_1}x_{\epsilon_1+\epsilon_j}\pm x_{-\epsilon_j}x_{\epsilon_j}\pm  x_{\epsilon_1- \epsilon_j}x_{-(\epsilon_1-\epsilon_j)})$,\newline

$A_{\pm\epsilon_{2}}$=$x_{\epsilon_3\pm  \epsilon_2}^2 (c_{\pm\epsilon_{2}}+x_{\epsilon_{2}}x_{-\epsilon_{2}}\pm x_{\epsilon_{1}+ \epsilon_{2}}x_{-\epsilon_{1}-\epsilon_{2}} +x_{\epsilon_2- \epsilon_{1}}x_{\epsilon_1- \epsilon_2}),$ \newline

$A_{\pm \epsilon_j}$=$x_{\epsilon
_2\pm \epsilon_j} (c_{\pm \epsilon_j}+ x_{\epsilon_{j}}x_{-\epsilon_j}\pm x_{\epsilon_1+ \epsilon_j}x_{-\epsilon_1- \epsilon_j}\pm  x_{\epsilon_j- \epsilon_1}x_{\epsilon_1-\epsilon_j})$ ,\newline

with the sign chosen so that they commute with $x_{\alpha}$ and with $c_{\alpha}\in F$ 
chosen so that $A_{-\epsilon_{1}}$ and parentheses are invertible.  For any other root $ \beta$ we put $A_{\beta}$= $x_{\beta}^2 $ or $x_{\beta}^3 $ or $x_{\beta}^4$ if possible. \newline

Otherwise attach to these sorts the parentheses(        ) used for designating $A_{-\beta}$ so that  $A_\gamma  \forall \gamma \in \Phi$ may commute with $x_\alpha$.\newline

(II) Next  we assume that  $\alpha $ is a long root; we may  thus let   $\alpha=  \epsilon_1- \epsilon_2 $  without loss of generality.\newline

As we  have  done before, we let   $B_i $ be defined as in (I).\newline

 We have a Lee's basis in $u(L)/\frak M_\chi$ as $\frak B$$:= \{(B_1+ A_{\epsilon_1- \epsilon_2})^{i_1}\otimes (B_2+ A_{-(\epsilon_1-\epsilon_2)})^{i_2}\otimes \cdots \otimes(B_{2l-2}+ A_{-(\epsilon_{l-1}-\epsilon_l)})^{i_{2l-2}}\otimes (B_{2l-1}+ A_{\epsilon_l})^{i_{2l-1}}\otimes (B_{2l}+ A_{-\epsilon_l})^{i_{2l}}\otimes (\otimes_{j=2l+1}^{2m}(B_j+ A_{\alpha_j})^{i_j})| 0\leq i_ j\leq p-1 \} $,\newline

where we put  \newline

$A_{\epsilon_1-\epsilon_2}= g_\alpha= g_{\epsilon_1- \epsilon_2}$,\newline

 $A_{\epsilon_2- \epsilon_1}=g_\alpha^2 \{c_{\epsilon_2- \epsilon_1}+ (h_{\epsilon_1- \epsilon_2}+ 1 )^2+ 4x_{\epsilon_2-\epsilon_1}x_{\epsilon_1- \epsilon_2}\} $,\newline

$A_{\epsilon_1\pm \epsilon_3}= g_\alpha^3(c_{\epsilon_1 \pm \epsilon_3}+ x_{\epsilon_2\pm \epsilon_3}x_{-(\epsilon_2 \pm \epsilon_3)}                 \pm x_{\epsilon_1 \pm \epsilon_3}x_{-(\epsilon_1 \pm \epsilon_3)}) $,\newline

$A_{-\epsilon_1- \epsilon_3}= g_\alpha^4(c_{-\epsilon_1- \epsilon_3}+ x_{\epsilon_1+ \epsilon_3}x_{-\epsilon_1- \epsilon_3}\pm x_{\epsilon_2+ \epsilon_3}x_{-\epsilon_2- \epsilon_3})$,\newline

$A_{\epsilon_3- \epsilon_1}= (c_{\epsilon_3- \epsilon_1}+ x_{\epsilon_2- \epsilon_3}x_{\epsilon_3- \epsilon2}\pm x_{\epsilon_1- \epsilon_3}x_{\epsilon_3- \epsilon_1})    $\newline

$ A_{\epsilon_2}= x_{\epsilon_3}( c_{\epsilon_2}+ x_{\epsilon_2- \epsilon_3}x_{\epsilon_3-\epsilon_2}\pm x_{\epsilon_1- \epsilon_3}x_{\epsilon_3-\epsilon_1})$,  \newline

 $A_{-\epsilon_2}= x_{-\epsilon_3 }^2(c_{-\epsilon_2}+ x_{\epsilon_2- \epsilon_3}x_{\epsilon_3- \epsilon_2}\pm x_{\epsilon_1- \epsilon_3}x_{\epsilon_3- \epsilon_1})  $\newline

 $A_{\epsilon_1}= x_{\epsilon_3}^2 (c_{\epsilon_1}+  x_{\epsilon_2- \epsilon_3}x_{\epsilon_3- \epsilon_2}\pm x_{-\epsilon_1- \epsilon_3}x_{\epsilon_3-  \epsilon_1}) $, \newline

 $A_{-\epsilon_1}= x_{-\epsilon_3}(c_{-\epsilon_1}+ x_{\epsilon_2- \epsilon_3}x_{\epsilon_3- \epsilon_2}\pm x_{\epsilon_1- \epsilon_3}x_{\epsilon_3- \epsilon_1}) $,\newline

$A_{\epsilon_2\pm \epsilon_j}= x_{-\epsilon_j}^{1 or 2}(c_{\epsilon_2\pm \epsilon_j}+ x_{\epsilon_2\pm \epsilon_j}x_{-(\epsilon_2\pm \epsilon_j)}\pm x_{\epsilon_1\pm \epsilon_j}x_{-(\epsilon_1\pm \epsilon_j)})   $,\newline

$A_{-\epsilon_2- \epsilon_j}= x_{-\epsilon_j}^3(c_{-\epsilon_2- \epsilon_j}+ x_{\epsilon_2+ \epsilon_j}x_{-\epsilon_2- \epsilon_j}\pm x_{\epsilon_1+ \epsilon_j}x_{-\epsilon_1- \epsilon_j})            $,\newline

$A_{\epsilon_j-\epsilon_2}= x_{\epsilon_j}(c_{\epsilon_j- \epsilon_2}+x_{\epsilon_2- \epsilon_j}x_{\epsilon_j- \epsilon_2}\pm x_{\epsilon_1- \epsilon_j}x_{\epsilon_j- \epsilon_1})$,\newline

$A_{\epsilon_1\pm \epsilon_j}= x_{\epsilon_3- \epsilon_j}^{1 or 2}(c_{\epsilon_1\pm \epsilon_j}+ x_{\epsilon_2\pm \epsilon_j}x_{-(\epsilon_2\pm \epsilon_j}\pm x_{\epsilon_1\pm \epsilon_j}x_{-(\epsilon_1\pm \epsilon_j)})$,\newline

$A_{-\epsilon_1- \epsilon_j}= x_{\epsilon_3- \epsilon_j}^3(c_{-\epsilon_1- \epsilon_j}+ x_{\epsilon_2+ \epsilon_j}x_{-\epsilon_2- \epsilon_j}\pm x_{\epsilon_1+ \epsilon_j}x_{-\epsilon_1- \epsilon_j}) $,\newline

$A_{\epsilon_j- \epsilon_1}= x_{\epsilon_j- \epsilon_3}(c_{\epsilon_j- \epsilon_1}+ x_{\epsilon_2- \epsilon_j}x_{\epsilon_j- \epsilon_2}\pm x_{\epsilon_1- \epsilon_j}x_{\epsilon_j- \epsilon_1})$,\newline

with the signs chosen so that they may commute with $x_\alpha$ and with $c_\beta\in F$ chosen so that $A_{\epsilon_2-\epsilon_1}$ and parentheses are invertible.\newline

For any other root $\beta$, we put $A_\beta= x_\beta^3 $ or $x_\beta^4 $ if possible.\newline

Otherwise we make use of the parentheses(      ) again used for designating $A_{-\beta}$. So in this case we put $A_\beta= x_\gamma^2 $       
 or $ x_\gamma^3  $ attached to these (      ) so that  $x_\alpha$ may commute with $A_\beta$.\newline

\end{prop}

\begin{proof}

Obvious from chapter5 in [KY-4].
Actually  we may prove  this proposition as follows.

Let $ \mathfrak {M}_\chi$ be the kernel of the corresponding  irreducible representation; so $\frak M_\chi$ is  a certain (2-sided) maximal ideal of $U(L).$ \newline

As suggested in the current  proposition, we consider separately two cases as follows.

(I) Assume first that $\alpha$ is a  short root; then we may put $\alpha$=$\epsilon_1$  without loss of generaity since all roots of a given length 
are conjugate under the Weyl group of the root system $\Phi$.\newline

If $\chi (x_{\epsilon_1}) \neq 0$ or $\chi (x_{- \epsilon_1}) \neq 0$, then our assertion is evident from proposition 4.4  above.\newline

So we may let $x_{\epsilon_1}^p \equiv x_{- \epsilon_1}^p \equiv 0$ modulo $\frak M_\chi$.\newline

 We let  $B_i$:=$b_{i1}$ $h_{\epsilon_{1}- \epsilon_{ 2}}$+ $b_{i2}$ $h_{\epsilon_{2}-\epsilon_{3}}$+$\cdots$+$b_{i,l-1}$ $h_{\epsilon_{l-1}-\epsilon_{ l }}$ + $b_{il}$ $h_{\epsilon_{l}}$  for $i=  1,2,$ $\cdots,2m,$ where ($b_{i1}$,$b_{i2}$   $\cdots,$$b_{il}$) $\in F^{l}$ are chosen so that any $ (l+1)-$$ B_i$'s  are linearly independent in $\mathbb {P}^l$ and  $\frak B$   below  becomes an $F-$ linearly independent  set in $U(L)$ if necessary and $x_\alpha$$B_i$ $\not\equiv$$B_i$$x_\alpha$ for $\alpha= \epsilon_1$

In  $u(L)$/$\mathfrak{M}_\chi$ we claim that we have a basis \newline

$\frak B$:= 
$\{(B_1 +A_{\epsilon_1})^{i_1}\otimes(B_2 +A_{-\epsilon_{1}})^{i_2}(B_3+ A_{\epsilon_1- \epsilon_2})^{i_3}\otimes(B_4+ A_{-\{\epsilon_1- \epsilon_2)})^{i_4}   \otimes \cdots \otimes(B_{2l} +A_{-(\epsilon_{l-1}-\epsilon_{l})})^{i_{2l}}\otimes(B_{2l+1}+ A_{\epsilon_{l}})^{i_{2l+1}}\otimes(B_{2l+2}+A_{-\epsilon_{l}})^{i_{2l+2}}\otimes(\otimes_{j=2l+3}^{2m}(B_j+A_{\alpha_{j}})^{i_{j}}) | 0 \leq i_{j}\leq p-1 \}$, \newline

where we put\newline

$A_{\epsilon_{1}}$= $g_\alpha$=
$g_{\epsilon_{1}},$ \newline

$A_{-\epsilon_{1}}$=$c_{-\epsilon_{1}}$ +$(h_{\epsilon_{1}}+1)^2$ +4$x_{-\alpha}$ $x_\alpha,$ \newline

$A_{-\epsilon_1\pm \epsilon_2}= x_{\epsilon_3\pm \epsilon_2}(c_{-\epsilon_{1}\pm \epsilon_{2}}+x_{-\epsilon_{1}\pm \epsilon_{2}}x_{\epsilon_1\mp \epsilon_2}\pm  x_{\pm \epsilon_{2}}x_{-(\pm \epsilon_{2})}\pm x_{\epsilon_1\pm \epsilon_2}x_{-(\epsilon_1\pm \epsilon_2)}),$\newline

$A_{\epsilon_1+ \epsilon_2}$=$x_{\epsilon
_3-\epsilon_2}^2$ $(c_{\epsilon_1+\epsilon_2}+ x_{-\epsilon_1-\epsilon_2}x_{\epsilon_1+\epsilon_2}\pm x_{-\epsilon_2}x_{\epsilon_2}\pm  x_{\epsilon_1- \epsilon_2}x_{-(\epsilon_1-\epsilon_2)})$,\newline

$A_{-\epsilon_1\pm \epsilon_j}= x_{-\epsilon_2\pm \epsilon_j}(c_{-\epsilon_1 \pm \epsilon_j}+x_{-\epsilon_1 \pm \epsilon_j}x_{\epsilon_1\mp \epsilon_j}\pm  x_{\pm \epsilon_j}x_{-(\pm \epsilon_j)}\pm x_{\epsilon_1\pm \epsilon_j}x_{-(\epsilon_1\pm \epsilon_j)}),$ \newline

$A_{\epsilon_1+ \epsilon_j}$=$x_{-\epsilon
_2-\epsilon_j}^2 (c_{\epsilon_1+\epsilon_j}+ x_{-\epsilon_i-\epsilon_1}x_{\epsilon_1+\epsilon_j}\pm x_{-\epsilon_j}x_{\epsilon_j}\pm  x_{\epsilon_1- \epsilon_j}x_{-(\epsilon_1-\epsilon_j)})$,\newline

$A_{\pm\epsilon_{2}}$=$x_{\epsilon_3\pm  \epsilon_2}^2 (c_{\pm\epsilon_{2}}+x_{\epsilon_{2}}x_{-\epsilon_{2}}\pm x_{\epsilon_{1}+ \epsilon_{2}}x_{-\epsilon_{1}-\epsilon_{2}} +x_{\epsilon_2- \epsilon_{1}}x_{\epsilon_1- \epsilon_2}),$ \newline

$A_{\pm \epsilon_j}$=$x_{\epsilon
_2\pm \epsilon_j} (c_{\pm \epsilon_j}+ x_{\epsilon_{j}}x_{-\epsilon_j}\pm x_{\epsilon_1+ \epsilon_j}x_{-\epsilon_1- \epsilon_j}\pm  x_{\epsilon_j- \epsilon_1}x_{\epsilon_1-\epsilon_j})$ ,\newline

with the sign chosen so that they commute with $x_{\alpha}$ and with $c_{\alpha}\in F$ 
chosen so that $A_{-\epsilon_{1}}$ and parentheses are invertible.  For any other root $ \beta$ we put $A_{\beta}$= $x_{\beta}^2 $ or $x_{\beta}^3 $ or $x_{\beta}^4$ if possible. \newline

Otherwise attach to these sorts the parentheses(        ) used for designating $A_{-\beta}$ so that  $A_\gamma  \forall \gamma \in \Phi$ may commute with $x_\alpha$.\newline

(II) Next  we assume that  $\alpha $ is a long root; we may  thus let   $\alpha=  \epsilon_1- \epsilon_2 $  without loss of generality.\newline

As we  have  done before, we let   $B_i $ be defined as in (I).\newline

 We have a Lee's basis in $u(L)/\frak M_\chi$ as $\frak B$$:= \{(B_1+ A_{\epsilon_1- \epsilon_2})^{i_1}\otimes (B_2+ A_{-(\epsilon_1-\epsilon_2)})^{i_2}\otimes \cdots \otimes(B_{2l-2}+ A_{-(\epsilon_{l-1}-\epsilon_l)})^{i_{2l-2}}\otimes (B_{2l-1}+ A_{\epsilon_l})^{i_{2l-1}}\otimes (B_{2l}+ A_{-\epsilon_l})^{i_{2l}}\otimes (\otimes_{j=2l+1}^{2m}(B_j+ A_{\alpha_j})^{i_j})| 0\leq i_ j\leq p-1 \} $,\newline

where we put  \newline

$A_{\epsilon_1-\epsilon_2}= g_\alpha= g_{\epsilon_1- \epsilon_2}$,\newline

 $A_{\epsilon_2- \epsilon_1}=g_\alpha^2 \{c_{\epsilon_2- \epsilon_1}+ (h_{\epsilon_1- \epsilon_2}+ 1 )^2+ 4x_{\epsilon_2-\epsilon_1}x_{\epsilon_1- \epsilon_2}\} $,\newline

$A_{\epsilon_1\pm \epsilon_3}= g_\alpha^3(c_{\epsilon_1 \pm \epsilon_3}+ x_{\epsilon_2\pm \epsilon_3}x_{-(\epsilon_2 \pm \epsilon_3)}                 \pm x_{\epsilon_1 \pm \epsilon_3}x_{-(\epsilon_1 \pm \epsilon_3)}) $,\newline

$A_{-\epsilon_1- \epsilon_3}= g_\alpha^4(c_{-\epsilon_1- \epsilon_3}+ x_{\epsilon_1+ \epsilon_3}x_{-\epsilon_1- \epsilon_3}\pm x_{\epsilon_2+ \epsilon_3}x_{-\epsilon_2- \epsilon_3})$,\newline

$A_{\epsilon_3- \epsilon_1}= (c_{\epsilon_3- \epsilon_1}+ x_{\epsilon_2- \epsilon_3}x_{\epsilon_3- \epsilon2}\pm x_{\epsilon_1- \epsilon_3}x_{\epsilon_3- \epsilon_1})    $\newline

$ A_{\epsilon_2}= x_{\epsilon_3}( c_{\epsilon_2}+ x_{\epsilon_2- \epsilon_3}x_{\epsilon_3-\epsilon_2}\pm x_{\epsilon_1- \epsilon_3}x_{\epsilon_3-\epsilon_1})$,  \newline

 $A_{-\epsilon_2}= x_{-\epsilon_3 }^2(c_{-\epsilon_2}+ x_{\epsilon_2- \epsilon_3}x_{\epsilon_3- \epsilon_2}\pm x_{\epsilon_1- \epsilon_3}x_{\epsilon_3- \epsilon_1})  $\newline

$A_{\epsilon_1}= x_{\epsilon_3}^2 (c_{\epsilon_1}+  x_{\epsilon_2- \epsilon_3}x_{\epsilon_3- \epsilon_2}\pm x_{-\epsilon_1- \epsilon_3}x_{\epsilon_3-  \epsilon_1}) $, \newline

 $A_{-\epsilon_1}= x_{-\epsilon_3}(c_{-\epsilon_1}+ x_{\epsilon_2- \epsilon_3}x_{\epsilon_3- \epsilon_2}\pm x_{\epsilon_1- \epsilon_3}x_{\epsilon_3- \epsilon_1}) $,\newline

$A_{\epsilon_2\pm \epsilon_j}= x_{-\epsilon_j}^{1 or 2}(c_{\epsilon_2\pm \epsilon_j}+ x_{\epsilon_2\pm \epsilon_j}x_{-(\epsilon_2\pm \epsilon_j)}\pm x_{\epsilon_1\pm \epsilon_j}x_{-(\epsilon_1\pm \epsilon_j)})   $,\newline

$A_{-\epsilon_2- \epsilon_j}= x_{-\epsilon_j}^3(c_{-\epsilon_2- \epsilon_j}+ x_{\epsilon_2+ \epsilon_j}x_{-\epsilon_2- \epsilon_j}\pm x_{\epsilon_1+ \epsilon_j}x_{-\epsilon_1- \epsilon_j})            $,\newline

$A_{\epsilon_j-\epsilon_2}= x_{\epsilon_j}(c_{\epsilon_j- \epsilon_2}+x_{\epsilon_2- \epsilon_j}x_{\epsilon_j- \epsilon_2}\pm x_{\epsilon_1- \epsilon_j}x_{\epsilon_j- \epsilon_1})$,\newline

$A_{\epsilon_1\pm \epsilon_j}= x_{\epsilon_3- \epsilon_j}^{1 or 2}(c_{\epsilon_1\pm \epsilon_j}+ x_{\epsilon_2\pm \epsilon_j}x_{-(\epsilon_2\pm \epsilon_j}\pm x_{\epsilon_1\pm \epsilon_j}x_{-(\epsilon_1\pm \epsilon_j)})$,\newline

$A_{-\epsilon_1- \epsilon_j}= x_{\epsilon_3- \epsilon_j}^3(c_{-\epsilon_1- \epsilon_j}+ x_{\epsilon_2+ \epsilon_j}x_{-\epsilon_2- \epsilon_j}\pm x_{\epsilon_1+ \epsilon_j}x_{-\epsilon_1- \epsilon_j}) $,\newline

$A_{\epsilon_j- \epsilon_1}= x_{\epsilon_j- \epsilon_3}(c_{\epsilon_j- \epsilon_1}+ x_{\epsilon_2- \epsilon_j}x_{\epsilon_j- \epsilon_2}\pm x_{\epsilon_1- \epsilon_j}x_{\epsilon_j- \epsilon_1})$,\newline

with the signs chosen so that they may commute with $x_\alpha$ and with $c_\beta\in F$ chosen so that $A_{\epsilon_2-\epsilon_1}$ and parentheses are invertible.\newline

For any other root $\beta$, we put $A_\beta= x_\beta^3 $ or $x_\beta^4 $ if possible.\newline

Otherwise we make use of the parentheses(      ) again used for designating $A_{-\beta}$. So in this case we put $A_\beta= x_\gamma^2 $       
 or $ x_\gamma^3  $ attached to these (      ) so that  $x_\alpha$ may commute with $A_\beta$.\newline

We may see without difficulty that $\frak B$ is a linearly independent set in $u(L)$ by virtue of P-B-W theorem.\newline

We shall prove that a nontrivial linearly dependent  equation leads to absurdity. We assume first that we have a dependence equation which is of least degree with respect to $h_{\alpha_j}\in H$ and the number of whose highest  degree terms is also least.\newline

In case it is conjugated by $g_\alpha$, then there arises a nontrivial dependence equation of lower degree than the given one,which contradicts to our assumption.\newline

Otherwise  it reduces to one of the following forms, so that we have only to prove that \newline

(i)$x_{\beta}K+ K'\in \frak M_\chi$,

(ii) $g_\alpha K+ K'\in \frak M_\chi$ \newline

lead to a contradiction, where both $K$ and $K'$ commute with $x_{\pm \alpha}$  modulo $\frak M_\chi $. In particular $K$ commute with $g_\alpha$.
\newline

For the case (i), we may change it to the form $x_{\alpha}K+ K''\in \frak M_\chi$ for some $K''$ commuting with $x_\alpha$ modulo $\frak M_\chi$.\newline

So we have $x_\alpha^p K+ x_\alpha^{p-1}K''\equiv 0$, and hence $x_\alpha^{p-1}K''\equiv 0$.\newline

Subtracting from this $x_{-\alpha}x_\alpha K+ x_{-\alpha}K''\equiv 0$, we get \newline

$-x_{-\alpha}x_\alpha K+ g_\alpha K'' \equiv 0$. Recall here that $g_\alpha$ is invertible and $w_\alpha$ belongs to the center of $u(\frak {sl}_2 (F))$.

So we get  $4^{-1}\{(h_\alpha +1)^2- w_\alpha\}K+ g_\alpha K''\equiv 0$, and hence\newline

 $(\ast) g_\alpha^{p-1} 4^{-1}\{(h_\alpha + 1)^2- w_\alpha \}K+ cK'' \equiv 0$\newline

 is obtained and from the start equation we have \newline

$(\ast \ast)cx_\alpha K+ c K''\equiv 0$, where $g_\alpha^p- c \equiv 0$.\newline

Subtracting $(\ast \ast)$ from $(\ast)$, we have $4^{-1}g_\alpha^{p-1}\{(h_\alpha+ 1)^2- w_\alpha\}K- cx_\alpha K \equiv 0$.\newline

Multiplying this equation by $g_\alpha^{1-p}$ to the right, we obtain $4^{-1}g_\alpha^{p-1}\{(h_\alpha+ 1)^2- w_\alpha\}g_\alpha^{1-p}K- cx_\alpha g_\alpha^{1-p}K \equiv 0$ \newline

We thus have $4^{-1}\{(h_\alpha+ 1- 2)^2- w_\alpha\}K- x_\alpha g_\alpha K \equiv 0$.

So it follows that $4^{-1}\{(h_\alpha -1)^2- w_\alpha\}K+ x_\alpha x_{-\alpha}K \equiv 0 $.\newline

Next multiplying $x_{-\alpha}^{p-1}$ to the right of this last equation, we obtain $\{(h_{\alpha}- 1)^2- w_\alpha\}K x_{-\alpha}^{p-1}\equiv 0$.
Now multiply $x_\alpha$ in turn consecutively to the left of this equation until it becomes of the form \newline

( a nonzero polynomial of degree $\geq 1$ with respect to $h_\alpha)K $\newline
$\in \frak M_\chi$. \newline

By making use of  conjugation and subtraction consecutively, we are led to a contradiction.$K \in \frak M_\chi$.\newline

Finally for the case (ii),we consider $K+ g_\alpha^{-1}K' \in \frak M_\chi$.  So we have $x_\alpha K+ x_\alpha g_\alpha^{-1} K' \equiv 0$ modulo $\frak M_\chi$.

By analogy with the argument  as in the case (i), we obtain a contraiction $K \in \frak M_\chi$.

\end{proof}

Next we  note first that the root system of $C_l$-type Lie algebra over  $\mathbb C$ is just  $\Phi= \{\pm2 \epsilon_i,\pm (\epsilon_i \pm \epsilon_j)|1\leq i \neq j \leq l \geq 3\}$ with a base $\{\epsilon_1- \epsilon_2,\cdots ,\epsilon_{l-1}- \epsilon_l,2\epsilon_l \}, $ where $\epsilon_i$ and $\epsilon_j$ are linearly independent orthonormal unit vectors in $\mathbb R^l$. \newline

For a root $\alpha \in \Phi,$ we  also put $g_\alpha :=  x_\alpha^{p-1}- x_{- \alpha}$and $w_\alpha:= (h_\alpha+ 1)^2+ 4x_{-\alpha}x_\alpha$  as in section 2, where $[x_\alpha,x_{-\alpha}]= h_\alpha$.\newline

For an algebraically closed field $F$ of prime characteristic $p,$ the $C_l-$ type Lie algebra $L$ over $F$ is just the analogue over $F$ of the $C_l-$ type 
simple Lie algeba over $\mathbb C$.\newline

 In other words the $C_l-$ type Lie algebra over $F$ is isomorphic to the Chevalley Lie algebra of the form 
$\sum_{i=1}^{n}\mathbb{Z}c_i\otimes_\mathbb{Z}F,$ 
where $ n$= $dim_FL$ and $x_\alpha$= some $c_i$ for each $\alpha\in\Phi$ , $h_{\alpha}$= some $c_{j}$ with  $\alpha$  some base element  of  $\Phi$  for a Chevalley basis 
\{$c_{i}$\} of  the $C_{l}$ - type  Lie algebra over $\mathbb{C}$.\newline

We shall compute in this section the dimension of some simple modules of the $C_l$-type Lie algebra $L$ with a CSA  $H$ over an algebraically closed field $F$ of characteristic $p \geq 7$.\newline

Let $L$ be a  $C_l$-type simlpe Lie algebra over  $F$. Let $\chi$ be a character of any simple $L$-module with $\chi (x_\alpha)\neq 0$ for some $\alpha \in \Phi$,where $x_\alpha$ is an element in the Chevalley basis of  $L$ such that $F x_\alpha + Fh_\alpha+ Fx_{-\alpha}= \frak {sl}_2(F)$ with $[x_\alpha,x_{-\alpha}]= h_\alpha$.\newline

Then we have conjectured in [KY-6]  that any simple $L$-module with character $\chi$ is of dimension  $p^m=p^{n-l\over 2}$,where $n= dim L= 2m +l $ for a CSA  $ H $ with $dim H =l$. \newline

In this section we intend to actually  clarify this conjecture for modular $C_l$-type Lie algebra   $L$. \newline

\begin{prop}\label {thm3.1}

Let  $\alpha$  be any  root  in the root system  $\Phi$ of $L .$ If $\chi(x_\alpha)$ $\neq0,$ then $dim_F$$ \rho_\chi$$(u(L))$ = $p^{2m},$ where $ [Q(u(L)):Q(\mathfrak{Z})]$=$p^{2m}$=$p^{n-l}$ with $\mathfrak{Z}$ the center of $u(L)$  and  $Q$  denotes  the quotient algebra. \newline

So  we claim  that  the  simple module corresponding to this representation has $p^m$ as its  dimension.\newline

\end{prop}

\begin{proof}

Let $ \mathfrak {M}_\chi$ be the kernel of this irreducible representation,i.e., a certain (2-sided) maximal ideal of $u(L).$ \newline

(I) Assume first that $\alpha$ is a  short root; then we may put $\alpha$=$\epsilon_1-\epsilon_2$  without loss of generaity since all roots of a given length 
are conjugate under the Weyl group of the root system $\Phi$.\newline

 First we let \newline

 $ B_i$:=$b_{i1}$ $h_{\epsilon_{1}- \epsilon_{ 2}}$+ $b_{i2}$ $h_{\epsilon_{2}-\epsilon_{3}}$+$\cdots$+$b_{i,l-1}$ $h_{\epsilon_{l-1}-\epsilon_{ l }}$ + $b_{il}$ $h_{\epsilon_{l}}$  for $i=  1,2,$ $\cdots,2m,$ where ($b_{i1}$,$b_{i2}$   $\cdots,$$b_{il}$) $\in F^{l}$ are chosen so that any $ (l+1)-$$ B_i$'s  are linearly independent in $\mathbb{P}^l(F),$ the  $\frak B$   below  becomes an $F-$ linearly independent  set in $u(L)$ if necessary and $x_\alpha$$B_i$ $\not\equiv$$B_i$$x_\alpha$ for $\alpha$=$\epsilon_1-$$\epsilon_2.$ \newline

In  $u(L)$/$\mathfrak{M}_\chi$ we claim that we have a basis \newline

$\frak B$:= 
$\{(B_1 +A_{\epsilon_{1}-\epsilon_ { 2}})^{i_1}\otimes(B_2 +A_{-(\epsilon_{1}-\epsilon_ {2})})^{i_2}\otimes \cdots \otimes(B_{2l-2} +A_{-(\epsilon_{l-1}-\epsilon_{l})})^{i_{2l-2}}\otimes(B_{2l-1}+ A_{2\epsilon_{l}})^{i_{2l-1}}\otimes(B_{2l}+A_{-2\epsilon_{l}})^{i_{2l}}\otimes(\otimes_{j=2l+1}^{2m}(B_j+A_{\alpha_{j}})^{i_{j}}) | 0 \leq i_{j}\leq p-1 \}$, \newline

where we put\newline

$A_{\epsilon_{1}-\epsilon_{2}}$= $x_\alpha$=
$x_{\epsilon_{1}-\epsilon_{2}},$ \newline

$A_{\epsilon_{2}-\epsilon_{1}}$=$c_{-(\epsilon_{1}-\epsilon_{2})}$ +$(h_{\epsilon_{1}-\epsilon_{2}}+1)^2$ +4$x_\alpha$ $x_{-\alpha},$ \newline

$A_{{\epsilon_{2}}{\pm}\epsilon_{3}}$=$x_{\pm2\epsilon_{3}}$ $(c_{\epsilon_{2}\pm\epsilon_{3}}+ x_{\epsilon_{2}\pm\epsilon_{3}}x_{-(\epsilon_{2}\pm\epsilon_{3})}\pm x_{\epsilon_{1}\pm\epsilon_{3}}x_{-(\epsilon_{1}\pm\epsilon_{3})}),$ \newline

$A_{\epsilon_{1}+\epsilon_{2}}$=$x_{\epsilon_{1}-\epsilon_{2}}^2$ $(c_{\epsilon_{1}+\epsilon_{2}}+3x_{\epsilon_{1}+\epsilon_{2}}x_{-\epsilon_{1}-\epsilon_{2}}\pm 2x_{2\epsilon_{1}}x_{-2\epsilon_{1}}\pm 2 x_{2\epsilon_{2}}x_{-2\epsilon_{2}}),$\newline

 $A_{\epsilon_{2}\pm\epsilon_{k}}$=$x_{\epsilon_{3}\pm\epsilon_{k}}( c_{\epsilon_{2}\pm\epsilon_{k}}+x_{\epsilon_{2}\pm\epsilon_{k}}x_{-(\epsilon_{2}\pm\epsilon_{k})} \pm x_{\epsilon_{1}\pm\epsilon_{k}}x_{-(\epsilon_{1}\pm\epsilon_{k})} ),$ \newline

$A_{2\epsilon_{2}}$=$x_{2\epsilon_{3}}^2 (c_{2\epsilon_{2}}+2x_{2\epsilon_{2}}x_{-2\epsilon_{2}}\pm 3x_{\epsilon_{1}+ \epsilon_{2}}x_{-\epsilon_{1}-\epsilon_{2}} +2x_{2\epsilon_{1}}x_{-2\epsilon_{1}}),$ \newline

$A_{-2\epsilon_{1}}$=
$x_{-2\epsilon_{3}}^2 ( c_{-2\epsilon_{1}}+ 2x_{-2\epsilon_{1}}x_{2\epsilon_{1}}\pm3x_{-\epsilon_ {1}-\epsilon_{2}}x_{\epsilon_{1}+\epsilon_{2}}
\pm2x_{-2\epsilon_{2}}x_{2\epsilon_{2}}), $ \newline

$A_{-(\epsilon_{1}\pm\epsilon_{3})}$= $x_{-(\pm\epsilon_{3})}(c_{-(\epsilon_{2}\pm\epsilon_{3})}
+  x_{\epsilon_{2}\pm\epsilon_{3}}x_{-(\epsilon_{2}\pm\epsilon_{3})}\pm x_{\epsilon_{1}\pm\epsilon_{3}}x_{-(\epsilon_{1}\pm\epsilon_{3})}),$\newline

$A_{-(\epsilon_{1}\pm\epsilon_{k})}$= $x_{-(\epsilon_{3}\pm\epsilon_{k})}(c_{-(\epsilon_{1}\pm\epsilon_{k})}+ x_{\epsilon_{2}\pm\epsilon_{k}}x_{-(\epsilon_{2}\pm\epsilon_{k})}\pm x_{\epsilon_{1}\pm\epsilon_{k}}x_{-(\epsilon_{1}\pm\epsilon_{k})} ),$ \newline

$A_{2\epsilon_{l}}$= $x_{2\epsilon_{l}}^2, $ \newline

$A_{-2\epsilon_{l}}$= $x_{-2\epsilon_{l}}^2, $ 
\newline

with the sign chosen so that they commute with $x_{\alpha}$ and with $c_{\alpha}\in F$ 
chosen so that $A_{\epsilon_{2}-\epsilon_{1}}$ and parentheses are invertible.  For any other root $ \beta$ we put $A_{\beta}$= $x_{\beta}^2 $ or $x_{\beta}^3 $ if possible. Otherwise attach to these sorts the parentheses(        ) used for designating $A_{-\beta}$ so that  $A_\gamma  \forall \gamma \in \Phi$ may commute with $x_\alpha$.\newline

We shall prove that $\frak B$ is a basis in $u(L))$/$\mathfrak {M}_\chi$. \newline

By virtue of P-B-W theorem, it is not difficult to see that $\frak B$ is evidently a linearly independent set  over $F$ in $u(L)$. Furthermore $\forall$ $\beta$ $ \in\Phi$, $A_{\beta}\notin\mathfrak {M}_\chi$(see detailed proof below).\newline

We shall prove that a nontrivial linearly dependent equation leads to absurdity. We assume first that there is a dependence equation which is of least degree with respect to $h_{\alpha_{j}}\in H$ and the number of whose highest degree terms is also least. \newline

In case it is conjugated by $x_{\alpha}$, then there arises a nontrivial dependence equation of lower degree than the given one, which contradicts to our assumption.\newline

  Otherwise it reduces to one of the following forms:\newline

(i) $x_{2\epsilon_{j}}$$K$ + $K'$ $\in$ $\mathfrak{M}_\chi$ ,\newline

(ii) $x_{-2\epsilon_{j}}$$K$+ $K'$ $\in$ $\mathfrak{M}_\chi$ ,\newline

(iii)$x_{\epsilon_{j}+\epsilon_{k}}$ $K$ + $K'$$\in$ $\mathfrak{M}_\chi$,\newline

(iv)$x_{-\epsilon{j}-\epsilon_{k}}$$K$ + $K'$ $\in$ $\mathfrak{M}_\chi$,\newline

(v)$x_{\epsilon_{j}-\epsilon_{k}}$$K$ + $K'$ $\in$ $\mathfrak{M}_\chi$ ,\newline

where $K$, $K'$ commute with $x_{\alpha}$.\newline

For the case (i), we deduce successively \newline

$x_{\epsilon_{2}-\epsilon_{j}} x_{2\epsilon{j}}$$K$ + $x_{\epsilon_{2}-\epsilon{j}}$$K'$ $\in$ $\mathfrak{M}_\chi$

$\Rightarrow$ $x_{\epsilon_{2}+\epsilon_{j}}$$K$ + $x_{2\epsilon_{j}}$ $x_{\epsilon_{2}-\epsilon_{j}}$$K$ + $x_{\epsilon_{2}-\epsilon{j}}$$K'$
 $\in$ $\mathfrak{M}_\chi$ $\Rightarrow$($x_{\epsilon_{1}+\epsilon_{j}}$ or $x_{2\epsilon_{1}}$)$K$ + $x_{2\epsilon_{j}}$($x_{\epsilon_{1}-\epsilon_{j}}$ or $h_{\epsilon_{1}-\epsilon_{2}}$)$K$ +  ($x_{\epsilon_{1}-\epsilon_{j}}$ or $h_{\epsilon_{1}-\epsilon_{2}}$)$K'$ $\in$ $\mathfrak{M}_\chi$ \newline

by $adx_{\epsilon_{1}-\epsilon_{2}}$if $j$$\neq{1}$ or $j$=1 respectively, so that by successive $ adx_{\alpha}$ and rearrangement we get $x_{\epsilon_{1}\pm\epsilon_{j}}$$K+ K''$  $\in$ $\mathfrak{M}_\chi$ for some $K''$ commuting with $x_{\alpha}$ in view of  the start equation. So (i) reduces to (iii),(iv) or (v). \newline

Similarly as in (i)  and by adjoint operations , (ii) reduces to (iii),(iv) or (v). Also (iii),(iv) reduces to the form (v) putting $\epsilon_{j}$= -(-$\epsilon_{j}$), $\epsilon_{k}$= -(-$\epsilon_{k})$. \newline

Hence we have only to consider the case (v).
We consider\newline
 $x_{\epsilon_{k}-\epsilon_{2}}$ $x_{\epsilon_{j}-\epsilon_{k}}$ $K$+ $x_{\epsilon_{k}- \epsilon_{2}}$$K'$ $\in$ $\mathfrak{M}_\chi$ ,
so that ($x_{\epsilon_{j}-\epsilon_{2}}$+ $x_{\epsilon_{j}-\epsilon_{k}}$$x_{\epsilon_{k}-\epsilon_{2}}$)$K$ + $x_{\epsilon_{k}-\epsilon_{2}}$$K'$
$\in$ $\mathfrak{M}_\chi$ for $j,k$$\neq$1,2 . \newline

We thus have $x_{\epsilon_{j}-\epsilon_{2}}$$K$ + ($x_{\epsilon_{j}-\epsilon_{k}}$$x_{\epsilon_{k}-\epsilon_{2}}$ $K$ + $x_{\epsilon_{k}-\epsilon_{2}}$$K'$) $\in$ $\mathfrak{M}_\chi$, so that we may put this last (       )= another $K'$ alike as in the equation   (v).\newline

 Hence we need to show that  $x_{\epsilon_{j}-\epsilon_{2}}$$K$ + $K'$ $\in$ $\mathfrak{M}_\chi$ leads to absurdity.
We consider \newline

$x_{\epsilon_{2}-\epsilon_{j}}$$x_{\epsilon_{j}-\epsilon_{2}}$$K$ + $x_{\epsilon_{2}-\epsilon_{j}}$$K'$ $\in$ $\mathfrak{M}_\chi$
$\Rightarrow$ ($h_{\epsilon_{2}-\epsilon_{j}}+ x_{\epsilon_{j}-\epsilon_{2}}x_{\epsilon_{2}-\epsilon_{j}})K+  x_{\epsilon_{2}-\epsilon_{j}}K' \in \mathfrak{M}_\chi$ $\Rightarrow $    ($x_{\epsilon_{1}-\epsilon_{2}}$$\pm$$x_{\epsilon_{j}-\epsilon_{2}}$$x_{\epsilon_{1}-\epsilon_{j}}$)$K$ + $x_{\epsilon_{1}-\epsilon_{j}}$ $K'$$\in$ $\mathfrak{M}_\chi$  by $adx_{\epsilon_{1}-\epsilon_{2}}$ $\Rightarrow$ either $x_{\epsilon_{1}-\epsilon_{2}}$$K$ $\in$ $ \mathfrak{M}_\chi$ or  ( $x_{\epsilon_{1}-\epsilon_{2}}$ + $x_{\epsilon_{j}-\epsilon_{2}}$$x_{\epsilon_{1}-\epsilon_{j}}$)$K$+ $x_{\epsilon_{1}-\epsilon_{j}}$$K'$ $\in$ $\mathfrak{M}_\chi$ \newline

depending on [$x_{\epsilon_{j}-\epsilon_{2}}$, $x_{\epsilon_{1}-\epsilon_{j}}$]= +$x_{\epsilon_{1}-\epsilon_{2}}$  or    -$x_{\epsilon_{1}-\epsilon_{2}}$. The former case leads to $K$ $\in$ $\mathfrak{M}_\chi$, a contradiction. \newline

For the latter case

we consider \newline
$x_{\epsilon_{1}-\epsilon_{2}}$$K$ + ( $x_{\epsilon_{j}-\epsilon_{2}}$$x_{\epsilon_{1}-\epsilon_{j}}K$ + $x_{\epsilon_{1}-\epsilon_{j}}$$K'$)

$\in$ $\mathfrak{M}_\chi$.\newline

So we may put\newline

   $(\ast) x_{\epsilon_{1}-\epsilon_{2}}$$K$ + $K''$ $\in$$ \mathfrak{M}_\chi$,\newline

where $K''$=$x_{\epsilon_{j}-\epsilon_{2}}$$x_{\epsilon_{1}-\epsilon_{j}}$$K$+ $x_{\epsilon_{1}-\epsilon_{j}}$$K'$.
Thus $x_{\epsilon_{2}-\epsilon_{1}}$$x_{\epsilon_{1}-\epsilon_{2}}$$K$ + $x_{\epsilon_{2}-\epsilon_{1}}$$K''$ $\in$ $\mathfrak{M}_\chi$.
From $w_{\epsilon_1- \epsilon_2}:=(h_{\epsilon_1- \epsilon_2}+ 1)^2 +4 x_{\epsilon_2- \epsilon_1}x_{\epsilon_1- \epsilon_2}\in $ the center of  $u(\frak{sl}_2(F)), $ we get $ 4^{-1}\{w_{\epsilon_1- \epsilon_2}- (h+ 1)^2\}K+ x_{\epsilon_2- \epsilon_1}K'' \equiv 0 $  modulo $\frak M_\chi$.\newline

If $x_{\epsilon_2- \epsilon_1}^p\equiv c $ which is a constant,then \newline

$(\ast \ast)4^{-1}x_{\epsilon_2-\epsilon_1}^{p-1}\{w_{\epsilon_1- \epsilon_2}- (h_{\epsilon_1- \epsilon_2}+ 1)^2\}K+ cK''\equiv 0 $\newline

 is obtained.
\newline
From $(\ast),(\ast \ast)$, we have\newline

 $4^{-1}x_{\epsilon_2- \epsilon_1}^{p-1}\{w_{\epsilon_1- \epsilon_2}- (h_{\epsilon_1- \epsilon_2}+ 1)^2\}K- cx_{\epsilon_1- \epsilon_2}K\newline
\equiv 0$ modulo $\frak M_\chi$.\newline

Multiplying $x_{\epsilon_1- \epsilon_2}^{p-1}$ to this equation,we obtain \newline

 $(\ast \ast \ast)4^{-1}x_{\epsilon_1- \epsilon_2}^{p-1}x_{\epsilon_2- \epsilon_1}^{p-1}\{w_{\epsilon_1- \epsilon_2}- (h_{\epsilon_1- \epsilon_2}+ 1)^2\}K- cx_{\epsilon_1- \epsilon_2}^pK\equiv 0.  $\newline

By making use of $w_{\epsilon_1- \epsilon_2}$, we may deduce from $(\ast \ast \ast)$ an equation of the form \newline
( a polynomial of degree $\geq 1$ with respect to  $ h_{\epsilon_1- \epsilon_2})K- cx_{\epsilon_1- \epsilon_2}^pK\equiv 0.   $\newline

Finally if we use conjugation and subtraction consecutively,then we are led to a  contradiction $K\in \frak M_\chi.$\newline
Otherwise we may change the role of $K$ and $K'$ to obtain the absurdity alike.\newline

(II)Assume next that $\alpha$ is a long root; then we may put $\alpha=2 \epsilon_{1}$ because all roots of the same length are conjugate under the Weyl group of $\Phi$ .\newline

 Similarly as in (I), we let  $B_{i}$:= the same as in (I) except that this time  $\alpha=2\epsilon_{1}$ instead of $\epsilon_{1}-\epsilon_{2}$ . 
\newline

We claim that we have a basis in $u(L)/\frak M_\chi$ such as\newline

 $\frak B$:= $\{(B_{1}+ A_{2\epsilon_{1}})^{i_{1}}\otimes(B_{2} + A_{-2\epsilon_{1}})^{i_{2}}\otimes(B_{3}+ A_{\epsilon_{1}-\epsilon_{2}})^{i_{3}}\otimes(B_{4}+A_{-(\epsilon_{1}-\epsilon_{2})})^{i_{4}}\otimes\cdots \otimes(B_{2l}+ A_{-(\epsilon_{l-1}-\epsilon_{l})})^{i_{2l}}\otimes(B_{2l+1}+ A_{2\epsilon_{l}})^{i_{2l+ 1}}\otimes(B_{2l+ 2}+ A_{-2\epsilon_{l}})^{i_{2l+ 2}}\otimes(\otimes_{j=2l+ 3}^{2m}(B_{j}+ A_{\alpha_{j}})^{i_{j}}; 0 \leq i_{j}\leq p-1\}$ , \newline

where we put \newline

$A_{2\epsilon_{1}}$= $x_{2\epsilon_{1}}$, \newline

$A_{-2\epsilon_{1}}$= $c_{-2\epsilon_{1}}$+ $(h_{2\epsilon_{1}}+ 1)^{2}+ 4x_{-2\epsilon_{1}}x_{2\epsilon_{1}} $, \newline

$A_{-\epsilon_{1}\pm\epsilon_{2}}$=$ x_{-\epsilon_{3}\pm\epsilon_{2}}$$( c_{-\epsilon_{1}\pm\epsilon_{2}}$ $\pm$$x_{-\epsilon_{1}\pm\epsilon_{2}}$$x_{\epsilon_{1}\mp\epsilon_{2}}$$\pm$$x_{\epsilon_{1}\pm\epsilon_{2}}$$x_{-\epsilon_{1}\mp\epsilon_{2}})  $,\newline

$A_{-\epsilon_{!}\pm \epsilon_ {j}}= x_{-\epsilon_{2}\pm\epsilon_{j}}
( c_{-\epsilon_{1}\pm\epsilon_{j}} 
+ x_{\pm{\epsilon_{j}}-\epsilon_{1}}x_{\epsilon_{1}\mp\epsilon_{j}}$
$\pm$$x_{\epsilon_{1}\pm\epsilon_{j}}x_{-\epsilon_{1}\mp\epsilon_{j}})$ ,\newline

 and for any other root $\beta$ we put  $A_{\beta}= x_{\beta}^2 $ or $x_{\beta}^3 $  if possible.\newline

 Otherwise attach to these sorts the parentheses (           ) used for designating $A_{-\beta}$. 
 Likewise as in case (I),  we shall prove that $\frak B$ is a basis in $u(L)$/$\mathfrak{M}_\chi$. \newline

By virtue of P-B-W theorem, it is not difficult to see that $\frak B$ is evidently a linearly independent set over $F$ in $u(L)$. Moreover $\forall\beta$$\in$$ \Phi$, $A_{\beta}\notin \mathfrak{M}_\chi$(see detailed proof below).\newline

We shall prove that a nontrivial linearly dependent equation leads to absurdity. We assume first that there is a dependence equation which is of least degree with respect to $h_{\alpha_{j}}$ $\in$$H$ and the number of  whose highest degree terms is also least.\newline

 If it is conjugated by  $x_{\alpha}$, then there arises a nontrivial dependence equation of least degree than the given one,which contravenes our assumption. \newline

 Otherwise it reduces to one of the following forms: \newline

(i) $x_{2\epsilon_{j}}K + K'\in \mathfrak{M}_\chi$ ,\newline

(ii)  $x_{-2\epsilon_{j}}K + K' \in \mathfrak{M}_\chi$,\newline

(iii)$ x_{\epsilon_{j}+ \epsilon_{k}}K+ K' \in \mathfrak{M}_\chi$ ,\newline

(iv) $ x_{-\epsilon_{j}-\epsilon_{k}}K+ K' \in \mathfrak{M}_\chi$,\newline

(v) $x_{\epsilon_{j}-\epsilon_{k}}K + K' \in \mathfrak{M}_\chi$ ,\newline

where $K$ and $K'$ commute with $x_{\alpha}= x_{2\epsilon_{1}}$.\newline

For the case (i) , we consider  a particular case  $j$=1 first; if we assume $x_{2\epsilon_{1}}K+ K' \in \mathfrak{M}_\chi$, then we are led to a contradiction according to the similar argument ($\ast$) as in (I). \newline

So we assume  $x_{2\epsilon_{j}}K + K' \in \mathfrak{M}_\chi$ with $j\geq2$. Now we have  
$x_{2\epsilon_{j}}K+ K' \in \mathfrak{M}_\chi$ $\Rightarrow$ $x_{-\epsilon_{1}-\epsilon_{j}}x_{2\epsilon_{j}}K + x_{-\epsilon_{1}-\epsilon_{j}}K'\in 
\mathfrak{M}_\chi \Rightarrow $ $x_{-\epsilon_{1}+ \epsilon_{j}}K + x_{2\epsilon_{j}}x_{-\epsilon_{1}-\epsilon_{j}}K + x_{-\epsilon_{1}-\epsilon_{j}}K' \in \mathfrak{M}_\chi \Rightarrow$ by $adx_{2\epsilon_{1}},   x_{\epsilon_{1}+ \epsilon_{j}}K + x_{2\epsilon_{j}}x_{\epsilon_{1}-\epsilon_{j}}K + x_{\epsilon_{1}-\epsilon_{j}}K' \in \mathfrak{M}_\chi$ is obtained. Hence (i) reduces to (iii).\newline

 Similarly (ii)reduces to (iii) or (iv) or (v). So we have only to consider (iii), (iv) , (v). However (iii), (iv), (v) reduce to $x_{2\epsilon_{1}}K + K'' \in \mathfrak{M}_\chi$ after all considering the situation as in (I). Similarly following the argument as in (I), we are led to a contradiction $K \in \mathfrak{M}_\chi$ .

\end{proof}

Now we  are ready to consider another nonzero character $\chi$ different from that of proposition 4.6.

\begin{prop}
Let $\chi $ be a character of any simple $L$-module with $\chi(h_\alpha )\neq $ 0  for  some $\alpha \in \Phi$, where $h_\alpha$ is an element in the Chevalley basis of  $L$ such that $F x_\alpha + Fh_\alpha+ Fx_{-\alpha}= \frak {sl}_2(F)$ with $[x_\alpha,x_{-\alpha}]= h_\alpha \in H$.\newline

 We then  claim  that any simple $L$-module with character $\chi$ is of dimension  $p^m=p^{n-l\over 2}$,where $n= dim L= 2m +l $ for a CSA  H with $dim H =l$.
\end{prop}

\begin{proof}

Let $ \mathfrak {M}_\chi$ be the kernel of this irreducible representation,i.e., a certain (2-sided) maximal ideal of $u(L).$ \newline
If $x_{\epsilon_1- \epsilon_2}\not \equiv 0$ or $x_{\epsilon_2- \epsilon_1}\not \equiv 0$, then our assertion is evident from previous proposition 4.6  or proposition 4.1 in [KY-7].\newline

So we may let $x_{\epsilon_1- \epsilon_2}\equiv x_{\epsilon_2- \epsilon_1}\equiv 0$ modulo $\frak M_\chi$.\newline

(I) Assume first that $\alpha$ is a  short root; then we may put $\alpha$=$\epsilon_1-\epsilon_2$  without loss of generaity since all roots of a given length 
are conjugate under the Weyl group of the root system $\Phi$.\newline

 First we let  $B_i$:=$b_{i1}$ $h_{\epsilon_{1}- \epsilon_{ 2}}$+ $b_{i2}$ $h_{\epsilon_{2}-\epsilon_{3}}$+$\cdots$+$b_{i,l-1}$ $h_{\epsilon_{l-1}-\epsilon_{ l }}$ + $b_{il}$ $h_{\epsilon_{2l}}$  for $i=  1,2,$ $\cdots,2m,$ where ($b_{i1}$,$b_{i2}$   $\cdots,$$b_{il}$) $\in F^{l}$ are chosen so that any $ (l+1)-$$ B_i$'s  are linearly independent in $\mathbb{P}^l(F),$ the  $\frak B$   below  becomes an $F-$ linearly independent  set in $u(L)$ if necessary and $x_\alpha$$B_i$ $\not\equiv$$B_i$$x_\alpha$ for $\alpha$=$\epsilon_1-$$\epsilon_2.$ \newline

In  $u(L)$/$\mathfrak{M}_\chi$ we claim that we have a basis \newline

$\frak B$:= 
$\{(B_1 +A_{\epsilon_{1}-\epsilon_ { 2}})^{i_1}\otimes(B_2 +A_{-(\epsilon_{1}-\epsilon_ {2})})^{i_2}\otimes \cdots \otimes(B_{2l-2} +A_{-(\epsilon_{l-1}-\epsilon_{l})})^{i_{2l-2}}\otimes(B_{2l-1}+ A_{2\epsilon_{l}})^{i_{2l-1}}\otimes(B_{2l}+A_{-2\epsilon_{l}})^{i_{2l}}\otimes(\otimes_{j=2l+1}^{2m}(B_j+A_{\alpha_{j}})^{i_{j}}) | 0 \leq i_{j}\leq p-1 \}$, \newline

where we put\newline

$A_{\epsilon_{1}-\epsilon_{2}}= g_\alpha=
g_{\epsilon_{1}-\epsilon_{2}},$ \newline

$A_{\epsilon_{2}-\epsilon_{1}}$=$c_{-(\epsilon_{1}-\epsilon_{2})}$ +$(h_{\epsilon_{1}-\epsilon_{2}}+1)^2$ +4$x_{-\alpha}$ $x_\alpha,$ \newline

$A_{{\epsilon_{2}}{\pm}\epsilon_{3}}$=$x_{\pm2\epsilon_{3}}$ $(c_{\epsilon_{2}\pm\epsilon_{3}}+ x_{\epsilon_{2}\pm\epsilon_{3}}x_{-(\epsilon_{2}\pm\epsilon_{3})}\pm x_{\epsilon_{1}\pm\epsilon_{3}}x_{-(\epsilon_{1}\pm\epsilon_{3})}),$ \newline

$A_{\epsilon_{1}+\epsilon_{2}}$=$g_{\epsilon_1-\epsilon_2}^2$ $(c_{\epsilon_{1}+\epsilon_{2}}+2^{-1}x_{\epsilon_{1}+\epsilon_{2}}x_{-\epsilon_{1}-\epsilon_{2}}\pm 3^{-1}x_{2\epsilon_{1}}x_{-2\epsilon_{1}}\pm 3^{-1} x_{2\epsilon_{2}}x_{-2\epsilon_{2}}),$ \newline

$A_{-\epsilon_1- \epsilon_2}= g_{\epsilon_1- \epsilon_2}^3((c_{-\epsilon_{1}-\epsilon_{2}}+2^{-1}x_{\epsilon_{1}+\epsilon_{2}}x_{-\epsilon_{1}-\epsilon_{2}}\pm 3^{-1}x_{2\epsilon_{1}}x_{-2\epsilon_{1}}\pm 3^{-1} x_{2\epsilon_{2}}x_{-2\epsilon_{2}}),$ \newline

$A_{\epsilon_{2}\pm\epsilon_{k}}$=$x_{\epsilon_{3}\pm\epsilon_{k}}( c_{\epsilon_{2}\pm\epsilon_{k}}+x_{\epsilon_{2}\pm\epsilon_{k}}x_{-(\epsilon_{2}\pm\epsilon_{k})} \pm x_{\epsilon_{1}\pm\epsilon_{k}}x_{-(\epsilon_{1}\pm\epsilon_{k})} ),$ \newline

$A_{2\epsilon_{2}}$=$g_{\epsilon_1- \epsilon_2}^6 (c_{2\epsilon_{2}}+3^{-1}x_{2\epsilon_{2}}x_{-2\epsilon_{2}}\pm 2^{-1}x_{\epsilon_{1}+ \epsilon_{2}}x_{-\epsilon_{1}-\epsilon_{2}} +3^{-1}x_{2\epsilon_{1}}x_{-2\epsilon_{1}}),$ \newline

$A_{-2\epsilon_2}= g_\alpha A_{\epsilon_2- \epsilon_1}(c_{-2\epsilon_2}+ 3^{-1}x_{2\epsilon_2}x_{-2\epsilon_2}\pm 2^{-1}x_{\epsilon_1+ \epsilon_2}x_{-\epsilon_1- \epsilon_2}+ 3^{-1}x_{2\epsilon_1}x_{-2\epsilon_1}) $\newline

$A_{2\epsilon_1}= g_{\epsilon_1- \epsilon_2}^4(c_{2\epsilon_1}+ 3^{-1}x_{-2\epsilon_1}x_{2\epsilon_1}\pm 2^{-1}x_{-\epsilon_1- \epsilon_2}x_{\epsilon_1+ \epsilon_2}\pm 3^{-1}x_{-2\epsilon_2}x_{2\epsilon_2})$\newline

$A_{-2\epsilon_1}$=$g_{\epsilon_1- \epsilon_2}^5 ( c_{-2\epsilon_{1}}+ 3^{-1}x_{-2\epsilon_{1}}x_{2\epsilon_{1}}\pm2^{-1}x_{-\epsilon_ {1}-\epsilon_{2}}x_{\epsilon_1+\epsilon_2}
\pm3^{-1}x_{-2\epsilon_2}x_{2\epsilon_2}), $ \newline

$A_{-(\epsilon_{1}\pm\epsilon_{3})}$= $x_{-(\pm\epsilon_{3})}(c_{-(\epsilon_{2}\pm\epsilon_3)}
+  x_{\epsilon_{2}\pm\epsilon_{3}}x_{-(\epsilon_{2}\pm\epsilon_{3})}\pm x_{\epsilon_{1}\pm\epsilon_{3}}x_{-(\epsilon_{1}\pm\epsilon_{3})}),$
\newline

$A_{-(\epsilon_{1}\pm\epsilon_{k})}$= $x_{-(\epsilon_{3}\pm\epsilon_{k})}(c_{-(\epsilon_{1}\pm\epsilon_{k})}+ x_{\epsilon_{2}\pm\epsilon_{k}}x_{-(\epsilon_{2}\pm\epsilon_{k})}\pm x_{\epsilon_{1}\pm\epsilon_{k}}x_{-(\epsilon_{1}\pm\epsilon_{k})} ),$ \newline

$A_{2\epsilon_{l}}$= $x_{2\epsilon_{l}}^2$(if $l\neq 1,2) , $ \newline

$A_{-2\epsilon_{l}}$= $x_{-2\epsilon_{l}}^2, $ 
\newline

with the sign chosen so that they commute with $x_{\alpha}$ and with $c_{\alpha}\in F$ 
chosen so that $A_{\epsilon_{2}-\epsilon_{1}}$ and parentheses are invertible.  For any other root $ \beta$ we put $A_{\beta}$= $x_{\beta}^2 $ or $x_{\beta}^3 $ if possible. \newline

Otherwise attach to these sorts the parentheses(        ) used for designating $A_{-\beta}$ so that  $A_\gamma  \forall \gamma \in \Phi$ may commute with $x_\alpha$.\newline

We shall prove that $\frak B$ is a basis in $u(L))$/$\mathfrak {M}_\chi$. By virtue of P-B-W theorem, it is not difficult to see that $\frak B$ is evidently a linearly independent set  over $F$ in $u(L)$. Furthermore $\forall$ $\beta$ $ \in\Phi$, $A_{\beta}\notin\mathfrak {M}_\chi$(see detailed proof below).\newline

We shall prove that a nontrivial linearly dependent equation leads to absurdity.\newline

We assume first that there is a dependence equation which is of least degree with respect to $h_{\alpha_{j}}\in H$ and the number of whose highest degree terms is also least. \newline

In case it is conjugated by $x_{\alpha}$, then there arises a nontrivial dependence equation of lower degree than the given one, which contradicts our assumption.\newline

 Otherwise it reduces to one of the following forms:\newline

(i)$x_{\pm 2\epsilon_j}K+ K' \in \frak M_\chi, $\newline

(ii)$x_{\pm \epsilon_j \pm \epsilon_k}K+ K' \in \frak M_\chi,$\newline

(iii)$g_{\epsilon_1- \epsilon_2}K+ K' \in \frak M_\chi,$\newline

where $K,K'$ commute with $x_\alpha$ and $x_{-\alpha}$ modulo $\frak M_\chi$.\newline

By making use of proofs of proposition 4.6 and theorem 2.1 in [KY-7],we may reduce (i)  and (ii) to the equation of the form\newline

 $x_{\epsilon_1- \epsilon_2}K+ K'\in \frak M_\chi,$\newline

where $K$ commute with $x_{\pm (\epsilon_1- \epsilon_2)}$ and $K'$ commute with $x_{\epsilon_1- \epsilon_2}$ modulo $\frak M_\chi$.\newline

We have $x_{\epsilon_1- \epsilon_2}^p K+ x_{\epsilon_1- \epsilon_2}^{p-1}K' \equiv 0 $, so we get $x_{\epsilon_1- \epsilon_2}^{p-1}K'\equiv 0$.\newline

Subtracting $x_{\epsilon_2- \epsilon_1}x_{\epsilon_1- \epsilon_2}K+ x_{\epsilon_2- \epsilon_1}K'\equiv 0$ from  this equation, we obtain $ -x_{\epsilon_2- \epsilon_1}x_{\epsilon_1- \epsilon_2}K+ g_\alpha K'\equiv 0$.
We should remember that $g_\alpha$ is invertible in $u(L)/\frak M_\chi$ by virtue of  [RS].\newline

By the way we use $w_\alpha := (h_\alpha+ 1)^2+ 4x_{-\alpha}x_\alpha \in $ the center of $u(\frak {sl}_2(F))$.Hence we have $-4^{-1}\{w_\alpha- (h_\alpha + 1)^2\}K+ g_\alpha K'\equiv 0$. So we obtain

$4^{-1}g_\alpha^{p-1}\{(h_\alpha+ 1)^2- w_\alpha\}+ cK'\equiv 0 \cdots (\ast)$ \newline

and from the start equation we get \newline

$cx_\alpha K+ cK'\equiv 0 \cdots (\ast \ast)$.\newline

Subtracting $(\ast \ast)$ from $(\ast)$, we get $4^{-1}g_\alpha^{-1}\{(h_\alpha+ 1)^2- w_\alpha\}K- cx_\alpha K \equiv 0$.
Multiplying this equation by $g_\alpha^{1-p}$ to the right, we have\newline

 $4^{-1}g_\alpha^{p-1}\{h_\alpha+1)^2- w_\alpha\}g_\alpha^{1-p}K- cx_\alpha g_\alpha^{1-p}K \equiv 4^{-1}g_\alpha^{p-1} \{(h_\alpha+ 1)^2- w_\alpha\}g_\alpha^{1-p}K+ x_\alpha x_{-\alpha}K \equiv 0$.\newline

Conjugation of the brace of this equation $(p-1)$- times by $g_\alpha$ gives rise to $4^{-1}\{(h_\alpha- 1)^2- w_\alpha\}K+ x_\alpha x_{-\alpha}K\equiv 0$.
Next mutiplying $x_{-\alpha}^{p-1}$ to the right of the last equation, we obtain \newline

$\{(h_\alpha -1)^2- w_\alpha\}K x_{-\alpha}^{p-1}\equiv 0$ modulo $\frak M_\chi$. \newline

Now we multiply $x_\alpha$  to the left of this equaion consecutively until it becomes of the form \newline

(a nonzero polynomial of degree $\geq 1$ with respect to $h_\alpha)K \newline
\equiv 0$ modulo $ \frak M_\chi$.\newline

If we make use of conjugation and subtraction consecutively, then we arrive at a contradiction $K\equiv 0$.\newline
Otherwise we may change the role of $K$ and $K'$ to obtain the absurdity alike.\newline

Next for the case (iii),we change it to the form (iii)$'K+ g_\alpha^{-1}K'\in \frak M_\chi$.\newline

We thus have an equation\newline

 $x_{\epsilon_1- \epsilon_2}K+ x_{\epsilon_1- \epsilon_2}g_{\epsilon_1- \epsilon_2}^{-1}K'\equiv 0$ modulo $\frak M_\chi$.
According to the above argument, we are also  led to a contradiction $K\in \frak M_\chi$.

\end{proof}

Now  we turn to the classical Lie algebra of $D_l$- type. We let $L$ be any modular Lie algebra of $D_l$-type over any algebraically closed field  $F$ of characteristic $p\geq 7$.\newline

We note first that the  Lie algebra of $D_l$- type with rank $l$, i.e., the $D_l$-type Lie algera over $\mathbb{C}$ has its root system $\Phi$=$\{\pm(\epsilon_i\pm \epsilon_j) | i \neq j\}$,where $\epsilon_i , \epsilon_j $ are linearly independent  orthonormal unit vectors in $\mathbb{R}^l$ with $l \geq 4$. The base of $\Phi$ is equal to  $\{\epsilon_1-\epsilon_2,\epsilon_2- \epsilon_3,\cdots ,\epsilon_{l-1}-\epsilon_l,\epsilon_{l- 1}+ \epsilon_l\}.   $\newline

\begin{prop}

We let  $\alpha$  be any  root  in the root system  $\Phi$ of $L .$ If $\chi(x_\alpha)$ $\neq0,$ then $dim_F$$ \rho_\chi$$(u(L))$ = $p^{2m},$ where $ [Q(u(L)):Q(\mathfrak{Z})]$=$p^{2m}$=$p^{n-l}$ with $\mathfrak{Z}$ the center of $u(L)$  and  $Q$  denotes  the quotient algebra. \newline

So  we claim that the  simple module corresponding to this representation has $p^m$ as its  dimension and the Lee's basis for this simple module is given  as  follows.\newline

Let $\frak M_\chi$ be the kernel of this simple representation with character $\chi$. 
 We may put $\alpha=\epsilon_1-\epsilon_2$ since only long roots exist and all roots of the same length are conjugate under the Weyl group of $\Phi$.\newline

We put $B_i:=b_{i1}$ $h_{\epsilon_{1}- \epsilon_{ 2}}$+ $b_{i2}$ $h_{\epsilon_{2}-\epsilon_{3}}$+$\cdots$+$b_{i,l-1}$ $h_{\epsilon_{l-1}-\epsilon_{ l }}$+ $b_{il}h_{\epsilon_{l-1}+ \epsilon_l},$ where ($b_{i1}$,$b_{i2}$   $\cdots,$$b_{il}$) $\in F^{l}$ are chosen so that any $ (l+1)-$$ B_i$'s  are linearly independent in $\mathbb{P}^l(F),$ the  $\frak B$   below  becomes an $F-$ linearly independent  set in $u(L)$ if necessary and $x_\alpha$$B_i$ $\not\equiv$$B_i$$x_\alpha$ for $\alpha$=$\epsilon_1-$$\epsilon_2.$ \newline

In $ u(L))/\frak M_\chi$ we have a basis $\frak B$$:= \{(B_1+ A_{\epsilon_1- \epsilon_2})^{i_1}\otimes (B_2+ A_{-(\epsilon_1-\epsilon_2)})^{i_2}\otimes \cdots \otimes(B_{2l-2}+ A_{-(\epsilon_{l-1}-\epsilon_l)})^{i_{2l-2}}\otimes (B_{2l-1}+ A_{\epsilon_{l-1}+ \epsilon_l})^{i_{2l-1}}\otimes (B_{2l}+ A_{-(\epsilon_{l-1}+ \epsilon_l)})^{i_{2l}}\otimes (\otimes_{j=2l+1}^{2m}(B_j+ A_{\alpha_j})^{i_j})| 0\leq i_ j\leq p-1 \} $,\newline

where we put  \newline

$A_{\epsilon_1-\epsilon_2}= x_\alpha= x_{\epsilon_1- \epsilon_2}$,\newline

$ A_{\epsilon_2- \epsilon_1}=c_{\epsilon_2- \epsilon_1}+ (h_{\epsilon_1- \epsilon_2}+ 1 )^2+ 4x_{\epsilon_2-\epsilon_1}x_{\epsilon_1- \epsilon_2}$\newline

$A_{\epsilon_2\pm \epsilon_3}=     x_{\pm \epsilon_1- \epsilon_2}(c_{\epsilon_2\pm \epsilon_3}+ x_{\epsilon_2\pm \epsilon_3}x_{-(\epsilon_2\pm \epsilon_3)}\pm x_{\epsilon_1\pm \epsilon_3}x_{-(\epsilon_1\pm \epsilon_3)})$, \newline

$A_{-(\epsilon_1\pm \epsilon_3)}= ( x_{-\epsilon_1- \epsilon_2}^2$  or  $x_{-\epsilon_1- \epsilon_2}^3)(c_{-(\epsilon_1\pm \epsilon_3)}+ x_{\epsilon_2\pm \epsilon_3}x_{-(\epsilon_2\pm \epsilon_3)}\pm x_{\epsilon_1 \pm\epsilon_3}x_{-(\epsilon_1\pm \epsilon_3)}) $,\newline

$A_{\epsilon_2\pm \epsilon_k}= x_{\epsilon_3\pm \epsilon_k}
(c_{\epsilon_2\pm \epsilon_k}+ x_{\epsilon_2\pm \epsilon_k}x_{-(\epsilon_2\pm \epsilon_k)}\pm x_{\epsilon_1\pm \epsilon_k}x_{-(\epsilon_1\pm \epsilon_k)})  $ \newline

if $k\neq 1$, \newline

$A_{-(\epsilon_1\pm \epsilon_k)}= x_{-(\epsilon_2\pm \epsilon_k)}(c_{-(\epsilon_1\pm \epsilon_k)}+ x_{\epsilon_2\pm \epsilon_k }x_{-(\epsilon_2\pm \epsilon_k)}\pm x_{\epsilon_1\pm \epsilon_k}x_{-(\epsilon_1\pm \epsilon_k)}) $, \newline

with the signs chosen so that they  may commute with $x_\alpha$ and with $c_\beta\in F$ chosen so that $A_{\epsilon_2-\epsilon_1}$ and parentheses are invertible.\newline
For any other root $\beta$, we put $A_\beta= x_\beta^2 $ or $x_\beta^3 $ or $x_\beta^4 $ if possible.\newline

Otherwise we make use of the parentheses(      ) again used for designating $A_{-\beta}$. So in this case we put $A_\beta= x_\gamma^2 $       
 or $ x_\gamma^3  $ attached to these (      ) so that  $x_\alpha$ may commute with $A_\beta$.\newline

\end{prop}

\begin{proof}
Also evident from [KY-4].
We may actually prove  this proposition as follows.\newline

Let $\frak M_\chi$ be the kernel of this simple representation with character $\chi$. 
 We may put $\alpha=\epsilon_1-\epsilon_2$ since only long roots exist and all roots of the same length are conjugate under the Weyl group of $\Phi$.\newline

The Lee's basis for this simple module is given  as  follows.\newline

As  suggested in the  present proposition, we put $B_i:=b_{i1}$ $h_{\epsilon_{1}- \epsilon_{ 2}}$+ $b_{i2}$ $h_{\epsilon_{2}-\epsilon_{3}}$+$\cdots$+$b_{i,l-1}$ $h_{\epsilon_{l-1}-\epsilon_{ l }}$+ $b_{il}h_{\epsilon_{l-1}+ \epsilon_l},$ where ($b_{i1}$,$b_{i2}$   $\cdots,$$b_{il}$) $\in F^{l}$ are chosen so that any $ (l+1)-$$ B_i$'s  are linearly independent in $\mathbb{P}^l(F),$ the  $\frak B$   below  becomes an $F-$ linearly independent  set in $u(L)$ if necessary and $x_\alpha$$B_i$ $\not\equiv$$B_i$$x_\alpha$ for $\alpha$=$\epsilon_1-$$\epsilon_2.$ \newline

In $ u(L))/\frak M_\chi$ we have a basis $\frak B$$:= \{(B_1+ A_{\epsilon_1- \epsilon_2})^{i_1}\otimes (B_2+ A_{-(\epsilon_1-\epsilon_2)})^{i_2}\otimes \cdots \otimes(B_{2l-2}+ A_{-(\epsilon_{l-1}-\epsilon_l)})^{i_{2l-2}}\otimes (B_{2l-1}+ A_{\epsilon_{l-1}+ \epsilon_l})^{i_{2l-1}}\otimes (B_{2l}+ A_{-(\epsilon_{l-1}+ \epsilon_l)})^{i_{2l}}\otimes (\otimes_{j=2l+1}^{2m}(B_j+ A_{\alpha_j})^{i_j})| 0\leq i_ j\leq p-1 \} $,\newline

where we put  \newline

$A_{\epsilon_1-\epsilon_2}= x_\alpha= x_{\epsilon_1- \epsilon_2}$,\newline

$ A_{\epsilon_2- \epsilon_1}=c_{\epsilon_2- \epsilon_1}+ (h_{\epsilon_1- \epsilon_2}+ 1 )^2+ 4x_{\epsilon_2-\epsilon_1}x_{\epsilon_1- \epsilon_2}$\newline

$A_{\epsilon_2\pm \epsilon_3}=     x_{\pm \epsilon_1- \epsilon_2}(c_{\epsilon_2\pm \epsilon_3}+ x_{\epsilon_2\pm \epsilon_3}x_{-(\epsilon_2\pm \epsilon_3)}\pm x_{\epsilon_1\pm \epsilon_3}x_{-(\epsilon_1\pm \epsilon_3)})$, \newline

$A_{-(\epsilon_1\pm \epsilon_3)}= ( x_{-\epsilon_1- \epsilon_2}^2$  or  $x_{-\epsilon_1- \epsilon_2}^3)(c_{-(\epsilon_1\pm \epsilon_3)}+ x_{\epsilon_2\pm \epsilon_3}x_{-(\epsilon_2\pm \epsilon_3)}\pm x_{\epsilon_1 \pm\epsilon_3}x_{-(\epsilon_1\pm \epsilon_3)}) $,\newline

$A_{\epsilon_2\pm \epsilon_k}= x_{\epsilon_3\pm \epsilon_k}
(c_{\epsilon_2\pm \epsilon_k}+ x_{\epsilon_2\pm \epsilon_k}x_{-(\epsilon_2\pm \epsilon_k)}\pm x_{\epsilon_1\pm \epsilon_k}x_{-(\epsilon_1\pm \epsilon_k)})  $ \newline

if $k\neq 1$, \newline

$A_{-(\epsilon_1\pm \epsilon_k)}= x_{-(\epsilon_2\pm \epsilon_k)}(c_{-(\epsilon_1\pm \epsilon_k)}+ x_{\epsilon_2\pm \epsilon_k }x_{-(\epsilon_2\pm \epsilon_k)}\pm x_{\epsilon_1\pm \epsilon_k}x_{-(\epsilon_1\pm \epsilon_k)}) $, \newline

with the signs chosen so that they  may commute with $x_\alpha$ and with $c_\beta\in F$ chosen so that $A_{\epsilon_2-\epsilon_1}$ and parentheses are invertible.\newline
For any other root $\beta$, we put $A_\beta= x_\beta^2 $ or $x_\beta^3 $ or $x_\beta^4 $ if possible.\newline

Otherwise we make use of the parentheses(      ) again used for designating $A_{-\beta}$. So in this case we put $A_\beta= x_\gamma^2 $       
 or $ x_\gamma^3  $ attached to these (      ) so that  $x_\alpha$ may commute with $A_\beta$.\newline

It is not difficult to see  that the rest of this proof is similar to that of proposition 4.3.
We would like to actually complete the proof  as follows.

We may see without difficulty that $\frak B$ is a linearly independent set in $u(L)$ by virtue of P-B-W theorem.\newline

We shall prove that a nontrivial linearly dependent  equation leads to absurdity. We assume first that we have a dependence equation which is of least degree with respect to $h_{\alpha_j}\in H$ and the number of whose highest  degree terms is also least.\newline

In case it is conjugated by $x_\alpha$, then there arises a nontrivial dependence equation of lower degree than the given one,which contradicts to our assumption.\newline

Otherwise  it reduces to the following form, which we have only to prove:\newline

 $(\ast) x_{\epsilon_{1}-\epsilon_{2}}$$K$ + $K'$ $\in$$ \mathfrak{M}_\chi$ 

, where $K$ and $K'$  commute with $x_{\alpha}$ modulo $\frak {M}_\chi $.\newline

 We thus deal with and assume  $x_{\epsilon_{2}-\epsilon_{1}}$$x_{\epsilon_{1}-\epsilon_{2}}$$K$ + $x_{\epsilon_{2}-\epsilon_{1}}$$K''$ $\in$ $\mathfrak{M}_\chi$.\newline

From $w_{\epsilon_1- \epsilon_2}:=(h_{\epsilon_1- \epsilon_2}+ 1)^2 +4 x_{\epsilon_2- \epsilon_1}x_{\epsilon_1- \epsilon_2}\in $ the center of  $u(\frak{sl}_2(F)), $ we get $ 4^{-1}\{w_{\epsilon_1- \epsilon_2}- (h+ 1)^2\}K+ x_{\epsilon_2- \epsilon_1}K'' \equiv 0 $  modulo $\frak M_\chi$.\newline

If $x_{\epsilon_2- \epsilon_1}^p\equiv c $ which is a constant,then \newline

$(\ast \ast)4^{-1}x_{\epsilon_2-\epsilon_1}^{p-1}\{w_{\epsilon_1- \epsilon_2}- (h_{\epsilon_1- \epsilon_2}+ 1)^2\}K+ cK''\equiv 0 $\newline

 is obtained.
\newline
From $(\ast),(\ast \ast)$, we have\newline

 $4^{-1}x_{\epsilon_2- \epsilon_1}^{p-1}\{w_{\epsilon_1- \epsilon_2}- (h_{\epsilon_1- \epsilon_2}+ 1)^2\}K- cx_{\epsilon_1- \epsilon_2}K\newline
\equiv 0$ modulo $\frak M_\chi$.\newline

Multiplying $x_{\epsilon_1- \epsilon_2}^{p-1}$ to this equation,we obtain \newline

 $(\ast \ast \ast)4^{-1}x_{\epsilon_1- \epsilon_2}^{p-1}x_{\epsilon_2- \epsilon_1}^{p-1}\{w_{\epsilon_1- \epsilon_2}- (h_{\epsilon_1- \epsilon_2}+ 1)^2\}K- cx_{\epsilon_1- \epsilon_2}^pK\equiv 0.  $\newline

By making use of $w_{\epsilon_1- \epsilon_2}$, we may deduce from $(\ast \ast \ast)$ an equation of the form \newline
( a polynomial of degree $\geq 1$ with respect to  $ h_{\epsilon_1- \epsilon_2})K- cx_{\epsilon_1- \epsilon_2}^pK\equiv 0.   $\newline

Finally if we use conjugation and subtraction consecutively,then we are led to a  contradiction $K\in \frak M_\chi.$\newline
It may be necessary  for us to change the role of $K$ and $K'$ to obtain the absurdity alike.\newline

\end{proof}

\begin{prop}

We let $\chi $ be a character of any simple $L$-module with $\chi(h_\alpha )\neq $ 0  for  some $\alpha \in \Phi$,where $h_\alpha$ is an element in the Chevalley basis of  $L$ such that $F x_\alpha + Fh_\alpha+ Fx_{-\alpha}= \frak {sl}_2(F)$ with $[x_\alpha,x_{-\alpha}]= h_\alpha \in H$(a CSA of $L$).\newline
We then  claim that  $dim_F$$ \rho_\chi$$(u(L))$ = $p^{2m},$ where $ [Q(u(L)):Q(\mathfrak{Z})]$=$p^{2m}$=$p^{n-l}$ with $\mathfrak{Z}$ the center of $u(L)$  and  $Q$  denotes  the quotient algebra. \newline

Namely, we claim that  the  simple module corresponding to this representation has $p^m$ as its  dimension.
\end{prop}

\begin{proof}

We may   let   $\alpha=  \epsilon_1- \epsilon_2 $  without loss of generality since all roots are long and are conjugate under the Weyl group of $\Phi$.\newline

If  $\chi (x_{\alpha})\neq 0$ or $\chi (x_{-\alpha}) \neq 0$, then  our claim is clear from the preceding proposition4.8.
So we may let $x_{\epsilon_1}^p \equiv x_{- \epsilon_1}^p \equiv 0$ modulo $\frak M_\chi$.\newline

As we  have  done before, we let   $B_i $ be defined as in proposition 4.3.\newline

In  $u(L)/\frak M_\chi$, we have a Lee's basis as $\frak B$$:= \{(B_1+ A_{\epsilon_1- \epsilon_2})^{i_1}\otimes (B_2+ A_{-(\epsilon_1-\epsilon_2)})^{i_2}\otimes \cdots \otimes(B_{2l-2}+ A_{-(\epsilon_{l-1}-\epsilon_l)})^{i_{2l-2}}\otimes (B_{2l-1}+ A_{\epsilon_{l-1}+ \epsilon_l})^{i_{2l-1}}\otimes (B_{2l}+ A_{-(\epsilon_{l-1}+ \epsilon_l})^{i_{2l}}\otimes (\otimes_{j=2l+1}^{2m}(B_j+ A_{\alpha_j})^{i_j})| 0\leq i_ j\leq p-1 \} $,\newline

where we put  \newline

$A_{\epsilon_1-\epsilon_2}= g_\alpha= g_{\epsilon_1- \epsilon_2}$,\newline

 $A_{\epsilon_2- \epsilon_1}=c_{\epsilon_2- \epsilon_1}+ (h_{\epsilon_1- \epsilon_2}+ 1 )^2+ 4x_{\epsilon_2-\epsilon_1}x_{\epsilon_1- \epsilon_2} $,\newline

$A_{\epsilon_2\pm \epsilon_3}= g_\alpha^{2 or 3} c_{\epsilon_2 \pm \epsilon_3}+ x_{\epsilon_2\pm \epsilon_3}x_{-(\epsilon_2 \pm \epsilon_3)}                 \pm x_{\epsilon_1 \pm \epsilon_3}x_{-(\epsilon_1 \pm \epsilon_3)}) $,\newline

$A_{-(\epsilon_1\pm \epsilon_3)}= g_\alpha^{4  or  5}(c_{-(\epsilon_1\pm \epsilon_3)}+ x_{\epsilon_1\pm  \epsilon_3}x_{-(\epsilon_1\pm \epsilon_3)}\pm x_{\epsilon_2\pm \epsilon_3}x_{-(\epsilon_2\pm \epsilon_3)})$,\newline

$A_{\epsilon_2\pm \epsilon_k}= x_{\epsilon_3\pm \epsilon_k}(c_{\epsilon_2\pm \epsilon_k}+ x_{\epsilon_2\pm \epsilon_k}x_{-(\epsilon_2\pm \epsilon_k)}\pm x_{\epsilon_1\pm \epsilon_k}x_{-(\epsilon_1\pm \epsilon_k)})   $,\newline

$A_{-(\epsilon_1\pm \epsilon_k)}= x_{-(\epsilon_3\pm \epsilon_k)}^ 2(c_{-(\epsilon_1\pm \epsilon_k)}+ x_{\epsilon_2\pm \epsilon_k}x_{-(\epsilon_2\pm \epsilon_k)}\pm x_{\epsilon_1\pm \epsilon_k}x_{-(\epsilon_1\pm \epsilon_k)})$,\newline

with the signs chosen so that they may commute with $x_\alpha$ and with $c_\beta\in F$ chosen so that $A_{\epsilon_2-\epsilon_1}$ and parentheses are invertible.\newline

For any other root $\beta$, we put $A_\beta= x_\beta^3 $ or $x_\beta^4 $ if possible.\newline

Otherwise we make use of the parentheses(      ) again used for designating $A_{-\beta}$. So in this case we put $A_\beta= x_\gamma^3 $       
 or $ x_\gamma^4 $ attached to these (      ) so that  $x_\alpha$ may commute with $A_\beta$.\newline
 
We may see without difficulty that $\frak B$ is a linearly independent set in $u(L)$ by virtue of P-B-W theorem.\newline

We shall prove that a nontrivial linearly dependent  equation leads to absurdity. We assume first that we have a dependence equation which is of least degree with respect to $h_{\alpha_j}\in H$ and the number of whose highest  degree terms is also least.\newline

In case it is conjugated by $g_\alpha$, then there arises a nontrivial dependence equation of lower degree than the given one,which contradicts to our assumption.\newline

Otherwise  it reduces to one of the following forms, so that we have only to prove that \newline

(i)$x_{\beta}K+ K'\in \frak M_\chi$,

(ii) $g_\alpha K+ K'\in \frak M_\chi$ \newline

lead to a contradiction, where both $K$ and $K'$ commute with $x_{\pm \alpha}$  modulo $\frak M_\chi $. In particular $K$ commute with $g_\alpha$.
\newline

For the case (i), we may change it to the form $x_{\alpha}K+ K''\in \frak M_\chi$ for some $K''$ commuting with $x_\alpha$ modulo $\frak M_\chi$.\newline

So we have $x_\alpha^p K+ x_\alpha^{p-1}K''\equiv 0$, and hence $x_\alpha^{p-1}K''\equiv 0$.\newline

Subtracting from this $x_{-\alpha}x_\alpha K+ x_{-\alpha}K''\equiv 0$, we get \newline

$-x_{-\alpha}x_\alpha K+ g_\alpha K'' \equiv 0$. Recall here that $g_\alpha$ is invertible and $w_\alpha$ belongs to the center of $u(\frak {sl}_2 (F))$.

So we get  $4^{-1}\{(h_\alpha +1)^2- w_\alpha\}K+ g_\alpha K''\equiv 0$, and hence\newline

 $(\ast) g_\alpha^{p-1} 4^{-1}\{(h_\alpha + 1)^2- w_\alpha \}K+ cK'' \equiv 0$\newline

 is obtained and from the start equation we have \newline

$(\ast \ast)cx_\alpha K+ c K''\equiv 0$, where $g_\alpha^p- c \equiv 0$.\newline

Subtracting $(\ast \ast)$ from $(\ast)$, we have $4^{-1}g_\alpha^{p-1}\{(h_\alpha+ 1)^2- w_\alpha\}K- cx_\alpha K \equiv 0$.\newline

Multiplying this equation by $g_\alpha^{1-p}$ to the right, we obtain $4^{-1}g_\alpha^{p-1}\{(h_\alpha+ 1)^2- w_\alpha\}g_\alpha^{1-p}K- cx_\alpha g_\alpha^{1-p}K \equiv 0$ \newline

We thus have $4^{-1}\{(h_\alpha+ 1- 2)^2- w_\alpha\}K- x_\alpha g_\alpha K \equiv 0$.

So it follows that $4^{-1}\{(h_\alpha -1)^2- w_\alpha\}K+ x_\alpha x_{-\alpha}K \equiv 0 $.\newline

Next multiplying $x_{-\alpha}^{p-1}$ to the right of this last equation, we obtain $\{(h_{\alpha}- 1)^2- w_\alpha\}K x_{-\alpha}^{p-1}\equiv 0$.
Now multiply $x_\alpha$ in turn consecutively to the left of this equation until it becomes of the form \newline

( a nonzero polynomial of degree $\geq 1$ with respect to $h_\alpha)K $\newline
$\in \frak M_\chi$. \newline

By making use of  conjugation and subtraction consecutively, we are led to a contradiction.$K \in \frak M_\chi$.\newline

Finally for the case (ii),we consider $K+ g_\alpha^{-1}K' \in \frak M_\chi$.  So we have $x_\alpha K+ x_\alpha g_\alpha^{-1} K' \equiv 0$ modulo $\frak M_\chi$.

By analogy with the argument  as in the case (i), we obtain a contraiction $K \in \frak M_\chi$.

\end{proof}

\

 $\textbf{[Remark]}$
We have seen up to now that  classical modular Lie algebras of $D_l,C_l, B_l,A_l$-types  have no subregular point under mild conditions.\newline

So we cannot help believing that  all  classical  Lie algebras and  exceptional type Lie algebras  behave like this, and thus we conjectured  such a similar result in [KWa].\newline
In the next paragraphs and followups, we shall preferably deal with  the exceptional type $E_8$.\newline

As is well known the $F_4$- type Lie algebra over $F$ has its dimension $n=52$ with $ l=4$ and the $G_2$- type Lie algebra  over $F$ has its dimension
$n= 14$  with $l=2$. So comparing these with  the types $E_6,E_7$, and $E_8$, we  may find  out their Lee's  bases without difficulty  due to  the fact that  for any $x_{\beta}$ in the Chevalley basis of $L$ over $F$, elements $x_{\beta}, x_{\beta}^2,\cdots, x_{\beta}^6$  all  have   distinct weights with respect to  the  $ad(H)$- module action.\newline

Furthermore  the  root systems of  $E_7$ and $E_6$ are  subroot systems  of  the root system of  $E_8$, so we have only to  deal  exclusively with $E_8$- type  to figure  out  the Lee's bases for the exceptional type Lie algebras. We  might refer to the reference [7]  for the  Lee's bases of  other  exceptional types.\newline

We note  that  the  Lie algebra of $E_8$- type with rank $8$, i.e., the $E_8$-type Lie algera over $\mathbb{C}$ has its root system $\Phi$=$\{\pm(\epsilon_i\pm \epsilon_j) |1\leq i \neq j \leq 8 ; 2^{-1}\Sigma_{i=1}^8 (-1)^{k(i)}\epsilon_i$,where $ k(i)$= 0,1 and add up to an even  integer                                                        \}. Here $\epsilon_i , \epsilon_j $ are linearly independent  orthonormal unit vectors in $\mathbb{R}^8$ . The base of $\Phi$ is equal to  $\{\epsilon_1+\epsilon_2,\epsilon_2- \epsilon_1, \epsilon_3- \epsilon_2,\epsilon_4-\epsilon_3,\epsilon_5- \epsilon_4,\epsilon_6- \epsilon_5,\epsilon_7- \epsilon_6,2^{-1}(\epsilon_1+\epsilon_8-(\epsilon_2+ \cdots +\epsilon_7 ) \}.   $\newline

\begin{prop}
Let  $\alpha$  be any  root  in the root system  $\Phi$ of $L= E_8 .$ If $\chi(x_\alpha)$ $\neq0,$ then $dim_F$$ \rho_\chi$$(u(L))$ = $p^{2m},$ where $ [Q(u(L)):Q(\mathfrak{Z})]$=$p^{2m}$=$p^{n-l}=p^{248- 8}$ with $\mathfrak{Z}$ the center of $u(L)$  and  $Q$  denotes  the quotient algebra. \newline

So  we claim  that the  simple module corresponding to this representation has $p^m=p^{120}$ as its  dimension.\newline
\end{prop}
\begin{proof}

We give the Lee's basis for this simple module   as follows.

Let $\alpha= \epsilon_1+ \epsilon_2$ without loss of generality  due to Weyl group and let us put $B_i$=$b_{i1}$$h_{\epsilon_1 +\epsilon_2}$ +
$b_{i2}h_{\epsilon_2-\epsilon_1}+b_{i3}h_{\epsilon_3- \epsilon_2}+ b_{i4}h_{\epsilon_4- \epsilon_3}+b_{i5}h_{\epsilon_5- \epsilon_4}+ b_{i6}h_{\epsilon_6- \epsilon_5}+ b_{i7}h_{\epsilon_7- \epsilon_6}+b_{i8}h_{2^{-1}(\epsilon_1+ \epsilon_8-(\epsilon_2+\cdots +\epsilon_7))}              $ for $i=1,2,\cdots,2m $ ,where ($b_{i1},\cdots,b_{il})
\in F^l$ are chosen so that  arbitrary ($l+1)-B_i$'s are linearly independent in $\mathbb P^l(F)$,the $\frak B$ below becomes  an $F$-linearly independent set in $u(L)$ if necessary and  $x_\alpha B_i$ $\not \equiv B_i x_\alpha $ with $\alpha =\epsilon_1+ \epsilon_2$.\newline

Let $\frak M_\chi$ be the kernel of the irreducible representation $\rho _{\chi}: L  \rightarrow \frak {gl}(V) $ of the restricted Lie algebra ($L ,[p])$ associated  with any given irreducible $L$-module $V$ with a character $\chi$.\newline

In $u(L)/\frak M_\chi$ we give a Lee's basis as $\frak B$:=$\{(B_1 +A_{\epsilon_1+ \epsilon_2})^{i_1}\otimes (B_2 + A_{-(\epsilon_1+ \epsilon_2)})^{i_2}\otimes \cdots\otimes (B_{16}+A_{-2^{-1}((\epsilon_{1}+ \epsilon_8)- (\epsilon_2+ \cdots+ \epsilon_7))})^{i_{16}}\otimes(\otimes_{j=17}^{240}(B_j+ A_{\alpha_j})^{i_j}   \}$ for 0 $\leq i_j \leq p-1$,\newline

 where we put \newline
$A_{\alpha}=A_{\epsilon_1+ \epsilon_2}= x_{\epsilon_1+ \epsilon_2}$, \newline
$A_{-\epsilon_1-\epsilon_2 }=c_{-\epsilon_1- \epsilon_2}+ (h_{\epsilon_1+ \epsilon_2} +1)^2 +
4x_{-\epsilon_1- \epsilon_2} x_{\epsilon_1+ \epsilon_2}, \newline
A_{-(\epsilon_1\pm \epsilon_k)}= x_{\epsilon_3- (\pm \epsilon_k)}(c_{-\epsilon_1\pm \epsilon_k}+x_{\epsilon_1\pm \epsilon_k}x_{-(\epsilon_1\pm \epsilon_k)}\pm x_{- (\epsilon_2 \mp \epsilon_k) }x_{ \epsilon_2+ (\mp \epsilon_k)}) $,\newline
$A_{-(\epsilon_1 \pm \epsilon_3)}=x_{\epsilon_2- (\pm \epsilon_3)}(c_{-(\epsilon_1 \pm \epsilon_3}+x_{\epsilon_1\mp \epsilon_3}
x_{-(\epsilon_1\pm \epsilon_3)} \pm x_{-(\epsilon_2 \mp\epsilon_3}x_{\epsilon_2+ (\mp \epsilon_3)},$\newline

$A_{-(\epsilon_2\pm \epsilon_k)}=x_{-\epsilon_3- (\pm \epsilon_k)}(c_{-(\epsilon_2\pm\epsilon_k)}+x_{\epsilon_2\pm \epsilon_k}x_{-(\epsilon_2\pm \epsilon_k)}\pm x_{-(\epsilon_1\mp \epsilon_k)}x_{\epsilon_1\mp \epsilon_k)}$,
$A_{-(\epsilon_2 \pm \epsilon_3)}=x_{\epsilon_2- (\pm \epsilon_3)}^2(c_{-(\epsilon_2\pm \epsilon_3)}+ x_{\epsilon_2\pm \epsilon_3}x_{-(\epsilon_2\pm \epsilon_3)\pm x_{-(\epsilon_1\mp \epsilon_3)}x_{\epsilon_1 \mp \epsilon_3} \epsilon_j)}\pm \newline
 x_{\epsilon_j- \epsilon_1}x_{\epsilon_1-\epsilon_j})$, \newline
$A_{\frac{1}{2} \Sigma_{i=1}^{8} (-1)^{k(i)}\epsilon_i}= x_{\frac{1}{2}\Sigma_{i=1}^{8} (-1)^{k(i)}\epsilon_i}  $;
for  $\beta_{k(i)}= -\frac{1}{2}\epsilon_1- \frac{1}{2}\epsilon_2+ \frac{1}{2}\Sigma_{i=3}^{8} (-1)^{k(i)} \epsilon_i$ and  for $\gamma_{k(i)}= \frac{1}{2}\epsilon_1+ \frac{1}{2}\epsilon_2+ \frac{1}{2}\Sigma_{i=3}^{8} (-1)^{k(i)} \epsilon_i$, we set  $ A_{\beta_{k(i)}}= x_{-\frac{1}{2}\Sigma_{i=3}^{8} (-1)^{k(i)}} \epsilon_i$$ (C_{\beta_{k(i)}}+ x_{\beta_{k(i)}}x_{-\beta_{k(i)}}\pm x_{\gamma_{k(i)}}x_{- \gamma_{k(i)}}), A_{\gamma_{k(i)}}= x_{\frac{1}{2}\Sigma_{i=3}^{8} (-1)^{k(i)} \epsilon_i}(c_{\gamma_{k(i)}}+ x_{\beta_{k(i)}} x_{-\beta_{k(i)}}\pm  x_{\gamma_{k(i)}}x_{- \gamma_{k(i)}}).$

with the sign chosen so that  they commute with $x_\alpha$ and with $c_\beta \in F$ chosen so that $A_{-\epsilon_1}$ and parentheses(             ) are invertible.\newline
For any other root $\beta$,  we put $A_\beta={x_\beta}^2$  or $ x_\beta^3 $ if possible.\newline

Otherwise we make use of the parentheses(      )again used for designating $A_{-\beta}$. So in this case we put $A_\beta = { x_\gamma}^2$ or ${x_\gamma}^3$ attached to these (        ) so that  $x_\alpha$ may commute with $A_\beta$.\newline

For the rest of the complete proof , refer to  the consolidated  proof  in proposition 4.12  below.

\end{proof}

\begin{prop}

Let $\chi $ be a character of any simple $L$-module with $\chi(h_\alpha )\neq $ 0  for  some $\alpha \in \Phi$,where $h_\alpha$ is an element in the Chevalley basis of  $L= E_8$ such that $F x_\alpha + Fh_\alpha+ Fx_{-\alpha}= \frak {sl}_2(F)$ with $[x_\alpha,x_{-\alpha}]= h_\alpha \in H$.\newline

 We then claim that any simple $L$-module with character $\chi$ is of dimension  $p^m=p^{n-l\over 2}$,where $n= dim L= 2m +l $ for a CSA  H with $dim H =l$. 

\end{prop}

\begin{proof}

Let  $\alpha$  be any  root  in the root system  $\Phi$ of $L .$ If $\chi(x_\alpha)$ $\neq0,$ or $\chi(x_{-\alpha})\neq 0$, then $dim_F$$ \rho_\chi$$(u(L))$ = $p^{2m}$  due to the  consolidating proof of proposition 4.12  below, where $ [Q(u(L)):Q(\mathfrak{Z})]$=$p^{2m}$=$p^{n-l}=p^{248- 8}$ with $\mathfrak{Z}$ the center of $u(L)$  and  $Q$  denotes  the quotient algebra. \newline

We give the Lee's basis for this simple module   as follows.\newline

Let $\alpha= \epsilon_1+ \epsilon_2$ without loss of generality  due to Weyl group and let us put $B_i$=$b_{i1}$$h_{\epsilon_1 +\epsilon_2}$ +
$b_{i2}h_{\epsilon_2-\epsilon_1}+b_{i3}h_{\epsilon_3- \epsilon_2}+ b_{i4}h_{\epsilon_4- \epsilon_3}+b_{i5}h_{\epsilon_5- \epsilon_4}+ b_{i6}h_{\epsilon_6- \epsilon_5}+ b_{i7}h_{\epsilon_7- \epsilon_6}+b_{i8}h_{2^{-1}(\epsilon_1+ \epsilon_8-(\epsilon_2+\cdots +\epsilon_7))}              $ for $i=1,2,\cdots,2m $ ,where ($b_{i1},\cdots,b_{il})
\in F^l$ are chosen so that  arbitrary ($l+1)-B_i$'s are linearly independent in $\mathbb P^l(F)$,the $\frak B$ below becomes  an $F$-linearly independent set in $u(L)$ if necessary and  $x_\alpha B_i$ $\not \equiv B_i x_\alpha $ with $\alpha =\epsilon_1+ \epsilon_2$.\newline

Let $\frak M_\chi$ be the kernel of the irreducible representation $\rho _{\chi}: L  \rightarrow \frak {gl}(V) $ of the restricted Lie algebra ($L ,[p])$ associated  with any given irreducible $L$-module $V$ with a character $\chi$.\newline

In $u(L)/\frak M_\chi$ we give a Lee's basis as $\frak B$:=$\{(B_1 +A_{\epsilon_1+ \epsilon_2})^{i_1}\otimes (B_2 + A_{-(\epsilon_1+ \epsilon_2)})^{i_2}\otimes \cdots\otimes (B_{16}+A_{-2^{-1}((\epsilon_{1}+ \epsilon_8)- (\epsilon_2+ \cdots+ \epsilon_7))})^{i_{16}}\otimes(\otimes_{j=17}^{240}(B_j+ A_{\alpha_j})^{i_j}   \}$ for 0 $\leq i_j \leq p-1$,\newline

 where we put \newline

$A_{\alpha}=A_{\epsilon_1+ \epsilon_2}= x_{\epsilon_1+ \epsilon_2^{p-1}}- x_{-(\epsilon_1+ \epsilon_2)}= g_{\epsilon_1+ \epsilon_2}$, \newline

$A_{-\epsilon_1-\epsilon_2 }=c_{-\epsilon_1- \epsilon_2}+ (h_{\epsilon_1+ \epsilon_2} +1)^2 +
4x_{-\epsilon_1- \epsilon_2} x_{\epsilon_1+ \epsilon_2},$ \newline

$A_{-(\epsilon_1\pm \epsilon_k)}= x_{\epsilon_3- (\pm \epsilon_k)}(c_{-\epsilon_1\pm \epsilon_k}+x_{\epsilon_1\pm \epsilon_k}x_{-(\epsilon_1\pm \epsilon_k)}\pm x_{- (\epsilon_2 \mp \epsilon_k) }x_{ \epsilon_2+ (\mp \epsilon_k)}) $,\newline

$A_{-(\epsilon_1 \pm \epsilon_3)}=x_{\epsilon_4- (\pm \epsilon_3)}(c_{-(\epsilon_1 \pm \epsilon_3)}+x_{\epsilon_1\mp \epsilon_3}
x_{-(\epsilon_1\pm \epsilon_3)} \pm x_{-(\epsilon_2 \mp\epsilon_3)}x_{\epsilon_2+ (\mp \epsilon_3)},$\newline

$A_{-(\epsilon_2\pm \epsilon_k)}=x_{-\epsilon_3- (\pm \epsilon_k)}(c_{-(\epsilon_2\pm\epsilon_k)}+x_{\epsilon_2\pm \epsilon_k}x_{-(\epsilon_2\pm \epsilon_k)}\pm x_{-(\epsilon_1\mp \epsilon_k)}x_{\epsilon_1\mp \epsilon_k)}$,

$A_{-(\epsilon_2 \pm \epsilon_3)}=x_{\epsilon_5- (\pm \epsilon_3)}(c_{-(\epsilon_2\pm \epsilon_3)}+ x_{\epsilon_2\pm \epsilon_3}x_{-(\epsilon_2\pm \epsilon_3)}\pm x_{-(\epsilon_1\mp \epsilon_3)}x_{\epsilon_1 \mp \epsilon_3} 
 )$, \newline

$A_{\frac{1}{2} \Sigma_{i=1}^{8} (-1)^{k(i)}\epsilon_i}= x_{\frac{1}{2}\Sigma_{i=1}^{8} (-1)^{k(i)}\epsilon_i}  $;
for  $\beta_{k(i)}= -\frac{1}{2}\epsilon_1- \frac{1}{2}\epsilon_2+ \frac{1}{2}\Sigma_{i=3}^{8} (-1)^{k(i)} \epsilon_i$ and  for $\gamma_{k(i)}= \frac{1}{2}\epsilon_1+ \frac{1}{2}\epsilon_2+ \frac{1}{2}\Sigma_{i=3}^{8} (-1)^{k(i)} \epsilon_i$, \newline

we set  

$ A_{\beta_{k(i)}}= x_{-\frac{1}{2}\Sigma_{i=3}^{8} (-1)^{k(i)}} \epsilon_i$$ (C_{\beta_{k(i)}}+ x_{\beta_{k(i)}}x_{-\beta_{k(i)}}\pm x_{\gamma_{k(i)}}x_{- \gamma_{k(i)}}),$\newline

 $A_{\gamma_{k(i)}}= x_{\frac{1}{2}\Sigma_{i=3}^{8} (-1)^{k(i)} \epsilon_i}(c_{\gamma_{k(i)}}+ x_{\beta_{k(i)}} x_{-\beta_{k(i)}}\pm  x_{\gamma_{k(i)}}x_{- \gamma_{k(i)}})$  \newline

 respectively    with the sign chosen so that  they commute with $x_\alpha$ and with $c_\beta \in F$ chosen so that $A_{-\epsilon_1}$ and parentheses(             ) are invertible.\newline

For any other root $\beta$,  we put $A_\beta={x_\beta}^2$  or $ x_\beta^3 $ if possible.\newline
Otherwise we make use of the parentheses(      )again used for designating $A_{-\beta}$. So in this case we put $A_\beta = { x_\gamma}^2$ or ${x_\gamma}^3$ attached to these (        ) so that  $x_\alpha$ may commute with $A_\beta$.\newline
For the rest of the complete proof , refer to  the consolidated  proof  in proposition 4.12  below.

\end{proof}

\

\
 Let $L$ be any modular Lie algebra of classical type or exceptional  type  as listed in section 1 over  any algebraically closed  field $F$ of characteristic $p \geq 7$. \newline

For a root $\alpha \in \Phi$ for $L$, we put $g_\alpha :=  x_\alpha^{p-1}- x_{- \alpha}$ and $w_\alpha:= (h_\alpha+ 1)^2+ 4x_{-\alpha}x_{\alpha}.$

According to [RS], we see that  $w_\alpha$ is contained in the center of $u(\frak {sl}_2(F))$ and must satisfy an irreducible integral equation  of degree $p$ over the center  $\frak Z$  of $u(L)$. \newline

\begin{prop} Let $\chi$ be a character of any simple $L$-module with $\chi (x_{\alpha})\neq 0, \chi (x_{-\alpha})\neq 0$  or  $\chi(h_{\alpha})\neq 0$ for some $\alpha \in \Phi$, where $x_{\pm \alpha}$ and $h_\alpha$ are  elements  in the Chevalley basis of  $L$ such that $F x_\alpha + Fh_\alpha+ Fx_{-\alpha}\cong  \frak {sl}_2 (F)$ .\newline

Then we claim  that any simple $L$-module with character $\chi\neq 0$ is  of  dimension  $p^m=p^{n-l\over 2}$,where $n= dim L= 2m +l $ for a CSA  $ H $ with $dim H =l$. \newline
\end{prop}
\begin{proof}
For  any root $\alpha$, assume first that  $\chi (x_{\alpha})\neq 0 .$

We may see without difficulty that $\frak B$ suggested  in the previous sections is a linearly independent set in $u(L)$ by virtue of P-B-W theorem.\newline

We shall prove that a nontrivial linearly dependent  equation leads to absurdity. We assume first that we have a dependence equation which is of least degree with respect to $h_{\alpha_j}\in H$ and the number of whose highest  degree terms is also least.\newline

In case it is conjugated by $x_\alpha$, then there arises a nontrivial dependence equation of lower degree than the given one,which contradicts to our assumption.\newline

Otherwise  it reduces to the following form, which we have only to prove:\newline

 $(\ast) x_{\alpha}K$ + $K'$ $\in$$ \mathfrak{M}_\chi$ 
, where $K$ and $K'$  commute with $x_{\alpha}$ modulo $\frak {M}_\chi $.\newline

 We thus deal with and assume  $x_{- \alpha}$$x_{\alpha}$$K$ + $x_{-\alpha}$$K''$ $\in$ $\mathfrak{M}_\chi$.\newline

From $w_{\alpha}:=(h_{\alpha}+ 1)^2 +4 x_{-\alpha}x_{\alpha}\in $ the center of  $u(\frak{sl}_2(F)), $ we get $ 4^{-1}\{w_{\alpha}- (h_{\alpha}+ 1)^2\}K+ x_{-\alpha}K'' \equiv 0 $  modulo $\frak M_\chi$.\newline

If $x_{-\alpha}^p\equiv c $ which is a constant, then \newline

$(\ast \ast)4^{-1}x_{-\alpha}^{p-1}\{w_{\alpha}- (h_{\alpha}+ 1)^2\}K+ cK''\equiv 0 $\newline

 is obtained.
\newline
From $(\ast),(\ast \ast)$, we have\newline

 $4^{-1}x_{-\alpha}^{p-1}\{w_{\alpha}- (h_{\alpha}+ 1)^2\}K- cx_{\alpha}K\newline
\equiv 0$ modulo $\frak M_\chi$.\newline

Multiplying $x_{\alpha}^{p-1}$ to this equation,we obtain \newline

 $(\ast \ast \ast)4^{-1}x_{\alpha}^{p-1}x_{-\alpha}^{p-1}\{w_{\alpha}- (h_{\alpha}+ 1)^2\}K- cx_{\alpha}^pK\equiv 0.  $\newline

By making use of $w_{\alpha}$ and the properties of prime ideal , we may deduce from $(\ast \ast \ast)$ an equation of the form \newline
( a polynomial of degree $\geq 1$ with respect to  $ h_{\alpha})K- cx_{\alpha}^pK\equiv 0.   $\newline

Finally if we use conjugation and subtraction consecutively,then we are led to a  contradiction $K\in \frak M_\chi.$
It may be necessary  for us to change the role of $K$ and $K'$ to obtain the absurdity alike.\newline

Next  if $\chi (x_{\alpha}) \neq 0$ or $\chi (x_{- \alpha}) \neq 0$, then our assertion is evident from the proof mentioned above.\newline

So for the case  $\chi(h_{\alpha})\neq 0$, we may assume  $x_{\alpha}^p \equiv x_{- \alpha}^p \equiv 0$ modulo $\frak M_\chi$.\newline

We may see without difficulty that $\frak B$  as in the previous propositions  is a linearly independent set in $u(L)$ by virtue of P-B-W theorem.\newline

We shall prove that a nontrivial linearly dependent  equation leads to absurdity. We assume first that we have a dependence equation which is of least degree with respect to $h_{\alpha_j}\in H$ and the number of whose highest  degree terms is also least.\newline

In case it is conjugated by $g_\alpha$, then there arises a nontrivial dependence equation of lower degree than the given one,which contradicts to our assumption.\newline

Otherwise  it reduces to one of the following forms, so that we have only to prove that \newline

(i)$x_{\beta}K+ K'\in \frak M_\chi$,

(ii) $g_\alpha K+ K'\in \frak M_\chi$ \newline

lead to a contradiction, where both $K$ and $K'$ commute with $x_{\pm \alpha}$  modulo $\frak M_\chi $. In particular $K$ commute with $g_\alpha$.
\newline

For the case (i), we may change it to the form $x_{\alpha}K+ K''\in \frak M_\chi$ for some $K''$ commuting with $x_\alpha$ modulo $\frak M_\chi$.\newline

So we have $x_\alpha^p K+ x_\alpha^{p-1}K''\equiv 0$, and hence $x_\alpha^{p-1}K''\equiv 0$.\newline

Subtracting from this $x_{-\alpha}x_\alpha K+ x_{-\alpha}K''\equiv 0$, we get \newline

$-x_{-\alpha}x_\alpha K+ g_\alpha K'' \equiv 0$. Recall here that $g_\alpha$ is invertible and $w_\alpha$ belongs to the center of $u(\frak {sl}_2 (F))$.

So we get  $4^{-1}\{(h_\alpha +1)^2- w_\alpha\}K+ g_\alpha K''\equiv 0$, and hence\newline

 $(\ast) g_\alpha^{p-1} 4^{-1}\{(h_\alpha + 1)^2- w_\alpha \}K+ cK'' \equiv 0$\newline

 is obtained and from the start equation we have \newline

$(\ast \ast)cx_\alpha K+ c K''\equiv 0$, where $g_\alpha^p- c \equiv 0$.\newline

Subtracting $(\ast \ast)$ from $(\ast)$, we have $4^{-1}g_\alpha^{p-1}\{(h_\alpha+ 1)^2- w_\alpha\}K- cx_\alpha K \equiv 0$.\newline

Multiplying this equation by $g_\alpha^{1-p}$ to the right, we obtain $4^{-1}g_\alpha^{p-1}\{(h_\alpha+ 1)^2- w_\alpha\}g_\alpha^{1-p}K- cx_\alpha g_\alpha^{1-p}K \equiv 0$ \newline

We thus have $4^{-1}\{(h_\alpha+ 1- 2)^2- w_\alpha\}K- x_\alpha g_\alpha K \equiv 0$.

So it follows that $4^{-1}\{(h_\alpha -1)^2- w_\alpha\}K+ x_\alpha x_{-\alpha}K \equiv 0 $.\newline

Next multiplying $x_{-\alpha}^{p-1}$ to the right of this last equation, we obtain $\{(h_{\alpha}- 1)^2- w_\alpha\}K x_{-\alpha}^{p-1}\equiv 0$.
Now multiply $x_\alpha$ in turn consecutively to the left of this equation until it becomes of the form \newline

( a nonzero polynomial of degree $\geq 1$ with respect to $h_\alpha)K $\newline
$\in \frak M_\chi$. \newline

By making use of  conjugation and subtraction consecutively and by virtue of  properties of  prime ideals, we are led to a contradiction.$K \in \frak M_\chi$.\newline

Finally for the case (ii),we consider $K+ g_\alpha^{-1}K' \in \frak M_\chi$.  So we have $x_\alpha K+ x_\alpha g_\alpha^{-1} K' \equiv 0$ modulo $\frak M_\chi$.

By analogy with the argument  as in the case (i), we obtain a contraiction $K \in \frak M_\chi$.

\end{proof}

We are now concerned with some prototype of  exemplary subregular points before closing  this section.

Sometimes we are preferably fond of the $A_l$- type   because of the simplicity of  its root system. By making use of this type we would like to easily figure out  in this section an indecomposable  Lie algebra having  certain subregular points.

\begin{prop}
Let $L$ be any Lie algebra of  $A_l$-type over an algebraically closed field $F$ of characteristic $p\geq 7$.
Then  L becomes a Hypo- Lie algebra, which is  also clear  for other classical  type  Lie algebras.
\end{prop}

\begin{proof}
Refer to  previous propositions  of this paper and the chapter4 in [KY-4].
 \end{proof}
Because we proved  that  all modular $A_l$- type Lie algebras for  characteristic $p\geq 7$ have no subregular point,i.e., $S(L,p,\chi)= \phi$, we need to find out  in connection with  this $A_l$- type Lie algebra some example of nontrivial indecomposable modular Lie algebra which has no subregular point.\newline

By making use of the fact that $sl_2(F)$ has no subregular point for $p>2$(refer to [RS],[KY-4]), we can  extend this one to $L_5$ which is explained below.\newline

We shall show that 
a certain indecomposable Lie algebra over $F$ has a subregular point even for $p\geq 7$, whose fact is notable related to Kim's conjecture[KY-6].\newline

\begin{prop}Let $L_5$ be a Lie algebra generated by\newline

 $x_{\epsilon_j- \epsilon_2},x_{\epsilon_2- \epsilon_j},
h_{\epsilon_2- \epsilon_j},x_{\epsilon_1- \epsilon_2},
x_{\epsilon_1- \epsilon_j }(j\neq1,2)$ over an algebraically closed field $F$ of characteristic $p>2$, where all these generators are included in classical $A_l$-type Lie algebra. Suppose that $\chi$ is a character of $L_5$ s.t. $\chi(x_{\epsilon_j- \epsilon_2})\neq 0$ or $\chi(x_{\epsilon_2- \epsilon_j})\neq 0$
or $\chi(h_{\epsilon_2- \epsilon_j})\neq 0$.\newline

Then we  claim  that $L_5$ is centerless and indecomposable with $S(L_5,p,\chi)\neq \phi$.
\end{prop}

\begin{proof}

Because  $L_5$ is a proper sub-Lie algebra of an $A_l$-type Lie algebra, it is easy to see that $L_5$  is a centerless and   indecomposable modular Lie algebra.  Let $< x _{\epsilon_1- \epsilon_2}, x_{\epsilon_1- \epsilon_j}>_F$ denote the sub- Lie algebra of $L$ generated by $x_{\epsilon_1- \epsilon_2}$ and $x_{\epsilon_1- \epsilon_j}$ over $F$.\newline

It is just a proper ideal of $L_5$, so that we have a projective homomorphism $\pi$ of Lie algebras\newline

$\pi: L_5 \rightarrow 
L_5   /<x_{\epsilon_1- \epsilon_2}, x_{\epsilon_1- \epsilon_j}>_F$.\newline

 Since we have $[Q(u(L_5)): Q(\frak Z(u(L_5)))]\geq p^3$, it follows that    $[Q(u(L_5)): Q(\frak Z (u(L_5)))]=  p^4$, and so  the maximal   dimension of  irreducible  $L_5$-modules is $p^2$.\newline

 Hence evidently there exists a subregular point for this $L_5$,i.e., $S(L_5,p,\chi) \neq \phi$  in view of  definitions  in section 3.

\end{proof}

We remark here that the codomain of the projective homomorphism $\pi$ is isomorphic to $sl_2(F)$ and the Lie algebra $L_5$  over $F$ in the  domain has dimension 5.\newline

In [KWa] we made a remark  conjecturing that some extended Lie algebras of  classical modular Lie algebras must be  Hypo-Lie algebras   relating to subregular points.\newline

 If  we have regular points  almost everywhere except for a finite number of subregular points and $p$-points related to a Lie algebra $L$ included in a classical modular Lie algebra and if we can find Lee's bases associated to  all these regular points , then we called $L$  a  $\textit{Hypo- Lie algebra}$. \newline

So in connection with such concepts, the Lie algebra $L_5$  above  in this paper is no way  Hypo- Lie algebra.

We thus conclude that $L_5$ closely related to $A_l$- type  is a nontrivial centerless indecomposable Lie algebra having  infinitely many subregular points over an algebraically closed field of characteristic $p>2$.

\section{proof of $P\neq NP$ and a conjecture}

Now in this section we come back to our main story related to a conjecture.\newline

Because the $B_l$- type Lie algebra looks like  a  general prototype  out of all classical  Lie algebras and we used it already in [NWK-1,2], so we would like  to still use it preferably. \newline

 For a while we fix an algebraically closed field $F$ of characteristic $p\geq 7$ unless otherwise specified until we remark on  our conjecture related  to $NP$completeness.\newline

 As we are well aware, there exists the universal Casimir element $s\in u(L)$ for any $B_l$- type classical Lie algebra $L$. Let the rank  $l$ of $L$ is $l \geq 3$,i.e., the dimension of  CSA $\geq 3.$\newline

We may express the universal Casimir element $s$ as $ s= \sum_{\alpha\in \Phi^+}a_{\alpha}x_{\alpha} x_{-\alpha}+  \sum_{i=1}^l b_ih_{\alpha_i}+ \sum_{i,j}a_{ij} h_{\alpha_i} h_{\alpha_j} $, where $0\neq b_i \in F$ , $a_{\alpha}, a_{ij}\in F$, $H$ is a  Cartan subalgebra(abbreviated as CSA) with a basis $\{h_{\alpha_i} \vert 1\leq i \leq l \}$, and $\{x_{\alpha}, \alpha\in \Phi; h_{\alpha_i}, 1\leq i \leq l\}$ is  the standard Chevalley basis of $L$ with a root system $\Phi$ including $\Phi^+ $as a set of positive roots. \newline

Here we rearrange the Chevalley basis as\newline
 $\{h_{\alpha_1}, h_{\alpha_2},\cdots, h_{\alpha_l} ; x_{\alpha_1},x_{-\alpha_1},\cdots, x_{\alpha_l}, x_{-\alpha_l},$ \newline  $x_{\alpha_{l+1}}, x_{-\alpha_{l+1}}, \cdots, x_{\alpha_m}, x_{-\alpha_m}\}$,
where dim $L =n =l+ 2m$, and $l= rank L= dim H$ over $F$.\newline

Clearly we have $x_{\alpha_i}^{[p]}= x_{-\alpha_i}^{[p]}= 0$  for $ 1\leq i \leq m$ and $h_{\alpha_j}^{[p]}= h_{\alpha_j}$ for $1\leq j \l \leq l$. We shall denote the center of  $u(L)$ simply by $\frak Z$ as before.

\begin{prop}
We have the following in the quotient algebra $Q(u(L))$:\newline

(i)$dim_{Q(\frak Z)} Q(\frak Z)(h_{\alpha_1},\cdots, h_{\alpha_l})= p^l$,
(ii)$Q(\frak Z)(h_{\alpha_1},\cdots, h_{\alpha_l})$ becomes a Galois field over $Q(\frak Z)$.
\end{prop}

\begin{proof}
Since $Q(u(L))$ becomes a central simple Artinian algebra over $Q(\frak Z)$ which is actually a division algebra,we obtain (i) immediately from the Noether-Skolem theorem and from the identities such as $h_{\alpha_i} x_{\alpha_i}= x_{\alpha_i}(h_{\alpha_i}+ 2)$ for $i=1,2,\cdots, l$ and $[Q(\frak Z)(h_{\alpha_1},
h_{\alpha_2},\cdots, h_{\alpha_l}): Q(\frak Z)]= [Q(\frak Z)(h_{\alpha_1}, h_{\alpha_2}, \cdots, h_{\alpha_l}) :    Q(\frak Z)(h_{\alpha_1}, h_{\alpha_2}, \cdots, h_{\alpha_{l-1}}) ]\cdot$ \newline

  $[ Q(\frak Z)(h_{\alpha_1}, h_{\alpha_2}, \cdots, h_{\alpha_{l-1}})  :  Q(\frak Z)(h_{\alpha_1}, h_{\alpha_2}, \cdots, h_{\alpha_{l-2}})]\cdots $\newline

$[Q(\frak Z)(h_{\alpha_1}): Q(\frak Z)]  $ for $i= 1,2,\cdots, l$.\newline

(ii) is an immediate consequence  of (i) considering that  each bracket [ ] in the above factorization  gives rise to  $p$- distinct  conjugates of $h_{\alpha_i}$ 
for $j= 1,2,\cdots, l$.

\end{proof}                                

\begin{prop}
We may obtain the irreducible polynomial $Irr(s,\mathcal O(L))$ of $s$ over the $p$-center $\mathcal O(L)$ by expanding out  the following norm

$(\ast)$ $N_{Q(\frak Z)}^{Q(\frak Z)(h_{\alpha_1},\cdots, h_{\alpha_l})}$$\{s_i-  (b_1 h_{\alpha_1}+ \cdots b_l h_{\alpha_l}+ \sum\limits_{i,j}a_{ij}h_{\alpha_i} h_{\alpha_j}) \}= N_{Q(\frak Z)}^{Q(\frak Z)(h_{\alpha_1},\cdots, h_{\alpha_l})}(\sum \limits_{\alpha \in \Phi^+} a_{\alpha} x_{\alpha} x_{-\alpha})  $ . Moreover we have $deg(Irr(s, \mathcal O(L))= p^l $ and $s$ becomes separable over $\mathcal O (L)$ and so over $Q(\mathcal O (L))$.                                                                                                                                
\end{prop}

\begin{proof}

Firstthing we should perceive that\newline

 the left hand side of $(\ast)= s^{p^l}+ \hat {a}_1 s^{p^l-1}+  $ $\cdots +\hat {a}_{p^l-1} s+ \hat {a}_{p^l} $ for some $\hat {a}_i \in  
 Q(\frak Z)(h_{\alpha_1}, h_{\alpha_2}.\cdots, h_{\alpha_l})$ for $i= 1,2,\dots, p^l$. \newline

Here we contend that $\hat {a}_i \in Q(O(L))$ in fact. We choose any distinct  $p^l$- elements $s_i\in \mathcal O(L)$ and take norms such as $N_{Q(\frak Z)}^{Q(\frak Z)(h_{\alpha_1},\cdots, h_{\alpha_l})} $$\{s_i- (b_1h_{\alpha_1}+\cdots +b_l h_{\alpha_l}+ \sum\limits_{i,j}a_{i,j} h_{\alpha_i} h_{\alpha_j})\}=  N_{Q(\frak Z)}^{Q(\frak Z)(h_{\alpha_1},\cdots, h_{\alpha_l})} \{s_i- s+ (\sum \limits_{\alpha \in \Phi^+} a_{\alpha} x_{\alpha} x_{-\alpha}) \} $.\newline

We know that  the Noether -Skolem theorem allows us to extend every automorphism of $Q(\frak Z)(h_{\alpha_1},\cdots, h_{\alpha_l})$ to an inner  automorphism of  $Q(u(L)).$\newline

 So if we conjugate this by some elements in $Q(u(L))$, then we obtain $p^l$-  distinct  automorphisms of $Q(\frak Z) (h_{\alpha_1},\cdots, h_{\alpha_l})$ over $Q(\frak Z)$.\newline

Next  we need to note that $[Q(\mathcal O (L)) (h_{\alpha_1},\cdots, h_{\alpha_l}): Q(\mathcal O(L)]= [Q(\frak Z)(h_{\alpha_1},\cdots, h_{\alpha_l}): Q(\frak Z)]= p^l$ and that  isomorphisms of $Q(\frak Z)(h_{\alpha_1},\cdots, h_{\alpha_l})$ over $Q(\frak Z)$ are the same as those of $Q(\mathcal O(L))(h_{\alpha_1},\cdots, h_{\alpha_l})$ over $Q(\mathcal O (L))$.\newline

So we obtain

 $N_{Q(\frak Z)}^{Q(\frak Z)(h_{\alpha_1},\cdots, h_{\alpha_l})} $$\{s_i- (b_1h_{\alpha_1}+\cdots +b_l h_{\alpha_l}+ \sum\limits_{i,j}a_{i,j} h_{\alpha_i} h_{\alpha_j})\}=  
 N_{Q(\mathcal O(L))}^{Q(\mathcal O(L))(h_{\alpha_1},\cdots, h_{\alpha_l})} $ $\{s_i- (b_1h_{\alpha_1}+\cdots +b_l h_{\alpha_l}+ \sum\limits_{i,j}a_{i,j} h_{\alpha_i} h_{\alpha_j})\}\in Q(\mathcal O(L))  $, \newline

which should boil down to  the form
$N_{Q(\frak Z)}^{Q(\frak Z)(h_{\alpha_1},\cdots, h_{\alpha_l})} \{s_i- s+ (\sum \limits_{\alpha \in \Phi^+} a_{\alpha} x_{\alpha} x_{-\alpha}) \} $.\newline

On the other hand, if we put\newline

$ N_{Q(\mathcal O(L))}^{Q(\mathcal O(L))(h_{\alpha_1},\cdots, h_{\alpha_l})} $ $\{s_i- (b_1h_{\alpha_1}+\cdots +b_l h_{\alpha_l}+ \sum\limits_{i,j}a_{i,j} h_{\alpha_i} h_{\alpha_j})\}=: s_i^{p^l}+ \hat {b}_1s_i^{p^{l-1}}+\cdots, \hat{b}_{p^l-1} s_i+ \hat{b}_{p^l}= k_i  $ for some $k_i\in Q(\mathcal O(L))$ and $\hat{b}_j\in Q(\mathcal O(L))(h_{\alpha_1},\cdots, h_{\alpha_l})$ for $j=1,2.\cdots, p^l$, \newline

then we have a system  in  indeterminates  $\hat{b}_j$ of $p^l$- nonhomogeneous linear equations with the determinant  of coefficients:
$\det A = \begin{vmatrix}
1 & s_1\cdots  s_1^{p^l-1} \\
 \vdots & \ddots  \vdots  \\  
1 & s_{p^l}\cdots  s_{p^l}^{p^l-1}
\end{vmatrix}$\newline

$=\Pi_{i<j} (s_j- s_i)\neq 0 $, which turns out to be the so called $Vandermonde$ determinant.\newline

If we make use of the Cramer's formula, then we  may have the solutions $\hat{b}_j\in Q(\mathcal O (L))$ of the system above. So we get  $\hat{a}_j= \hat{b}_j$ for $j=1,2,\cdots, p^l$, and hence $\hat{a}_j\in Q(\mathcal O(L))$ follows. \newline

By the way , since $s$ is integral over $\mathcal O (L)$, we have that \newline

the right hand side of $(\ast)= N_{Q(\frak Z)}^{Q(\frak Z)(h_{\alpha_1},\cdots, h_{\alpha_l})} (\sum_{\alpha\in \Phi^+} a_{\alpha} x_{\alpha} x_{-\alpha})$ also belongs to $Q(\mathcal O (L))$ which is secured by the following lemma.

\end{proof}

\begin{lem}
If we suppose that  $s$ satisfies an algebraic equation of the form  $f(\mathbb X):= \mathbb X^{p^l}+ \hat{a}_1\mathbb X^{p^l-1}+\cdots +\hat{a}_{p^l-1}\mathbb X+ \hat{a}_{p^l}= 0$  with $\hat {a}_j\in Q(\mathcal O(L))$ for $1\leq j\leq p^l-1,$ which is the same form as in the above argument for the left hand side  of norm equation $(\ast)$, \newline

then we have that  $\hat{a}_{p^l}$ also belongs to  $Q(\mathcal O(L))$

\end{lem}

\begin{proof}
Thanks to proposition5.1, we may have $p^l$- distinct automorphisms $\sigma_i$ of $Q(\frak Z)(h_{\alpha_1},\cdots, h_{\alpha_l})$ over $Q(\frak Z)$ , which are represented by inner automorphisms of $Q(L)$.\newline

 Next the integral equation of $s$ over  $\mathcal O(L)$ is just $Irr(s, Q(\mathcal O(L))$ itself because $\mathcal O(L)$ becomes the Noether normalization of $\frak Z$ and hence a unique factorization domain.\newline

Now the map $s \mapsto \sigma_i ( {b}_1 h_{\alpha_1} + \cdots + b_l h_{\alpha_l}+ \sum\limits_{i,j}a_{ij}h_{\alpha_i} h_{\alpha_j}) +  \sigma_j(\sum \limits_{\alpha \in \Phi^+} a_{\alpha} x_{\alpha} x_{-\alpha}) $ for $1\leq i,j \leq p^l$\newline

 gives rise to an isomorphism of  the algebra $Q(\mathcal O (L)(h_{\alpha_1},\cdots, h_{\alpha_l}) (s)$ over $Q(\mathcal O(L))$.\newline

 Hence the irreducible polynomial $Irr(s, Q(\mathcal O(L)))$ must divide the polynomial $f(\mathbb X)$, and  thus we have that $f(\mathbb X)= Irr(s, Q(\mathcal O(L)))\cdot g(\mathbb X)$ for some unique $g(\mathbb X):= \mathbb X^{p^l- p^{l'}}+ c_1\cdot \mathbb X^{p^l- p^{l'}-1}+ \cdots+ c_{p^l-p^{l'}}$ with $c_j\in Q(\frak Z)$ and $p^l- p^{l'}\geq 0$. \newline

Seeing that $c_j$'s  are uniquely determined  only by the coefficients of $Irr(s, Q(\mathcal O(L)))$ and those  of terms\newline
 $ \mathbb X^{p^l}, \mathbb X^{p^l-1}, \cdots, \mathbb X^{p^{l'}+ 1}, \mathbb X^{p^{l'}}$

in $f(\mathbb X)$ and that those coefficients belong to $Q(\mathcal O(L))$, we know that our claim is proven.

\end{proof}

Now coming back to our main proof, we must still show that  the algebraic equation of $s$  obtained just above is the $Irr(s, Q(\mathcal O (L)))$, which  is also the integral equation of $s$ over $\mathcal O(L).$  Such a fact  is attributed to the following another lemma.

\begin{lem}In this lemma only,\newline

let $E$ be  a finite extension field of arbitrary field $F$ with nonzero characteristic $p$ and let $\alpha,\beta, \gamma$ be elements of $E- F$ with $[F(\alpha)(\gamma): F]= p^n$. \newline

Suppose  that $f(\mathbb X):= Irr(\alpha, F)= (\mathbb X- \sigma_1 (\alpha))\cdots (\mathbb X- \sigma_{p^m}(\alpha))$ for  some distinct $\sigma_i(\alpha)\in F(\alpha)$ for $i=1,2,\cdots, p^m\leq p^n$  with $\sigma_i\in {isomorphisms\; of \; E \, over\,  F}\; \newline
and\; that \; \beta= \alpha+ \gamma$  is an element  such that  $\forall \sigma_i,\beta$ satisfies $\sigma_i(\beta)= \beta.$ \newline

Suppose further that  $\Pi_{i= 1}^{p^m} \sigma_i (\gamma) \in F$ and $g(\mathbb X):= \Pi_{i=1}^{p^m} (\mathbb X- \sigma_i(\alpha))- \Pi_{i=1}^{p^m} \sigma_i (\gamma)\in F[\mathbb X].$\newline

 Then we have $g(\mathbb X)= Irr(\beta, F)$ which is separable over $F$.

\end{lem}
\begin{proof}
It is not difficult  to see that $\gamma \notin F(\alpha)$, and hence  there exists at least $p^{m+1}$- distinct isomorphisms of  $F(\alpha) \ni \beta$ over $F$.\newline

Next we consider  a field lattice diagram as follows:\newline

        \qquad \qquad  \qquad      $F(\alpha)(\gamma) \ni \beta$\newline

 at least p-dimensional    $\vert$

 \qquad \qquad  \qquad \qquad $F(\alpha)$\newline

\;\qquad$p^m$-dimensional   $\vert$

\qquad \qquad  \qquad  \qquad\;\;   $F$

Given any nontrivial isomorphism $\tau$ of $F(\alpha)(\gamma)$ over $F(\alpha)$, we must have $\tau(\beta)= \tau(\alpha+ \gamma)= \tau( \sigma_i(\alpha)+ \sigma_i(\gamma))= \sigma_i(\alpha)+ \overline {\tau}(\gamma)$ for some isomorphism $\bar {\tau}$ of $F(\alpha)(\gamma)$ over $F(\gamma)$ and $\forall i$ with $1\leq i \leq p^m$.\newline

Because $F(\alpha)$ is a Galois extension of $F$, we see that there are at least $p^m$- distinct  conjugates of $\beta$ over $F$. However the degree of $g(\mathbb X)$ becomes  $deg g(\mathbb X)= p^m$ and  clearly $g(\beta)= 0$, so that  $Irr(\beta, F)$ also has  the same degre $p^m$. In addition no doubt $g(\mathbb X)$ becomes separable over $F$.
\end{proof}

Now so as to finish our proof of the main proposition5.2, we put $F:= Q(\mathcal O(L)), \beta:=s, \alpha:= b_1 h_{\alpha_1} + \cdots + b_l h_{\alpha_l}+ \sum\limits_{i,j}a_{ij}h_{\alpha_i} h_{\alpha_j}, \gamma:= \sum \limits_{\alpha \in \Phi^+} a_{\alpha} x_{\alpha} x_{-\alpha} $ and make use of automorpisms $\{\sigma_i \vert i= 1,2,\cdots,p^l \}$ of $ Q(\frak Z)(h_{\alpha_1},\cdots, h_{\alpha_l})$ over $Q(\frak Z)$ which are the same  automorphisms of  $Q(\mathcal O(L))(h_{\alpha_1},\cdots, h_{\alpha_l})  $
over $Q(\mathcal O(L))$. Finally  applying the preceding lemma , we are done after all.

\begin{prop}
We obtain $\frak Z= \mathcal O(L)[s]$.

\end{prop}

\begin{proof}
According to proposition5.2, we must have $Q(\frak Z)= Q(\mathcal O(L)[s])$. Since  $\frak Z$ becomes a finitely generated $\mathcal O(L)$- module , we have that  $\mathcal O(L)[s]$ is completely closed in  $\frak Z$,i.e.,any nontrivial quotients of $\mathcal O(L)[s]$ is not contained in $\frak Z- \mathcal O(L)[s]$.
We can explain this explicitly as follows.\newline

Assume first that some $\mu \in \frak Z- \mathcal O(L)[s]$ satisfies an equation of the form $\mu\cdot f= \hat {f}$ for some distinct polynomials  $f,\hat {f}$ in $F[x_{\alpha_1}^p, x_{-\alpha_1}^p,h_{\alpha_1}^p- h_{\alpha_1},\cdots, x_{\alpha_l}^p, x_{-\alpha_l}^p, h_{\alpha_l }^p- h_{\alpha_l}$,\newline

$ x_{\alpha_{l+1}}^p, x_{-\alpha_{l+1}}^p,\cdots, x_{\alpha_m}^p, x_{-\alpha_m}^p,s]$\newline

 satisfying  that  a nontrivial fraction $\mu= \hat{f}/f$ is  reduced and  that $\mu,s$ are integral over $\mathcal O(L)$.\newline

 Noticing that $ x_{\alpha_1}^p, x_{-\alpha_1}^p,h_{\alpha_1}^p- h_{\alpha_1},\cdots, x_{\alpha_l}^p, x_{-\alpha_l}^p, h_{\alpha_l }^p- h_{\alpha_l},$ \newline

 $ x_{\alpha_{l+1}}^p, x_{-\alpha_{l+1}}^p,\cdots, x_{\alpha_m}^p, x_{-\alpha_m}^p$ are all  algebraically independent  variables and that  the fraction
$\hat{f}/f$ should always be defined for all these variables , we perceive that  the denominator $f$ must be 1, a contradiction to our assumption. So we obtain our required result.

\end{proof}
\begin{prop}
Suppoose in particular that $L$ is $B_2$- type Lie algebra over  an algebraically closed field of characteristic  $p\geq 7$; \newline

then we have that the quotient algebra $Q(u(L))$ becomes a crossed product  in the Brauer group $B(Q(\frak Z))$ of  $Q(\frak Z)(\alpha) (x_{\alpha_1}x_{-\alpha_1}+ \cdots, x_{\alpha_m}x_{-\alpha_m})$ with $\alpha$ as in the proof of proposition5.4.

\end{prop}

\begin{proof}
We must notice first  that $dim L= 5\times rank L$.\newline

 So  we see immediately  that $Q(u(L))$ becomes  a crossed product  in the Brauer group $B(Q(\frak Z))$ of  $Q(\frak Z)(\alpha)(x_{\alpha_1}x_{-\alpha_1}+ \cdots, x_{\alpha_m}x_{-\alpha_m})$  keeping in mind  the above propositions.
\end{proof}

Now from here on, we fix our field $F$ to be algebraically closed field  of characteristic $p\geq 7$.\newline

 Furthermore we assume that  this field $F$ is  a $p$-adic field with its valuation ring $A$ whose maximal ideal is denoted by $\frak p$.\newline

 For the purpose of constructing  our  serious conjecture, relating to $NP$-hardness we consider only such a field henceforth.\newline

We let $L$ be any $B_l$-type Lie algebra with rank $L$= $l\geq3$ over an algebraically closed field $F$ of a  $p$-adic field of characteristic $p\geq 7$.
Let $A$ be the valuation ring of $F$ with $\frak p$ the maximal ideal of $A$  definitely once and for all.\newline

The Steinberg module $V$ associated with a regular $p$- point  arises from an  irreducible representation $\rho_  \chi: u(L)\rightarrow End_F(V)$ with $\chi=0$ and dim  $V=p^m$,where dim $L=n= 2m+ l$.\newline

 We shall use such a particular module  for our conjecture as in  [NWK-1,2].

\begin{prop}Let notations be as in the refernce [NWK-1,2].\newline

Then the counting problem in the reference [NWK-1,2] is an $NP$- problem. However  it cannot belong  to the  $P$- class,i.e., it is not a  $P$-problem.\newline

\end{prop}

\begin{proof}

We made use of a choice fuction which gives rise to  a system of nonhomogeneous linear equations including arbitrary random variables of cardinality $p^{2m}$, where $[Q(u(L)): Q(\frak Z(u(L)))]= p^{2m}$. \newline

Here $L$ is an $n$- dimensional  $B_l$- type classical Lie algebra of rank $l$ with $n=l+ 2m$ over an algebraically closed field of a $\frak p$- adic field  of characteristic $p \geq7$ and $Q$ indicates a quotient algebra of a noncommutative algebra. \newline

Let $s$ be the universal Casimir element of $u(L)$. Any  maximal ideal $\frak M_\chi$ of  $u(L)$ associated with a character $\chi \in L^{\ast}$ must contain an  ideal of the form  \newline

 $\{ \sum_{i=1}^{n}u(L)(\breve {x}_i- \breve {x}_i^{[p]}- \xi_i)  \}+ u(L)(s- \xi_{n+1}) $, where $\xi_j$'s are numbers in $F$, and $\breve {x}_i$ represents the basis elements of $L$ written  as ordered elements like \newline
 $x_{\alpha_1},x_{\alpha_ {-1}},h_{\alpha_1},\cdots, x_{\alpha_i },x_ {-\alpha_i}, h_{\alpha_i},x_{\alpha_{i+1}},x_{-\alpha_{i+1}},\cdots,x_{\alpha_m},x_{-\alpha_m} $.\newline

Moreover $dim_Fu(L)/\frak M_{\chi}\leq p^{n-l}$ \newline
considering $[Q(u(L)): Q(\frak Z(u(L)))]= p^{n-l}= p^{2m}$.\newline

For the purpose of  proving  our claim, we may make use our earlier propositions  in section3. However  we would like to take, a regular $p$- point associated with a $B_l$- type classical Lie algebra, which was used earlier in  [NWK-1,2].\newline

We see the Steinberg module $V$ over $L$ has the maximal dimension $p^m$ with its associated point $(0,\cdots,0, \xi_{n+1})$ and with its associated character $\chi = 0$.\newline

So  the Steinberg module  is closely connected with the restricted representation. We may assume $\xi_{i+1}= 0$ without loss of generality.\newline

We put $S_{\alpha_i} := Fx_{\alpha_i}+ Fx_{-\alpha_i}+ h_{\alpha_i}$ for $1\leq i \leq m, g_{\alpha_i}:= x_{\alpha_i}^{p-1}- x_{-\alpha_i}$ for $1\leq i \leq l$ and $w_{\alpha_i}:= (h_{\alpha_i}+ 1)^2+ 4x_{-\alpha_i}x_{\alpha_i}$ for $1\leq i \leq m$.   \newline

There exists $w_{\alpha_j}$ for $1\leq j \leq m$ with $2m= n-l$ such that $w_{\alpha_j}$ acts as zero for each irreducible quotient $S_{\alpha_j}$- module for  the composition series of $S_{\alpha_j}$- module $V$.\newline

 In this case $g_{\alpha_j}$ is invertible  in $u(L)/\frak M_{\chi}$.We exibit the basis of $u(L)/m_{\chi}$ as follows. \newline

As a preliminary step we put 
$ B_i:=  b_{i_1}h_{\epsilon_1- \epsilon_2}+ b_{i_2}h_{\epsilon_2- \epsilon_3}+ \cdots+ b_{i_l}h_{\epsilon_l- \epsilon_{l+1}} $for $l\geq 2$ and for  $1\leq i \leq 2m$, where $(b_{i_1},\cdots, b_{i_l}) \in F^l$ is chosen so that any $(l-) B_j$'s are linearly independent in $P^l(F)$ and the $\frak B$ below becomes an 
$F$-linearly independent set  in $u(L)$ if necessary for this particular $\alpha_j$ and $B_i \cdot x_{\alpha_j}\not\equiv x_{\alpha_j}\cdot B_i$ modulo $\frak M_\chi$.\newline

(I) We assume first $\alpha_j$ is a long root.\newline

We may assume $\alpha_j= \epsilon_1- \epsilon_2$ without loss of generality. \newline

Next we put
$\frak B:= \{(B_1+ A_{\epsilon_1-\epsilon_2})^{i_1}\otimes(B_2 +A_{-(\epsilon_1- \epsilon_2)})^{i_2}\otimes \cdots \otimes(B_{2i-2}+A_{-(\epsilon_{l-1}-\epsilon_l)})^{i_{2l-2}}\otimes (B_{2l-1}+ A_{\epsilon_l})^{2l-1}\otimes(B_{2l}+ A_{-\epsilon_l})^{i_{2l}}\otimes(\otimes _ {j=2l+1}^{2m}(B_j+ A_{\alpha_j})^{i_j})\}$ for $0\leq i_j \leq p-1$,\newline

 where we designate

$A_{\epsilon_1-\epsilon_2}=g_{\epsilon_1-\epsilon_2},$\newline

$A_{\epsilon_2-\epsilon_1}= g_{\epsilon_1- \epsilon_2}^2 (C_{\epsilon_2-\epsilon_1}+ w_{\epsilon_1- \epsilon_2}),$\newline

$A_{\epsilon_1 \pm \epsilon_3}= g_{\epsilon_1-\epsilon_2}^{3  or 4}(C_{\epsilon_1\pm \epsilon_3}+ x_{\epsilon_1\pm  \epsilon_3}x_{-(\epsilon_1\pm \epsilon_3)}\pm x_{\epsilon_2\pm \epsilon_3}\cdot x_{-(\epsilon_2\pm \epsilon_3)}),$\newline

$A_{-\epsilon_1- \epsilon_3} = g_{\epsilon_1-\epsilon_2}^5(   C_{-\epsilon_1- \epsilon_3}+ x_{\epsilon_1+ \epsilon_3}\cdot x_{-\epsilon_1- \epsilon_3}\pm x_{\epsilon_2+ \epsilon_3}\cdot x_{-\epsilon_2- \epsilon_3},$  \newline

$ A_{\epsilon_3- \epsilon_1}= (C_{\epsilon_3- \epsilon_1}+x_{\epsilon_1- \epsilon_3}\cdot x_{-(\epsilon_1- \epsilon_3)}\pm x_{\epsilon_2- \epsilon_3}\cdot x_{-(\epsilon_2- \epsilon_3)}), $\newline

$A_{\epsilon_2}= x_{\epsilon_3}(C_{\epsilon_2}+x_{\epsilon_1- \epsilon_3}\cdot x_{-(\epsilon_1- \epsilon_3)}\pm x_{\epsilon_2- \epsilon_3}\cdot x_{-(\epsilon_2- \epsilon_3)}),$\newline

$ A_{-\epsilon_2}= x_{-\epsilon_3}^2 (C_{-\epsilon_2}+x_{\epsilon_1- \epsilon_3}\cdot x_{-(\epsilon_1- \epsilon_3)}\pm x_{\epsilon_2- \epsilon_3}\cdot x_{-(\epsilon_2- \epsilon_3)}),$\newline

 $A_{\epsilon_1}=  x_{\epsilon_3}^2 (C_{\epsilon_2- \epsilon_3)}+ x_{\epsilon_1- \epsilon_3}\cdot x_{-(\epsilon_1- \epsilon_3)}\pm x_{\epsilon_2- \epsilon_3}\cdot x_{-(\epsilon_2- \epsilon_3)}),$\newline

$ A_{-\epsilon_1}= x_{-\epsilon_3}(C_{-\epsilon_1}+  x_{\epsilon_1- \epsilon_3}\cdot x_{-(\epsilon_1- \epsilon_3)}\pm x_{\epsilon_2- \epsilon_3}\cdot x_{-(\epsilon_2- \epsilon_3)})$,\newline

 $A_{\epsilon_2 \pm \epsilon_j}=x_{-\epsilon_j}^{1 or 2}(C_{\epsilon_2\pm \epsilon_j}+ x_{\epsilon_2\pm \epsilon_j}\cdot x_{-(\epsilon_2\pm \epsilon_j)}\pm x_{\epsilon_1\pm \epsilon_j}\cdot x_{-(\epsilon_1\pm \epsilon_j)} ),$\newline

$A_{-\epsilon_2-\epsilon_j}= x_{-\epsilon_j}^3(C_{-\epsilon_2-\epsilon_j}+x_{\epsilon_2+ \epsilon_j}\cdot x_{-(\epsilon_2+\epsilon_j)} \pm x_{\epsilon_1+ \epsilon_j}\cdot x_{-(\epsilon_1+ \epsilon_j)} ),$\newline

$A_{\epsilon_j- \epsilon_2}= x_{\epsilon_j(}(C_{\epsilon_j- \epsilon_2}+ x_{\epsilon_2- \epsilon_j}\cdot x_{-(\epsilon_2- \epsilon_j)}\pm x_{\epsilon_1- \epsilon_j}\cdot x_{-(\epsilon_1- \epsilon_j)})$,\newline

$A_{\epsilon_1\pm \epsilon_j}= x_{\epsilon_3- \epsilon_j}^{1 or 2}( C_{\epsilon_1\pm \epsilon_j}+ x_{\epsilon_2\pm \epsilon_j}\cdot x_{-(\epsilon_2\pm \epsilon_j)} \pm x_{-(\epsilon_1 \pm \epsilon_j )})$,\newline

$A_{-\epsilon_1- \epsilon_j}= x_{\epsilon_1- \epsilon_j}^{2 or 3}( C_{-\epsilon_1- \epsilon_j}+ x_{\epsilon_1+ \epsilon_j}\cdot x_{-\epsilon_1- \epsilon_j}\pm x_{\epsilon_2+ \epsilon_j}\cdot  x_{-\epsilon_2- \epsilon_j})$,\newline

$A_{\epsilon_j- \epsilon_1}= x_{\epsilon_j- \epsilon_3}(C_{\epsilon_j- \epsilon_1}+ x_{\epsilon_1- \epsilon_j}\cdot x_{-(\epsilon_1- \epsilon_j)}\pm x_{\epsilon_2- \epsilon_j}\cdot x_{-(\epsilon_2- \epsilon_j)}), $\newline

and $A_{\alpha}= x_{\alpha}^2$ or $x_{\alpha}^3 $or $x_{\alpha}^4$ \newline

suitably for other roots $\alpha$ so that  they all have distinct and different weights with respect to $(ad H)$- module $u(L)$, \newline

where $H$ is the cartan subalgebra consisting of all linear combination of $h_{\alpha}$ for any root $\alpha$. Here $C_{\alpha}\in F$ are chosen so that   parentheses are invertible and signs are chosen so that they commute with $x_{\alpha_j}= x_{\epsilon_1- \epsilon_2}$\newline

(II) Next we assume $\alpha_j$  is  a short root. Without loss of generality we may put $\alpha_j= \epsilon_1$.
Likewise as in (I), we suggest  a  basis of  $u(L)/\frak M_{\chi}$ like the following.\newline

$\frak B:= \{(B_1+ A_{\epsilon_1})^{i_1}\otimes(B_2+ A_{-\epsilon_1})^{i_2}\otimes (B_3+ A_{\epsilon_1- \epsilon_2})^{i_3}\otimes (B_4+ A_{-(\epsilon_1-\epsilon_2)})^{i_4}\otimes \cdots \otimes (B_{2l}+ A_{-(\epsilon_{l-1}- \epsilon_l)})^{i_{2l}}\otimes (B_{2l+1}+ A_{\epsilon_l})^{i_{2l+1}}
\otimes (B_{2l+2}+ A_{-\epsilon_l})^{i_{2l+2}}\otimes(\otimes _{j=2l+3}^{2m}(B_l+ A_{\alpha_j})^{i_j}) \}   $ \newline

for $0\leq i_j \leq p-1$,\newline

 where $A_{\epsilon_1}= g_{\epsilon_1}$,\newline

$A_{-\epsilon_1}= (C_{-\epsilon_1}+ w_{\epsilon_1}),$\newline

$A_{-\epsilon_1\pm \epsilon_2}= x_{\epsilon_3\pm \epsilon_2}(C_{-\epsilon_1\pm \epsilon_2}+ x_{-\epsilon_1\pm \epsilon_2}\cdot x_{-(-\epsilon_1\pm \epsilon_2)}\pm x_{\pm \epsilon_2}\cdot x_{-(\pm \epsilon_2)}\pm x_{\epsilon_1\pm \epsilon_2}\cdot x_{-(\epsilon_1\pm \epsilon_2)}),$\newline

$A_{\epsilon_1+ \epsilon_2}= x_{\epsilon_3- \epsilon_2 }^2 (C_{\epsilon_1+ \epsilon_2}+ x_{-\epsilon_1- \epsilon_2}\cdot x_{\epsilon_1+ \epsilon_2}\pm x{-\epsilon_2}\cdot x_{\epsilon_2}\pm x_{\epsilon_1-\epsilon_2}\cdot x{\epsilon_2- \epsilon_1}),$\newline

$ A_{-\epsilon_1\pm \epsilon_j}= x_{-\epsilon_2\pm \epsilon_j}(C_{-\epsilon_1\pm \epsilon_j}+ x_{-\epsilon_1\pm \epsilon_j}\cdot x_{-(-\epsilon_1\pm \epsilon_j)}    + x_{\epsilon_1\pm \epsilon_j}\cdot x_{-(\epsilon_1\pm \epsilon_j)}\pm x_{\pm \epsilon_j}\pm x_{-(\pm \epsilon_j)}),$  \newline

$A_{\epsilon_1+ \epsilon_j}= x_{-\epsilon_2- \epsilon_j}^2(C_{\epsilon_1+ \epsilon_j}+ x_{-\epsilon_j- \epsilon_1}\cdot x_{\epsilon_1+ \epsilon_j}\pm x{-\epsilon_j}\cdot x_{\epsilon_j}\pm x_{\epsilon_1- \epsilon_j}\cdot x{-(\epsilon_1- \epsilon_j)},$\newline

$A_{\pm \epsilon_2}= x_{\epsilon_3\pm \epsilon_2}^2( C_{\pm \epsilon_2}+ x_{\epsilon_2}\cdot x_{-\epsilon_2}\pm x_{\epsilon_1+ \epsilon_2}\cdot x_{-(\epsilon_1+ \epsilon_2}\pm x_{\epsilon_3- \epsilon_1}\cdot x_{-(\epsilon_3- \epsilon_1)})$,\newline

$ A_{\epsilon_j}= x_{\epsilon_2+ \epsilon_j}(C_{\epsilon_j}+ x_{\epsilon_1}\cdot x{-\epsilon_1}\pm x_{\epsilon_1+ \epsilon_j}\pm x_{-(\epsilon_1+ \epsilon_j)}\pm x_{\epsilon_j- \epsilon_1}\cdot x_{-(\epsilon_j- \epsilon_1)}),$\newline

$ A_{-\epsilon_j}= x_{\epsilon_2- \epsilon_j}(C_{-\epsilon_j}+ x_{-\epsilon_j}\cdot x_{\epsilon_j}\pm x_{\epsilon_1- \epsilon_j}\cdot x_{-(\epsilon_1- \epsilon_j)}\pm x{-\epsilon_j- \epsilon_1}\cdot x_{\epsilon_1+ \epsilon_j},  $\newline

and $A_{\alpha}= x_{\alpha}^2$ or $x_{\alpha}^3$  or $x_{\alpha}^4$ for the remaining roots  $\alpha$.\newline

We proved in the references [KY-4],[NWK-1,2]  that $\frak B$'s really form a basis of  the factor algebra $u(L)/\frak M_{\chi}$ under consideration.

However  we  can exhibit its proof  in detail differently from those as below.\newline

We built up the bases $\frak B$  in both cases of the factor algebra $u(L)/\frak M_{\chi}$ above\newline

with the sign chosen so that they commute with $x_{\alpha}$ and with $c_{\alpha}\in F$ 
chosen so that $A_{\epsilon_{2}-\epsilon_{1}}$ and parentheses are invertible.  For any other root $ \beta$ we put $A_{\beta}$= $x_{\beta}^2 $ or $x_{\beta}^3 $ if possible. \newline

Otherwise we may attach to these sorts the parentheses(        ) used for designating $A_{-\beta}$ so that  $A_\gamma  \forall \gamma \in \Phi$ may commute with $x_\alpha$.\newline

We shall prove that $\frak B$ is a basis in $u(L))$/$\mathfrak {M}_\chi$.\newline

 By virtue of P-B-W theorem, it is not difficult to see that $\frak B$ is evidently a linearly independent set  over $F$ in $u(L)$. Furthermore $\forall$ $\beta$ $ \in\Phi$, $A_{\beta}\notin\mathfrak {M}_\chi$(see detailed proof below).\newline

We shall prove that a nontrivial linearly dependent equation leads to absurdity.\newline

We assume first that there is a dependence equation which is of least degree with respect to $h_{\alpha}\in H$ and the number of whose highest degree terms is also least. \newline

In case it is conjugated by $g_{\alpha_j}$, then there arises a nontrivial dependence equation of lower degree than the given one, which contradicts our assumption.\newline

 Otherwise it reduces to one of the following forms:\newline

(i)$x_{\pm \epsilon_j}K+ K' \in \frak M_\chi, $\newline

(ii)$x_{\pm \epsilon_j \pm \epsilon_k}K+ K' \in \frak M_\chi,$\newline

(iii)$g_{\epsilon_1- \epsilon_2}K+ K' \in \frak M_\chi,$\newline

where $K,K'$ commute with $x_\alpha$ and $x_{-\alpha}$ modulo $\frak M_\chi$.\newline
Here we  assumed  first  that  $\alpha_j$ is a long root , and  we may put  $\alpha_j= \alpha$ for brevity.\newline

By making use of proofs of  propositions in section 4  or  in [KY-7], we may reduce (i)  and (ii) to the equation of the form\newline

 $x_{\epsilon_1- \epsilon_2}K+ K'\in \frak M_\chi,$\newline

where $K$ commute with $x_{\pm (\epsilon_1- \epsilon_2)}$ and $K'$ commute with $x_{\epsilon_1- \epsilon_2}$ modulo $\frak M_\chi$.\newline

We have $x_{\epsilon_1- \epsilon_2}^p K+ x_{\epsilon_1- \epsilon_2}^{p-1}K' \equiv 0 $, so we get $x_{\epsilon_1- \epsilon_2}^{p-1}K'\equiv 0$.\newline

Subtracting $x_{\epsilon_2- \epsilon_1}x_{\epsilon_1- \epsilon_2}K+ x_{\epsilon_2- \epsilon_1}K'\equiv 0$ from  this equation, we obtain $ -x_{\epsilon_2- \epsilon_1}x_{\epsilon_1- \epsilon_2}K+ g_\alpha K'\equiv 0$.
We should remember that $g_\alpha$ is invertible in $u(L)/\frak M_\chi$ by virtue of  [RS].\newline

By the way we use $w_\alpha := (h_\alpha+ 1)^2+ 4x_{-\alpha}x_\alpha \in $ the center of $u(\frak {sl}_2(F))$. Hence we have $-4^{-1}\{w_\alpha- (h_\alpha + 1)^2\}K+ g_\alpha K'\equiv 0$. So we obtain

$4^{-1}g_\alpha^{p-1}\{(h_\alpha+ 1)^2- w_\alpha\}+ cK'\equiv 0 \cdots (\ast)$ \newline

and from the start equation we get \newline

$cx_\alpha K+ cK'\equiv 0 \cdots (\ast \ast)$.\newline

Subtracting $(\ast \ast)$ from $(\ast)$, we get $4^{-1}g_\alpha^{-1}\{(h_\alpha+ 1)^2- w_\alpha\}K- cx_\alpha K \equiv 0$.
Multiplying this equation by $g_\alpha^{1-p}$ to the right, we have\newline

 $4^{-1}g_\alpha^{p-1}\{h_\alpha+1)^2- w_\alpha\}g_\alpha^{1-p}K- cx_\alpha g_\alpha^{1-p}K \equiv 4^{-1}g_\alpha^{p-1} \{(h_\alpha+ 1)^2- w_\alpha\}g_\alpha^{1-p}K+ x_\alpha x_{-\alpha}K \equiv 0$.\newline

Conjugation of the brace of this equation $(p-1)$- times by $g_\alpha$ gives rise to $4^{-1}\{(h_\alpha- 1)^2- w_\alpha\}K+ x_\alpha x_{-\alpha}K\equiv 0$.
Next mutiplying $x_{-\alpha}^{p-1}$ to the right of the last equation, we obtain \newline

$\{(h_\alpha -1)^2- w_\alpha\}K x_{-\alpha}^{p-1}\equiv 0$ modulo $\frak M_\chi$. \newline

Now we multiply $x_\alpha$  to the left of this equaion consecutively until it becomes of the form \newline

(a nonzero polynomial of degree $\geq 1$ with respect to $h_\alpha)K \newline
\equiv 0$ modulo $ \frak M_\chi$.\newline

If we make use of conjugation and subtraction consecutively, then we arrive at a contradiction $K\equiv 0$.\newline

Next for the case (iii),we change it to the form (iii)$'K+ g_\alpha^{-1}K'\in \frak M_\chi$.\newline

We thus have an equation\newline

 $x_{\epsilon_1- \epsilon_2}K+ x_{\epsilon_1- \epsilon_2}g_{\epsilon_1- \epsilon_2}^{-1}K'\equiv 0$ modulo $\frak M_\chi$.
According to the above argument, we are also  led to a contradiction $K\in \frak M_\chi$.\newline

For a short root $\alpha$ , the proof is very similar  to the one just above.

Hence any coset of $u(L)/\frak M_{\chi}$ must be of the form :( a linear combination of all elements of $\frak B+ \frak M_{\chi})$.\newline

Now we let  $v$ be the nonarchimedean valuation of $F$ and $w$ the additive valuation of rank 1 satisfying that $v(x)= p^{-w(x)}, \forall x\in F$.\newline

For any given $x= (x_1,x_2,\cdots,x^{p^{2m}})\in F^{p^{2m}}$, where $x_i$'s are coefficients of the above linear combination, we put $v(x):=sup(v(x_i))$ and 
we put $U_{\lambda}:= \{x\in A: \lambda \geq 0, v(x)\leq p^{-\lambda}\},  $
which is an ideal of $A$. \newline

We let $G_{U_{\lambda}}:= \{x\in F^{p^{2m}}: x_i\in U_{\lambda}, 1\leq p^{2m}\}= \{x\in F^{p^{2m}}: v(x) \leq p^{-\lambda}\}$. \newline

Next we consider    an  infinite tower of  $G_{U_{\lambda}}$'s such as $G_{U_1}\supset G_{U_2} \supset G_{U_3}\supset \cdots   $\newline

 and pore over the infinite family of sets as follows:
$S_1:= G_{U_1}- G_{U_2}, S_2:= G_{U_2}- G_{U_3},\cdots, S_n:= G_{U_n}- G_{U_{n+1}},\cdots. $\newline

It is obvious that $G_{U_1}= \cup _{i=1}^{\infty}S_i$.
We put here $S'_i := \{x\in S_i:v(x)= p^{-i}\}$ for $i\in \mathbb N $(the natural number system).\newline

Here at this juncture by the axiom of choice we may get  a choice function $\psi$ of the family  $\{S'_i: i\geq 1\} $

satisfying that $\psi (S'_i)\in S'_i$.\newline

We consider  a counting problem asking whether or not  a given particular element $x$ of $u(L)$ is contained in a coset of the form $\psi (S'_i)+ \frak M_{\chi}$
 and how many coefficients  are there in $e\in \frak M_{\chi}$ modulo $\sum_{i } u(L)(x_i^p- x_i^{[p]}) $ having their absolute values equal to $v(x)$.\newline

First thing in order to solve this problem we must check whether $x+ e$ with some $e\in \frak M_{\chi}$ becomes of the form $\psi(S'_i)$.\newline

In the mean time we may delete any term having a factor $x_i^p- x^{[p]}$ for some $1\leq i \leq n$ no matter where they may exist in $x$, $e$, or $\psi (S'_i)$ in terms of P-B-W basis.\newline

 Hence we let  bars in the expression, say $\bar {x}+ \bar {e} \in \overline{\psi (S'_i)}$ indicate the expression after such deletion in prederence to the original counting problem. \newline

We must show that  it requires at most  a polynomial time of the input size  to check whether or not $\bar {x}+ \bar {e} \in \overline {\psi(S'_i)}$. \newline

Next we may define $\forall x \in u(L)$, $v(\bar {x})$:= max $\{v(c_j):\bar {x}= \sum {c_j} \bar {b_j}\}$ with $c_j\in F$ for P-B-W basis $b_j$ of $u(L)$.  We consider first that if $\bar{x}\in \overline {\psi(S'_i)}+ \bar {e}$, then we should get necessarily $v(\bar \psi(S'_i)) \leq v(\bar {x})$.\newline

Let $c_i$ be coefficients of $\beta_i's$ of $\frak B$ for $1\leq i \leq p^{2m}$.
Further let $C_{j_1 j_2 \cdots j_n }$ be coefficients of the term $x_{j_1 j_2 \cdots j_n}$ of a linear combination $\bar {x}$ of elements of the form $x_{j_1j_2\cdots j_n}= b_1^{j_1} b_2^{j_2} \cdots b_n^{j_n}$,\newline

where $b_1,\cdots, b_n$ are basis elements of $L$ and $0 \leq j_k \leq p-1$. \newline

We can rearrange the element $s$,  $\sum \overline {c_i \beta_i}$,and $\bar {e}= \sum \overline {e_{j_1 j_2\cdots j_n}}$$ \overline {x_{j_1J_2 \cdots j_n} }$ with respect to P-B-W basis in order to get a system of linear equations of  the form\newline

 $C_{j_1 j_2 \cdots j_n}= \sum _{i=1}^{p^{2m}} C'_{ij_1\cdots j_n}\cdot C_{ij_1 \cdots j_n}+ \sum e'_{j_1 \cdots j_n}\cdot e_{j_1\cdots j_n}   $

with $C'_{ij_1\cdots j_n}, e'_{j_1\cdots j_n}$ integers in $ \{1,2,\cdots, p-1\}$.\newline

We should note here that such integers arise due to rearranging $\overline {\psi(S_i)}$ and  $\bar {e}$ according to P-B-W basis.\newline

 Next  the substitution of  $C_{ij_1\cdots j_n}$ by $\psi(S'_i)$ gives rise to  a  system of nonhomogeneous linear equations in  indeterminates $e_{j_1\cdots j_n}$.
\newline

Obviously this system has  an algorithm by the Cramer's formula.
We perceive that the input time is at most $p^n$ which is the number of elements of P-B-W basis.\newline

So as to put forth flow charts for this algorithm, we need only polynomial time  in view of the fact that the following system \newline

$(\ast)\left \{ \begin{array}
{cc}a_{11}x_1+ \cdots + a_{1p^n}x_{p^n}= C_{j_1\cdots j_n}
- \sum _{i=1}^{p^{2m}}C'_{ij_1\cdots j_n} \psi_{ij_1\cdots j_n}$\newline,

$ \\ \cdots   \cdots $     $ \cdots  \cdots $\newline

 $\\a_{p^n 1} x_1+ \cdots + a_{p^np^n}x_{p^n}= C_{p-1,\cdots, p-1}

-\sum_{i=1}^{p^{2m}}C'_{i(p-1,\cdots, p-1)} \cdot \psi_{i(p-1,\cdots ,p-1)} \end{array} \right.$\newline

is solvable if and only if $ v(x_i) \leq v(x),  \forall {i}\leq p^n$, \newline

where $a_{ij}$ represents $e'_{j_1\cdots j_n}$, and $x_1,\cdots ,x_{p^n}$ represents $e_{j_1\cdots j_n}$ respectively with order $j_1\cdots j_n\leq j'_1\cdots j'_n$ defined  if and only if $j_1\leq j'_1,\cdots ,j_{k-1} \leq j'_{k-1}, $

and $j_k< j'_k   $ for some $k\leq p^n.$ \newline

 So it requires at most a polynomial  time of the input size to check whether or not $ \bar {x} \in \overline {\psi (S'_i)+  e}$. \newline

Moreover it is evident that within a big fixed constant  time of the input size, the number of coefficients in $e\in \frak M_{\chi}$ modulo $\sum u(L)(x_i^p- x_i^{[p]})   $ suitable for our bid  can be easily detected by vrtue of Cramer's formula. \newline

Hence the given counting problem is evidently contained in $NP$- class. \newline

Here we should recollect that an algorithm of a problem in $NP$- class means  at most polynomial time logical processes regarding input sizes  for the individual  check of solutions of the problem, \newline

whereas  an algorithm of  a problem in $P$- class  means at most  polynomial time logical processes  regarding the input sizes  for the general solution  of the problem.
\newline

We must also note that  algorithm shoud be determined  only by input data with other  data  finite and   definitely  determined and  at most countable. \newline

We may take an example for this matter. Consider the problem asking if a positive integer is a prime number or not.\newline

It is already known that 
any specific integer  can be checkted  by a calculator in order to know  whether it is prime or not  within at most polynomial time  regarding the number of digits  representing the integer.\newline

So we know that  the problem is  contained in the $NP$- class.\newline

 However we must still compute to know whether or not the algorithm  for the  solution  is generally  determined  within at most polynomial time  regarding the input sizes consisting of  information only  about  any  given integer. \newline

It means that we have to determine whether or not  the problem is  in the $P$- class. But it is well known that  the probem is contained in the $P$- class, which was  proved by the Indian Institute of Technology.\newline

Finally, we insist that  our builtup problem cannot belong to the $P$- class. \newline

Suppose now that  the counting problem under consideration  is contained  in the $P$- class. If we consider the cardinality of $\hat {n}(x):= \# \{j_1\cdots j_n $ : $ v(e'_{j_1\cdots j_n} \cdot e_{j_1\cdot j_n})= v(x) \}$, then the cardinality must be determined  only by  the input data of $x$. \newline

However  by the former system $(\ast) $ of nonhomogeneous linear equations, we see that $\hat {n}(x)$ is  determined not only  by  the data $C_{j_1\cdot j_n}$ of $\bar {x}$ but also by the coefficients  $\psi_{ij_1\cdot j_n}$ which are random variables  in the field $F$ regardless of $x$ except for $v(x)$.\newline

 Note also that  the field $F$ is  an uncountable  set, so  we cannot  find  out  clearly any element  in $F$ in polynomial time of the input size.\newline

So if we identify the coefficients of  $\bar {x}$ with $\bar {x}$ with respect to P-B-W basis for brevity, then the coefficiets time becomes the input time $\bar {x}$ and the output time $f(\bar {x})$ is obviously not less than any exponential function time $R^{\bar {x}}$. In other words we must  have  $R^{\bar x}=  \mathcal O(f(\bar x))$, i.e., $f(\bar x)= \Omega (R^{\bar x})$  in  terms  of  definition 2.1.\newline

In fact it may take an infinite time  even  to find out  a particular choice function  out of all choice functions because the ground field $F$ is uncountable. \newline

Hence we meet with a contradiction. So  our given counting problem is not contained in the $P$- class. Since it is known that $P$-class $\subseteq NP$-class, it  turns out  that $P$-class is properly contained in $NP$-class.\newline

Hence we conclude that $P\neq NP$ after all.

\end{proof}

It  may or may not be  true  that the counting problem is equivalent to the the subset sum problem.\newline

 In other words the subset sum problem  may or may not reduce to the counting problem in polynomial time.\newline

 However it looks like  we  may conjecture  the following, the conjecture being compared to  the subset sum problem.\newline

 $\textbf{[conjecture]}$\newline
We conjecture that the counting problem under consideration could be an $NP$- complete problem. We would like to find out the heuristic reason for this conjecture somehow.\newline

We know that the subset sum problem requires at most the time $\binom{n}{0}+ \binom{n}{1}+ \binom{n}{2}+ \cdots+\binom{n}{n}=2^n$

for the input time $n$.\newline

We  don't know  for now whether or not  we can reduce the time required.
 If it is possible to say that the counting problem  requires generally at least $2^n$ for the input time $n$, then we might say that   the latter algorithm is absolutely harder than the former one. \newline

So this could be roughly  the  heuristic reason for the conjecture above.

\section{Questions and answers}

Before closing our paper, we naturally raise two more questions attached to $P\neq NP$.\newline

\textbf{[Question-1]}Does $NP$-class form a  distributive lattice with respect to the usual logical operations V (disjunction) and (conjunction)?\newline

\textbf{[Question-2]}Does $P$-class form a proper distributive sublattice of the $NP$-class? \newline

We have found earlier that  the answers to both questions are affirmative. Hence the rest part of this  section focuses on the proofs of  such facts. \newline

Firstthing definitions  related to these questions should be recalled. Suppose that 
a certain  collection of statements or  a certain collection of problems  is closed  under conjunction operation $\Lambda$ and disjunction operation V.\newline

Generally speaking, these collections are  not necessarily Boolean algebras. However these may be  sometimes distributive lattices.\newline

Now let $x$ and $y$ be elements of  an ordered set $S_o$. If there  exists  the supremum of  $\{x, y\}$, then we call it the $\textit {join}$ of  $a,b$ and we denote it by $x\cup y$.
\newline

Analogously in case the infimum of $\{x, y\}$  exists, then we call it the $\textit {meet}$ of  $x,y$ and  denote it by the symbol $x\cap y$.

\begin{definition}
We shall call  an ordered set  $S_o$ a  $\textit {lattice}$  in case  every pair of  its elements  has a join and  a meet.
If $S_o$ satisfies in addition any of  the following equivalent conditions, we shall call $S_o$  a distributive lattice: \newline

(i) $x\cup (y\cap z)= (x\cup y)\cap (x\cup z)$,\newline

(ii)$x\cap (y\cup z)= (x\cap y)\cup (x\cap z)$,\newline

(iii)$(x\cup y)\cap (y\cup z)\cap (z\cup x)= (x\cap y) \cup (y\cap z)\cup (z \cap x)$.
\end{definition}

\begin{prop}
Any Boolean algebra  becomes  a  distributive lattice.

\end{prop}

\begin{proof}

We assume that  $<B, +, \ast>$ is a  Boolean algebra. Next we define $inf\{a,b\}:= a\ast b$ and  $sup\{a,b\}:= a+ b$. \newline

Such definitions  are  well defined  so long as we define a partial order $a\leq ^{\otimes}$
 in case that  the following equivalent equations  are satisfied:\newline

(i)$a\ast b'=0,$\newline

(ii)$a+ b= b,$\newline

(iii)$a'+ b= 1,$\newline

(iv)$a\ast b= a$, \newline

where the notation $'$ denotes  the complement.
It is well known that these four equations are always satisfied  in any Boolean algebra.\newline

 Furthermore it is not difficult to  check  that $<B, +, \ast>$ satisfies the axioms of definition6.1  mentioned above.

\end{proof}

Now we pause for a while to think of a problem in decent way  in connection with a mathematical statement.\newline

 As we are well aware, a mathematical problem usually has  one of  the forms  such as $^\prime$$ ^\prime$ Prove some fact, or prove or disprove something, or compute something $\cdots$ etc.$^\prime$$^\prime$ \newline

However the important point is whether or not we can solve such a problem within a polynomial  time relative to  input size. Actually in practice we cannot dispense with time limit as is well known.\newline

So in this regard we change any problem $p\in NP$  into  the mathematical statement  $^\prime$$^\prime$ we can solve $p$ within a polynomial time.$^\prime$$^\prime$ \newline

No doubt we we may still put put this new statement as the same letter  $p$ incurring no harm to the structure of the class $NP$.\newline

In this situation we may thus discuss about  the structure  more explicitly.
Hence  in the next proposition  we shall show that $NP$-class becomes a distributive lattice in reality.

\begin{prop}
Under the  operations disjunction $V$ and  conjunction  $\Lambda$  just described above, we insist that  $NP$- class has the distributive lattice structure

\end{prop}

\begin{proof}
By virtue of  proposition6.1, and proposition6.2, we are well aware that the  algebra of  propositions constitutes  a Boolean algebra  and so a  distributive lattice.\newline

 However  it is clear that  the $NP$- class is closed under  these two operations  disjunction $V$ and   conjunction $\Lambda$
, so that  the $NP$- class  becomes a  distributive lattice as a  sublattice of  the algebra of  propositions  obviously.
\end{proof}

\begin{prop}
$P$-class also becomes  a  distributive lattice  under the  two operations given above.

\end{prop}

\begin{proof}
We know that  the $P$- class  forms a  proper subclass of  $NP$- class by virtue of  the proof of proposition5.7.\newline

Moreover  $P$-class also  becomes  stable with two operations conjunction $\Lambda$ and  disjunction $V$, so that $P$-class becomes a  proper sublattice of the $NP$- class obviously.
\end{proof}

$\mathbf{[Acknowledgement]}$\newline
We  give special many thanks to  all persons  who have helped  us  establish  such a  pretty long story  relating to  a Clay problem.

\

\

\bibliographystyle{amsalpha}

\end{document}